%% file: main.tex
\documentclass[11pt,a4paper]{article} 
\usepackage{enumitem}
\usepackage[margin=1in]{geometry} 

\usepackage{setspace}
\usepackage{amsthm,amsmath,amssymb}
\usepackage{amsfonts,mathrsfs}
\usepackage{caption,subcaption}
\usepackage{threeparttable}
\usepackage{multirow}
\usepackage{longtable}
\usepackage[dvipsnames]{xcolor}
\usepackage[linesnumbered,ruled,vlined]{algorithm2e}
\usepackage[margin=1in]{geometry}
\usepackage{xspace,makecell,tikz}
\usepackage{siunitx}
\usepackage{float}
\usepackage[sort,numbers]{natbib}
\usepackage{hyperref}
\usepackage[capitalise]{cleveref}
\usepackage{comment}
\usepackage{lmodern}
\hypersetup{
colorlinks=true,
linkcolor=black,
filecolor=magenta,
urlcolor=cyan,
citecolor=Green,
}
\tikzset{  
    thickline/.style={line width=2pt}, 
}  
\usetikzlibrary{positioning,arrows.meta}
\crefname{equation}{}{}

\newcommand{\eat}[1]{}

\newcommand{\floor}[1]{\left\lfloor #1 \right\rfloor}
\newcommand{\ceil}[1]{\left\lceil #1 \right\rceil}
\newcommand{\assigned}[1]{[ #1 ]}

\newcommand{\calX}{\mathcal{X}}

\newcommand{\calF}{\mathcal{F}}
\newcommand{\calV}{\mathcal{V}}

\newcommand{\calL}{\mathcal{L}}
\newcommand{\calU}{\mathcal{U}}

\newcommand{\calN}{\mathcal{N}}

\newcommand{\ww}{\boldsymbol{\widetilde{w}}}

\newcommand{\da}{\downarrow}

\newcommand{\boldu}{\boldsymbol{u}}
\newcommand{\boldv}{\boldsymbol{v}}
\newcommand{\boldw}{\boldsymbol{w}}

\newcommand{\boldt}{\boldsymbol{t}}
\newcommand{\boldo}{\boldsymbol{o}}
\newcommand{\boldd}{\boldsymbol{d}}

\newcommand{\bbone}{\mathbf{1}}

\newcommand{\norm}{f}

\newcommand{\minnorm}{\textsf{MinNorm}}
\newcommand{\LBO}{\textsf{LogBgt}}
\newcommand{\minnormintcov}{\textsf{MinNorm-IntCov}}
\newcommand{\LBOintcov}{\textsf{LogBgt-IntCov}}

\newcommand{\LBOtreecov}{\textsf{LogBgt-TreeCov}}
\newcommand{\minnormsetcov}{\textsf{MinNorm-SetCov}}
\newcommand{\LBOsetcov}{\textsf{LogBgt-SetCov}}
\newcommand{\minnormknapcov}{\textsf{MinNorm-KnapCov}}
\newcommand{\LBOknapcov}{\textsf{LogBgt-KnapCov}}
\newcommand{\minnormmatch}{\textsf{MinNorm-PerMat}}
\newcommand{\LBOmatch}{\textsf{LogBgt-PerMat}}
\newcommand{\minnormpath}{\textsf{MinNorm-Path}}
\newcommand{\LBOpath}{\textsf{LogBgt-Path}}
\newcommand{\minnormcut}{\textsf{MinNorm-Cut}}
\newcommand{\LBOcut}{\textsf{LogBgt-Cut}}

\newcommand{\hv}{\widehat{v}}

\newcommand{\myproblem}{\mathfrak{A}}
\newcommand{\Algo}{\mathsf{Alg}}

\newcommand{\topdash}[1]{\mbox{Top-$#1$}}

\newcommand{\topp}[1]{\textsc{Top}_{#1}}
\newcommand{\topl}[2]{\textsc{Top}_{#1}\left(#2\right)}
\newcommand{\ordd}[1]{\textsc{Ord}_{#1}}
\newcommand{\ordered}[2]{\textsc{Ord}_{#1}\left(#2\right)}

\newcommand{\trU}{\operatorname{\calU^{tr}}}
\newcommand{\trS}[1]{%
  \ifx\relax#1\relax
    S^{\operatorname{tr}} 
  \else
    S^{\operatorname{tr}}_{#1} 
  \fi
}

\newcommand{\trF}{\operatorname{\calF^{tr}}}
\newcommand{\treta}{\operatorname{\eta^{tr}}}
\newcommand{\intU}{\operatorname{\calU^{int}}}
\newcommand{\intS}[1]{%
  \ifx\relax#1\relax
    S^{\operatorname{int}} 
  \else
    S^{\operatorname{int}}_{#1} 
  \fi
}
\newcommand{\intD}{D^{\operatorname{int}}}
\newcommand{\intF}{\operatorname{\calF^{int}}}
\newcommand{\inteta}{\operatorname{\eta^{int}}}
\newcommand{\ldU}{\operatorname{\calU^{ld}}}
\newcommand{\ldS}[1]{%
  \ifx\relax#1\relax
    S^{\operatorname{ld}} 
  \else
    S^{\operatorname{ld}}_{#1} 
  \fi
}
\newcommand{\ldD}{D^{\operatorname{ld}}}

\newcommand{\ldeta}{\operatorname{\eta^{ld}}}
\newcommand{\lamD}{D^{\operatorname{lam}}}

\newcommand{\trM}{\operatorname{M^{tr}}}

\newcommand{\lamU}{\operatorname{\calU^{lam}}}
\newcommand{\lamS}[1]{%
  \ifx\relax#1\relax
    S^{\operatorname{lam}}
  \else
    S^{\operatorname{lam}}_{#1} 
  \fi
}

\newcommand{\lameta}{\operatorname{\eta^{lam}}}

\newcommand{\ldR}{R^{\operatorname{ld}}}
\newcommand{\lamR}{R^{\operatorname{lam}}}

\newcommand{\pathU}{\operatorname{\calU^{path}}}
\newcommand{\pathS}[1]{%
  \ifx\relax#1\relax
    S^{\operatorname{path}} 
  \else
    S^{\operatorname{path}}_{#1} 
  \fi
}

\newcommand{\pathF}{\operatorname{\calF^{path}}}
\newcommand{\patheta}{\operatorname{\eta^{path}}}

\newcommand{\pathG}{G^{\operatorname{path}}}
\newcommand{\pathV}{V^{\operatorname{path}}}
\newcommand{\pmU}{\operatorname{\calU^{pm}}}
\newcommand{\pmS}[1]{%
  \ifx\relax#1\relax
    S^{\operatorname{pm}} 
  \else
    S^{\operatorname{pm}}_{#1} 
  \fi
}

\newcommand{\pmF}{\operatorname{\calF^{pm}}}
\newcommand{\pmeta}{\operatorname{\eta^{pm}}}
\newcommand{\pmG}{G^{\operatorname{pm}}}

\newcommand{\cutU}{\operatorname{\calU^{cut}}}
\newcommand{\cutS}[1]{%
  \ifx\relax#1\relax
    S^{\operatorname{cut}} 
  \else
    S^{\operatorname{cut}}_{#1} 
  \fi
}

\newcommand{\cutF}{\operatorname{\calF^{cut}}}
\newcommand{\cuteta}{\operatorname{\eta^{cut}}}
\newcommand{\cutG}{G^{\operatorname{cut}}}
\newcommand{\cutV}{V^{\operatorname{cut}}}
\newcommand{\Des}{\operatorname{Des}}
\newcommand{\Ch}{\operatorname{Ch}}
\newcommand{\Anc}{\operatorname{Anc}}
\newcommand{\Leaf}{\operatorname{Leaf}}
\newcommand{\Lev}{\operatorname{Id}}
\newcommand{\Par}{\operatorname{Par}}

\newcommand{\E}{\mathbb{E}}
\newcommand{\Z}{\mathbb{Z}}
\newcommand{\R}{\mathbb{R}}
\newcommand{\Rpos}{\mathbb{R}_+}

\newcommand{\POS}{\textsc{POS}}

\newcommand{\poly}{\mathsf{poly}}

\newtheorem{theorem}{Theorem}[section]
\newtheorem{lemma}[theorem]{Lemma}

\newtheorem{claim}[theorem]{Claim}

\theoremstyle{definition}
\newtheorem{definition}{Definition}[section]

\newcommand{\jian}[1]{{\it \color{red}{Jian: #1}}}
\newcommand{\ckw}[1]{}
\newcommand{\zyr}[1]{{\it \color{blue}{zyr: #1}}}
\newenvironment{proofofthm}[1] 
  {\vspace{0.2cm}
  \par
  \noindent 
  {\em Proof of Theorem #1.}\mbox{}} 
  {\hfill$\square$} 
\newenvironment{proofoflem}[1] 
  {\vspace{0.2cm}
  \par
  \noindent 
  {\em Proof of Lemma #1.}\mbox{}} 
  {\hfill$\square$} 
\newcommand{\keywords}[1]{%
  \vspace{1em}%
  \noindent\textbf{Keywords:} #1
}

\title{New Results on a General Class of Minimum Norm Optimization Problems}
\author{
Kuowen Chen\thanks{Tsinghua University, \texttt{ckw22@mails.tsinghua.edu.cn}}
\and
Jian Li\thanks{Tsinghua University, \texttt{lapordge@gmail.com}}
\and
Yuval Rabani\thanks{The Hebrew University of Jerusalem, \texttt{yrabani@cs.huji.ac.il}}
\and
Yiran Zhang\thanks{Tsinghua University, \texttt{zhangyir22@mails.tsinghua.edu.cn}}
}
\date{}

\begin{document}

\maketitle
\begin{abstract}
We study the general norm optimization for combinatorial problems, 
initiated by Chakrabarty and Swamy (STOC 2019).
We propose a general formulation
that captures a large class of combinatorial structures: 
we are given a set $\mathcal{U}$ of $n$ weighted elements
and a family of {\em feasible} subsets $\mathcal{F}$.
Each subset $S\in \mathcal{F}$ is called a feasible solution/set of the problem.
We denote the {\em value vector} by $\boldsymbol{v}=\{\boldsymbol{v}_i\}_{i\in [n]}$,
where $\boldsymbol{v}_i\geq 0$ is the value of element $i$.
For any subset $S\subseteq \mathcal{U}$, we use  
$\boldsymbol{v}[S]$ to denote the $n$-dimensional vector 
$\{v_e\cdot \mathbf{1}[e\in S]\}_{e\in \mathcal{U}}$ (i.e., we zero out
all entries that are not in $S$).
Let $f: \mathbb{R}^n\rightarrow\mathbb{R}_+$ be a symmetric monotone norm function.
Our goal is to minimize the norm objective
$f(\boldsymbol{v}[S])$ 
over feasible subset $S\in \mathcal{F}$.
The problem significantly generalizes the corresponding min-sum and min-max problems.

We present a general equivalent reduction of the norm minimization problem to 
a multi-criteria optimization problem with logarithmic budget constraints, up to 
a constant approximation factor. Leveraging this reduction, we obtain
constant factor approximation algorithms for the norm minimization versions of
several covering problems,
such as interval cover, multi-dimensional knapsack cover, and logarithmic factor approximation
for set cover.
We also study the norm minimization versions for perfect matching, $s$-$t$ path and $s$-$t$ cut.
We show the natural linear programming relaxations for these problems 
have a large integrality gap. To complement the negative result, we show that, for perfect matching, it is possible to 
obtain a bi-criteria result: for any constant $\epsilon,\delta>0$, we can find in polynomial time a nearly perfect matching
(i.e., a matching that matches at least $1-\epsilon$ proportion of vertices)
and its cost is at most $(8+\delta)$ times of the optimum for perfect matching.
Moreover, we establish the existence of a polynomial-time $O(\log\log n)$-approximation algorithm for the norm minimization variant of the $s$-$t$ path problem. Specifically, our algorithm achieves an $\alpha$-approximation with a time complexity of $n^{O(\log\log n / \alpha)}$, where $9 \leq \alpha \leq \log\log n$.
\end{abstract}


\input{intro}

\input{prel}


\input{reduction}

\input{knapsackcover}
\input{intervalcover}

\input{intgap}
\input{shortestpath}

\input{matching}
\input{conclusion}
\bibliography{reference}

\input{appendix}

\input{appendix-6}

\input{appendix-11-setcover}

\end{document}

%% file: intro.tex
\section{Introduction}
\label{sec:intro}

In many optimization problems, a feasible solution typically induces a multi-dimensional value vector (e.g., by the subset of elements of the solution), 
and the objective of the optimization problem is to minimize either the total sum (i.e., $\ell_1$ norm) or the maximum (i.e., $\ell_\infty$ norm) of the vector entry.
For example, in the minimum perfect matching problem, the solution is a subset of edges and the induced value vector is the weight vector of the matching (i.e., each entry of the vector is the weight of edge if the edge is in the matching and 0 for a non-matching edge) and we would like to minimize the total sum. 
Many of such problems are fundamental in combinatorial optimization but require different algorithms for their min-sum and min-max variants
(and other possible variants).
Recently there have been a rise of interests in developing algorithms for more general objectives, such as $\ell_p$ norms \cite{azar2005convex,golovin2008all}, top-$\ell$ norms \cite{maalouly2022exact}, ordered norms \cite{byrka2018constant,chakrabarty2018interpolating} and more general norms
\cite{chakrabarty2019approximation,chakrabarty2019simpler,ibrahimpur2020approximation,deng2022generalized,abbasi2023parameterized,kesselheim2024supermodular}, as interpolation or generalization of min-sum and min-max objectives. The algorithmic study of such generalizations helps unify, interpolate and generalize classic objectives and algorithmic
techniques. 

The study of approximation algorithm for general norm minimization problems is initiated by Chakrabarty and Swamy \cite{chakrabarty2019approximation}. They studied two fundamental problems, load balancing and $k$-clustering, and provided constant factor approximation algorithm for these problems. For load balancing, the induced value vector is the vector of machine loads and for 
$k$-clustering the vector is the vector of service costs.
Subsequently, the norm minimization has been studied for a variety of other combinatorial problem
such as general machine scheduling problem \cite{deng2022generalized}, stochastic optimization problems \cite{ibrahimpur2020approximation}, 
online algorithms 
\cite{patton2023submodular},
parameterized algorithms \cite{abbasi2023parameterized} etc.
In this paper, we study the norm optimization problem for a general set of combinatorial problems.
In our problem, a feasible set is a subset of elements
and the multi-dimensional value vector is induced by the subset of elements of the solution.
Our problem is defined formally as follows:

\begin{definition} (The Norm Minimization Problem (\minnorm))
We are given a set $\calU=[n]$ of $n$ weighted elements
and a family of {\em feasible} subsets $\calF$.
Each subset $S\in \calF$ is called a feasible solution/set of the problem.
We denote the {\em value vector} by $\boldv =\{\boldv_i\}_{i\in [n]}$,
where $\boldv_i\geq 0$ is the value of element $i$.
We say a subset $S\subseteq \calU$ {\em feasible} if
$S\in \calF$.
For any subset $S\subseteq \calU$, we use  
$\boldv[S]$ to denote the $n$-dimensional vector 
$\{v_e\cdot \bbone[e\in S]\}_{e\in \calU}$ (i.e., we zero out
all entries that are not in $S$), and we call $\boldv[S]$
{\em the value vector induced by $S$}.
Let $\norm: \R^n\rightarrow\Rpos$ be a symmetric monotone norm function.
Given the norm function $\norm(\cdot)$, our goal is to find a feasible solution in $\calF$ such that the norm of the value vector induced by the solution is minimized, i.e., we aim to solve the following optimization problem
$$
\text{\minnorm:}\quad\qquad
\text{minimize } \quad f(\boldv\assigned{S})  \quad \quad 
\text{ subject to } \quad \quad  S\in \calF. \qquad\qquad
$$
\end{definition}

Note that the case $f(\boldv\assigned{S})= \sum_{e\in S} v_e$ is the most studied min-sum objective and we call the corresponding problem the {\em original optimization problem}. Other interesting norms include $\ell_p$ norms, \topdash{\ell} norms (the sum of top-$\ell$ entries), ordered norms (see its definition in Section~\ref{section:preliminary}).
Note that our general framework covers the $k$-clustering studied in \cite{chakrabarty2019approximation}:
in the $k$-clustering problem, the universe $\calU$
is the set of edges and each feasible solution in $\calF$ is a subset of edges that corresponds to a $k$-clustering.
The load balancing problem does not directly fit into our framework, since
one needs to first aggregate the processing times to machine loads, then apply the norm.

Before stating our results, we briefly mention some results that are either known or very simple to derive. 
\begin{enumerate}
    \item (Matroid) Suppose the feasible set $\calF$ is a matroid
    and a feasible solution is a basis of this matroid.
    In fact, the greedy solution (i.e., the optimal min-sum solution)
    is the optimal solution for any monotone symmetric norm.
    This is a folklore result and can be easily seen as follows:
    First, it is easy to establish
    the following observation, using the exchange property of matroid: 
    We use $\topp{\ell}(S)$ to denote the sum of largest $\ell$ elements of 
    $S$. For any $\ell\in \mathbb{Z}_{\geq 1}$ and 
    any basis $S\in \calF$,
    $\topp{\ell}(S_{\text{greedy}})\leq \topp{\ell}(S)$
    where $S_{\text{greedy}}$ is the basis obtained by the greedy algorithm.
    Then using the majorization lemma by Hardy, Littlewood and P\`olya
    (Lemma~\ref{lemma:majorization}),
    we can conclude $S_{\text{greedy}}$ is optimal for any monotone symmetric norm.
    \item (Vertex Cover)
    We first relax the problem to the following convex program:
    $$\text{min.}\,\,\,\, f(v_1x_1,\ldots, v_nx_n) \quad\text{s.t. }\,\,\,\
    x_i+x_j\geq 1 \text{ for any }(i,j)\in E.$$
    The objective is convex since $f$ is norm (in particular the triangle inequality of norm). Then, we solve the convex program
    and round all $x_i\geq 1/2$ to $1$ and others to $0$. It is easy to see this gives a 2-approximation (using the property $f(\alpha x)=\alpha f(x)$ for $\alpha\geq 0$).
    \item (Set Cover) The norm-minimization set cover problem 
    is a special case of the generalized load balancing problem introduced in
    \cite{deng2022generalized}. Here is the reduction:
    each element corresponds to a job and each subset to a machine;
    if element $i$ is in set $S_j$, the processing time $p_{ij}=1$,
    otherwise $p_{ij}=\infty$; the inner norm of each machine is the max
    norm (i.e., $\ell_\infty$) and the outer norm is $f(\cdot)$.
    Hence, this implies an $O(\log n)$-approximation for norm-minimization set cover problem using the general result in \cite{deng2022generalized}. The algorithm in \cite{deng2022generalized} is based on a fairly involved configuration LP. In Appendix~\ref{sec:set-cover}, we provide a much simpler randomized rounding algorithm that is also 
    an $O(\log n)$-approximation. Note this is optimal up to a constant factor given the approximation hardness of set cover \cite{feige1998threshold,dinur2014analytical}.
    \item ($\topp{\ell}$ and Ordered Norms)
    If the min-sum problem can be solved or approximated efficiently, one can also solve or approximate the corresponding $\topp{\ell}$ and ordered norm optimization problems. This mostly follows from known techniques in \cite{byrka2018constant,chakrabarty2019approximation,maalouly2022exact}. 
    For completeness, the general and formal statements are provided in Appendix~\ref{sec:ordered-norm}.
\end{enumerate}

\subsubsection*{Our Contributions}

Our technical contribution can be summarized as follows:
\begin{enumerate}
    \item (Theorem~\ref{thm:equivalence}) We present a general reduction of the norm minimization problem to a multi-criteria optimization problem with logarithmic budget constraints, up to a constant approximation factor.
    This immediately implies an
    $O(\alpha \log n)$-approximation for the \minnorm\ problem 
    if there is a poly-time $\alpha$-approximation for the 
    corresponding weight minimization problem (See Theorem~\ref{thm:logapprox}).
    
    \item 
    Leveraging the reduction in Theorem~\ref{thm:equivalence}, 
    we obtain constant factor approximation algorithms for the norm minimization versions of
    several covering problems,
    such as interval covering (Theorem~\ref{thm:inter-cover}), 
    multi-dimensional knapsack cover (Theorem~\ref{thm:knap-1} and Theorem~\ref{thm:knap-2}).
    These algorithms are based on rounding the natural linear programming relaxation 
    of the multi-criteria optimization problem, possibly with a careful enumeration
    of partial solutions.
    For set cover, we obtain a simple randomized approximation algorithm 
    with approximation factor $O(\log n)$ (Theorem~\ref{thm:setcover}), which is much simpler than the general algorithm in \cite{deng2022generalized}.

    

     


    \item 
    We also study the norm minimization versions for perfect matching, $s$-$t$ path and $s$-$t$ cut. We show the natural linear programming relaxations for these problems 
    have a large integrality gap (Theorem~\ref{thm:intgap-1} and Theorem~\ref{thm:intgap-2}). This indicates that
    it may be difficult to achieve constant approximation factors for these problems.
    
    \item 
    To complement the above negative result, we show that, for perfect matching, it is possible to obtain a bi-criteria approximation: for any constant $\epsilon>0$, we can find a nearly perfect matching that
    matches at least $1-\epsilon$ proportion of vertices
    and the norm of this solution is at most $(8+\delta)$ times of the optimum for perfect matching where $\delta$ is any positive real constant
    (Theorem~\ref{matching-i5}). 
    
    \item We present an approximate dynamic programming approach that yields a $\alpha$-approximation $n^{O(\log\log n/\alpha)}$-time algorithm for the min-norm $s$-$t$ path problem for $9\le \alpha\le \log\log n$ (\cref{thm:algo-path}), demonstrating an alternative technique for solving norm minimization problems beyond LP rounding.
        
\end{enumerate}

\section{Related Work}

\vspace{0.2cm}
{\bf Top-$\ell$ and Ordered Optimization:}
As a special case of general norm optimization, ordered optimization for combinatorial optimization problems have received significant attention in the recent years. 
In fact, an ordered norm can be written as a conical combination of top-$\ell$ norms
(see Claim~\ref{lemma:ordered:conic:decomp}).
The ordered k-median problem was first studied by Byrka et al. \cite{byrka2018constant} and Aouad and Segev~\cite{aouad2019ordered}.
Byrka et al. \cite{byrka2018constant} obtained the first constant factor approximation algorithm (the factor is $38+\epsilon$).
Independently, Chakrabarty and Swamy~\cite{chakrabarty2018interpolating} obtained 
an algorithm with approximation factor $18$ for the top-$\ell$ norm), which can be combined with the enumeration procedure of Aouad and Segev~\cite{aouad2019ordered} to get the same factor for the general ordered $k$-median.
The current best known approximation is $5$, by Chakrabarty and Swamy \cite{chakrabarty2019approximation}.
Deng and Zhang \cite{deng2020ordered} studied ordered $k$-median with outliers and obtained a constant factor approximation algorithm.
Maalouly and Wulf \cite{maalouly2022exact} studied the top-$\ell$ norm optimization for 
the matching problem and obtained an polynomial time exact algorithm
(see also Theorem~\ref{thm:topl-1} in Appendix~\ref{sec:topl-norm-opt}).
Braverman et al. studied coreset construction for ordered clustering problems \cite{braverman2019coresets} which was motivated by applications in machine learning.
Batra et al. \cite{batra2023tight} studied the ordered min-sum vertex cover problem and obtained the first poly-time approximation approximation with approximation factor $2+\epsilon$.

\vspace{0.2cm}
\noindent
{\bf General Symmetric Norm Optimization:}
Chakrabarty and Swamy \cite{chakrabarty2019approximation} first studied general monotone symmetric norm objectives for clustering and unrelated machine load balancing and obtained constant factor approximation algorithms, substantially generalizing the results for $k$-Median and $k$-Center and makespan minimization for unrelated machine scheduling. 
In a subsequent paper \cite{chakrabarty2019simpler}, they obtained a simpler algorithm for 
load balancing that achieves an approximation factor of $2+\epsilon$.
Abbasi et al. \cite{abbasi2023parameterized} studied the parametrized
algorithms for the general norm clustering problems
and provided the first EPAS (efficient parameterized approximation scheme).
Deng et al. \cite{deng2022generalized} introduced the generalized load balancing problem, 
which further generalizes the problem studied by \cite{chakrabarty2019simpler}. 
In the generalized load balancing problem, the load of a machine $i$ is a symmetric, monotone (inner) norm of the vector of processing times of jobs assigned to $i$. 
The generalized makespan is another (outer) norm aggregating the loads.
The goal is to find an assignment of jobs to minimize the generalized makespan.
They obtained a logarithmic factor approximation, which is optimal up to constant factor
since the problem generalizes the set cover problem.
For the special case where the inner norms are top-$k$ norms,
Ayyadevara et al. \cite{ayyadevara2023minimizing} showed the natural 
configuration LP has a $\Omega(\log^{1/2}n)$ integrality gap.

\vspace{0.2cm}
\noindent
{\bf Submodular/Supermodular Optimization:}
Optimizing submodular/supermodular function under various combinatorial constraints
is another important class of optimization problems with general objectives and  
has been studied extensively in the literature. See e.g., \cite{calinescu2007maximizing,lee2010submodular,chekuri2011submodular,buchbinder2014submodular} and the survey \cite{krause2014submodular}.
However, note that results for submodular functions
does not imply results for general symmetric monotone norms, since a general symmetric monotone norm is not necessarily a submodular function (see e.g., \cite{deng2022generalized}).

Patton et al. \cite{patton2023submodular} studied submodular norm objectives (i.e., norms that also satisfies continuous submodular property). They showed that it can approximate well-known classes of norms, such as $\ell_p$ norms, ordered norms, and symmetric norms and applied it to a variety of problems such as online facility location, stochastic probing, and generalized load balancing.
Recently, Kesselheim et al. \cite{kesselheim2024supermodular} introduced the notion of p-supermodular norm and showed that every symmetric norm can be approximated by a p-supermodular norm. Leveraging the result, they obtain new algorithms 
online load-balancing and bandits with knapsacks, stochastic probing and so on.

\vspace{0.2cm}
\noindent
{\bf Multi-budgeted Optimization:}
There is a body of literature in the problem of optimizing linear or 
submodular objectives over a combinatorial structure with
additional budget constraints 
(see e.g., \cite{ravi1996constrained,camerini1992random,grandoni2009iterative,grandoni2010approximation,chekuri2011multi,berger2011budgeted,grandoni2014new}).
For a single budget constraint,
randomized or deterministic PTASes
have been developed for various combinatorial optimization problems
(e.g. spanning trees with a linear budget \cite{ravi1996constrained}). 
Assuming that a pseudopolynomial time algorithm for the exact version of the problems exists, Grandoni and Zenklusen showed that one can obtain a
PTAS for the corresponding problem with any fixed number of linear budgets \cite{grandoni2010approximation}.
More powerful techniques such as randomized dependent rounding and iterative rounding
have been developed to handle more general submodular objectives and/or other combinatorial structures such as matroid or intersection of matroid
(e.g., \cite{grandoni2009iterative,grandoni2010approximation,chekuri2011multi,grandoni2014new}).
Iterative rounding technique \cite{grandoni2009iterative,linhares2020approximate}
has been used in general norm minimization problems \cite{chakrabarty2019approximation,chakrabarty2019simpler}. Our algorithms for matching
(Section~\ref{sec:matching}) and knapsack cover (Section~\ref{sec:knapsackcover}) also adopt the technique.

%% file: prel.tex
\section{Preliminaries}
\label{section:preliminary}

Throughout this paper, for vector $\boldv\in\Rpos^n$, define $\boldv^\da$ as the non-increasingly sorted version of $\boldv$, and $\boldv[S]=\{\boldv_j\cdot\bbone[j\in S]\}_{j\in [n]}$ for any $S\subseteq [n]$.
Let $\topp{k}:\R^n\rightarrow\Rpos$ be the top-$k$ norm that returns the sum of the $k$ largest \emph{absolute values} of entries in any vector, $k\leq |n|$.
Denote $[n]$ as the set of positive integers no larger than $n\in\Z$, and $a^+=\max\{a,0\},\,a\in\R$.

We say function $\norm:\R^n\rightarrow\Rpos$ is a {\em norm} if: 
(i) $\norm(\boldv)=0$ if and only if $\boldv=0$, 
(ii) $\norm(\boldu+\boldv)\leq\norm(\boldu)+\norm(\boldv)$ for all $\boldu,\boldv\in\R^n$, 
(iii) $\norm(\theta\boldv)=|\theta|\norm(\boldv)$ for all $\boldv\in\R^n,\theta\in\R$. 
A norm $\norm$ is {\em monotone} if $\norm(\boldv)\leq\norm(\boldu)$ for all $0\leq\boldv\leq\boldu$, and {\em symmetric} if $\norm(\boldv)=\norm(\boldv')$ for any permutation $\boldv'$ of $\boldv$. 
We are also interested in the following special monotone symmetric norms.

\vspace{0.3cm}
\noindent 
{\bf \topdash{\ell} norms.}
Let $\ell\in[n]$.
A function is a \topdash{\ell} norm, denoted by $\topp{\ell}:\R^n\rightarrow\Rpos$, if for each input vector $\boldv\in\R^n$ it returns the sum of the largest $\ell$ \emph{absolute values} of entries in $\boldv$.
For non-negative vectors, it simply returns the sum of the largest $\ell$ entries.
We notice that by letting $\ell\in\{1,n\}$, $\topp{\ell}$ recovers the $\calL_\infty$ and $\calL_1$ norms, respectively, thus it generalizes the latter two.

\vspace{0.3cm}
\noindent 
{\bf Ordered norms.}
Let $\boldv\in\Rpos^{n}$ be a non-increasing non-negative vector.
For each vector $\boldv\in\R^{\calX}$, let $\boldv^\da\in\R^{|\calX|}$ denote its non-increasingly sorted version and define $|\boldv|=\{|\boldv_i|:i\in\calX\}\in\Rpos^{\calX}$. 
A function is a $\boldw$-ordered norm (or simply an ordered norm), denoted by 
$\ordd{\boldw}:\R^{\calX}\rightarrow\Rpos$, if for each input vector $\boldv\in\R^{\calX}$ it returns the inner product of $\boldw$ and $|\boldv|^\da$; we obtain $\ordered{\boldw}{\boldv}=\boldw^\top\boldv^\da$ whenever $\boldv\in\Rpos^{\calX}$.
It is easy to see that, by having $\boldv$ as a vector of $\ell$ 1s followed by $(|\calX|-\ell)$ 0s, $\ordd{\boldw}$ recovers $\topp{\ell}$.
On the other hand, it is known that each ordered norm can be written as a conical combination of \topdash{\ell} norms, as in the following claim.

\begin{claim}
\label{lemma:ordered:conic:decomp}
(See, e.g., \cite{chakrabarty2019approximation}).
For each $\boldv\in\Rpos^{\calX}$ and another non-increasing vector $\boldw\in\Rpos^{|\calX|}$, one has
\[
    \ordered{\boldw}{\boldv}=\sum_{\ell=1}^{|\calX|}(\boldw_{\ell}-\boldw_{\ell+1})\topl{\ell}{\boldv},
\]
where we define $\boldv_{|\calX|+1}=0$.
\end{claim}


The following lemma is due to Hardy, Littlewood and P\`olya.~\cite{hardy1934inequalities}.

\begin{lemma}\label{lemma:majorization}
(\cite{hardy1934inequalities}).
	If $\boldv,\boldu\in\Rpos^{\calX}$ and $\alpha\geq0$ satisfy $\topl{\ell}
{\boldv}\leq\alpha\cdot\topl{\ell}{\boldu}$ for each $\ell\in[|\calX|]$, one has 
$f(\boldv)\leq\alpha\cdot f(\boldu)$ for any symmetric monotone norm 
$f:\R^{\calX}\rightarrow\Rpos$.
\end{lemma}

The following results can be found in standard combinatorial optimization textbooks, e.g., Schrijver's book \cite{schrijver2003combinatorial}.

\begin{lemma}
\label{theorem:laminar:intersection}
Let $\calN_1,\calN_2$ be two laminar families on a common ground set $\calX$ and $A\in\{0,1\}^{(\calN_1\cup\calN_2)\times\calX}$ be the incidence matrix of $\calN_1\cup\calN_2$.
Then $A$ is totally unimodular.
\end{lemma}

\begin{lemma}
\label{theorem:TUM}
Let $A$ be a totally unimodular $m\times n$ matrix and $b\in\Z^m$, then the linear program $P=\{Ax\leq b\}$ has integral vertex solutions.
\end{lemma}

Here are some well-known results from linear programming.

\begin{lemma}
\label{lem:extreme:points}
(See, e.g., Chapter 2 of \cite{10.5555/548834}). Consider a linear program with the polyhedron 
$$
P = \{x \in \mathbb{R} ^n : A_1x = b_1, \,\,A_2x \leq b_2, \,\, A_3x \geq b_3\},
$$ 
where $A_1 \in \mathbb{R}^{m_1 \times n}, A_2 \in \mathbb{R}^{m_2 \times n}, A_3 \in \mathbb{R}^{m_3 \times n}$ are matrices and $b_1 \in \mathbb{R}^{m_1}, b_2 \in \mathbb{R}^{m_2}, b_3 \in \mathbb{R}^{m_3}$ are vectors. 
A point $x^* \in \mathbb{R}^n$ is an extreme point (or a basic feasible solution) of the linear program if and only if:
(1) $x^* \in P$, and
(2) there exists a set of $n$ linearly independent rows of the matrix $A = \begin{pmatrix} A_1 \\ A_2 \\ A_3 \end{pmatrix}$, which can be assembled into a submatrix $A_0$ such that
$A_0 x^* = b_0,$ 
where $b_0$ is the subvector of $b$ corresponding to the rows in $A_0$.
\end{lemma}

\begin{lemma}
\label{lem:LinearProg-2}
For a linear program with the polyhedron 
$$
P = \{x \in \mathbb{R}^n : 
A_1x = b_1, \,\, A_2x \leq b_2, \,\, A_3x \geq b_3\},
$$
where $A_1 \in \mathbb{R}^{m_1 \times n}, A_2 \in \mathbb{R}^{m_2 \times n}, A_3 \in \mathbb{R}^{m_3 \times n}$ are matrices and $b_1 \in \mathbb{R}^{m_1}, b_2 \in \mathbb{R}^{m_2}, b_3 \in \mathbb{R}^{m_3}$ are vectors. Let $x^*$ be an extreme point (or basic feasible solution) of the linear program. For an index set $Q \subseteq [n]$ (recall $x^*\assigned{Q}$ is the vector that retains only the elements indexed by $Q$), $x^*\assigned{Q}$ is still an extreme point of the linear program with polyhedron 
$$
P' = \{ x\assigned{Q} : x \in P, \,\,
x_i = x^*_i \text{ for all } i \in [n] \setminus Q \}.
$$ 
\end{lemma}

\begin{proof}
The lemma is standard and we provide a proof for completeness.
Clearly, $x^* \in P'$. 
Assume that $x^*$ corresponds to a basis $A_0x = b_0$ for the polyhedron $P$, where $A_0 \in \mathbb{R}^{n \times n}$. 
By removing the columns corresponding to the index set $[n] \setminus Q$, we denote $A_0^{(Q)} \in \mathbb{R}^{n \times |Q|}$ as the modified matrix.
Since the columns of $A_0$ are linearly independent, the columns of $A_0^{(Q)}$ are also linearly independent. Therefore, $A_0^{(Q)}$ is of rank $|Q|$, and we can select a submatrix $A_0'$ corresponding to exactly $|Q|$ linearly independent rows of $A_0^{(Q)}$.
The constraints of $P'$ can be interpreted as removing all variables in $[n] \setminus Q$, and $A_0'$ corresponds to a basis for $P'$. 
Thus, $x^*\assigned{Q}$ is also an extreme point of the linear program with polyhedron $P'$.
\end{proof}





%% file: reduction.tex
\section{A General Reduction to Multi-Budgeted Optimization Problem}
\label{sec:reduction}

In this section, we provide an equivalent formulation for the general symmetric norm minimization problem \minnorm\ (up to constant approximation factor). Recall that as defined in Section~\ref{sec:intro}, 
we are given a set $\calU$ of $n$ elements, and 
$\calF$ represents a family of feasible subsets of $\calU$.
The goal of \minnorm\ is to find a feasible subset $S\in\calF$ to minimize $f(\boldv[S])$,
where $f$ is a symmetric monotone norm function. 
We say that we find a $c$-approximation for the problem for some $c\geq 1$,
if we can find an $S$ such that $f(\boldv[S])\leq c\cdot f(\boldv[S^*])$, where $S^*$ is the optimal solution. 
Since a general norm function is quite abstract and hard to deal with,
we formulate the following (equivalent, up to constant approximation factor) optimization problem which is more combinatorial in nature. 


\begin{definition}
[Logarithmic Budgeted Optimization (\LBO)]
\label{def:LBO}
The input of a Logarithmic Budgeted Optimization Problem is a tuple $\eta=(\calU;S_1,S_2,\ldots,S_{T};\calF)$, where:
\begin{itemize}
    \item $\calU$ is a finite set with $n$ elements.
    \item $S_1,S_2,\cdots, S_{T}$ are disjoint subsets of $\calU$, where $T=\ceil{\log n}$ is the number of sets. For $1\le i\le T$, We refer to $S_i$ as the $i$-th group, and for any $u\in S_i$, we call $i$ the group index of $u$. 
    \item $\calF$ is a family of feasible subsets of $\calU$. The size of $|\calF|$ may be exponentially large in  $n$, but we ensure that there exists a polynomial-time algorithm to decide whether $D\in \calF$ for a subset $D\subseteq \calU$.
\end{itemize}
For any $c'>0$, we say a subset $D\subseteq \calU$ is a \textbf{$c'$-valid} solution if and only if:
\begin{enumerate}
    \item $D$ satisfies the feasibility constraint, i.e., $D\in \calF$, and 
    \item $|D\cap S_i|\le c'\times 2^i$ for all $1\le i\le T$.
\end{enumerate}
For any $c\ge c_0\ge 1$, we define $(c,c_0)$-\LBO\ problems as follows: 
Given an input $\eta$, the goal is 
to find a $c$-valid solution or certify that
there is no $c_0$-valid solution.
In particular, we denote $(c,1)$-\LBO\ 
as $c$-\LBO.
\end{definition}

Notice that the structure of a problem is defined by $\calU$ and $\calF$ (for example, the vertex cover problem is given by vertex set $\calU$ and $\calF$ contains all subsets of $\calU$ corresponding to a vertex cover), so each problem corresponds to a \minnorm\ version and an \LBO\ version.
We show that solving \LBO\ is equivalent to approximating \minnorm, up to constant approximation factors. 
In fact, the reduction from norm approximation to optimization problem with multiple budgets has been implicitly developed in prior work \cite{chakrabarty2019approximation,ibrahimpur2021minimum}. 
For generality and ease of usage, we encapsulate the reduction in the following general theorem.

\begin{theorem}
\label{thm:equivalence}
For any $c\geq 1$ ($c$ can depend on $n$) and $\epsilon>0$,
if we can solve $c$-\LBO\ in polynomial time, 
we can approximate the \minnorm\ problem within a factor of $(4c+\epsilon)$ in polynomial time.
On the other hand, if we can find a $c$-approximation for \minnorm\
in polynomial time, we can solve the $47c^2$-\LBO\ in polynomial time.
\end{theorem}

In the following, we give a formal proof for \cref{thm:equivalence}.
We also give another form of it in \cref{thm:diffequiv} that will be used in the following sections.

To prove Theorem~\ref{thm:equivalence}, we first define 
some notations.
Let the optimal solution for the \minnorm\ problem be $S^*$. Define $\boldu = \boldv[S^*]$ and let $o = \boldu^{\da}$ (recall that $\boldu^{\da}$ represents the vector obtained by sorting the elements of $\boldu$ in non-increasing order).

Let $\delta,\varepsilon$ be (small) positive constants. 
Define the set of positions
$$
\POS=\{\min\{2^s,n\}:s\geq 0\}.
$$
We need the following lemma, which is proved in \cite{chakrabarty2019approximation}.
The lemma says that we can in polynomial time construct a poly-sized set 
of threshold vectors such that there is one threshold vector
which is close to the optimal vector for all positions in $\POS$.

\begin{lemma}
\label{lm:polyguess}
\cite{chakrabarty2019approximation}
Let $\varepsilon$ be a fixed small positive constant. Suppose we can obtain in polynomial time a set $Q$ of polynomial size that is guaranteed to contain a real number in $[o_1^\da,(1+\varepsilon)o_1^\da]$.
Then, in time $O\left(|Q|\max\{(n/\varepsilon)^{O(1/\varepsilon)},\text{poly}(n)\}\right)$, we can obtain a set $T$ of polynomial many vectors in $\R^{\POS}$, which contains a threshold vector (denoted by $\boldsymbol{t}^*$) satisfying:
$o_\ell^\da\leq t^*_\ell\leq (1+\varepsilon)o_\ell^\da$ if $o_\ell^\da\geq \varepsilon o_1^\da/n$ and $t^*_\ell=0$ otherwise.
Moreover, for all $\ell\in \POS$,
$t^*_\ell$ is either 0 or a power of $1+\varepsilon$.
\end{lemma}

Let $\text{advance}(i)$ be the smallest element in $\POS$ that is at least $i$, for 
$1\leq i\leq n$.
For each $\boldt\in\R^\POS$, define $g(\boldt)$ to be the $n$-dimensional vector such that $g(\boldt)_i=\boldt_{\text{advance}(i)}$. 

\begin{lemma}
\label{lem:d-2}
Given a norm function $f$, if $\boldt$ is a valid threshold (i.e., it satisfies the statement in Lemma~\ref{lm:polyguess}), 
we have
$$f(g(\boldt))\leq (1+\varepsilon)f(\boldo).$$
\end{lemma}
\begin{proof}
By the construction of $g$, we can observe that
$$
g(\boldt)_i=\boldt_{\text{advance}(i)}\leq (1+\varepsilon)o_{\text{advance}(i)}
\leq (1+\varepsilon)o_{i}.
$$
Hence, $f(g(\boldt))\leq f((1+\varepsilon)\boldo)=(1+\varepsilon)f(\boldo)$,
which proves the lemma.
\end{proof}

Moreover, we have the following lemma:
\begin{lemma}
\label{lem:d-3}
We are given an $n$-dimensional vector $\boldu$,
a threshold vector $\boldt$, and 
a positive integer $c\geq 1$.
If for each $\ell\in \POS$ ($1\leq i\leq T$), there are at most $c\ell$ entries in 
$\boldu$ that is larger than $\boldt_\ell$ and the largest entry of $\boldu$
is at most $\boldt_1$, then $f(\boldu)\leq 2cf(g(\boldt)).$
\end{lemma}
\begin{proof}
We construct an $n$-dimensional vector $\boldu'$. 
For each $\ell\in \POS$ and $\ell'$ which is the largest element in $\POS\cup \{0\}$ that is less than $\ell$,
we have $\boldu'_i=t_\ell$ for $\min\{n,2c\ell'\}<i\leq \min\{n,2c\ell\}$.
So there are at least
$\min\{n,2c\ell\}$ elements in $\boldu'$ that is at least $\boldt_\ell$.

We consider the $i$th element in $\boldu^\da$. Let $\ell$ be the smallest number in $\POS$ such that $i\leq 2c\ell$.
Then if $\ell>1$, $\ell/2$ is also in $\POS$, and thus $i>c\ell$.
By definition of $\boldu'$, we have ${\boldu'_i}^\da=\boldt_\ell$.
If $\boldu^\da_i>\boldt_\ell$
then $\ell>1$ because the largest element in $\boldu$ is at most $\boldt_1$, and then the number of elements larger than $\boldt_\ell$ in $\boldu$ is at least $i>c\ell$, contradiction to the condition of the lemma.
Therefore, for any $1\leq i\leq n$, $\boldu_i^\da\leq {\boldu'_i}^\da$. This gives
$$f(\boldu)\leq f(\boldu').$$

Also, notice that for any $\ell\in\POS$ and $\ell'$ which is the largest element in $\POS\cup \{0\}$ that is less than $\ell$, there are $\ell-\ell'$ elements $\boldt_\ell$ in $g(\boldt)$.
So by duplicating $g(\boldt)$ for $2c$ times and choosing the largest $n$ elements, we gets $\boldu'$. Thus by triangle inequality of norm, we have $$f(\boldu')\leq 2cf(g(\boldt)).
$$
This completes the proof.
\end{proof}

Now we prove Theorem~\ref{thm:equivalence}
which establishes
the equivalence between norm approximation and
the \LBO\ problem.

\begin{proofofthm}{\ref{thm:equivalence}}
We first reduce the \minnorm\ problem to the \LBO\ problem. We enumerate all possible threshold vectors $\boldsymbol{t}\in \mathbb{R}^{\text{POS}}$ as defined in Lemma~\ref{lm:polyguess}
(note that there are only $n$ elements, so we just need to let $Q$ (in Lemma~\ref{lm:polyguess}) be the set of the values of all elements).
Suppose $\boldsymbol{t}$ is a valid guess. We can also assume that we have guessed the exact value of $o_1^\da$ (because it only has $n$ choices).
We construct sets $S_1,\cdots,S_T$ for \LBO\ 
in the following way. For each element $e\in \calU$, we do the following:
\begin{itemize}
    \item if its value $v_e$ is larger than $\boldt_1$, we do not add it to any set;
    \item if its value $v_e$ is at most $\max\{\boldt_n,\varepsilon o_1^\da/n\}=\max\{\boldt_n,\varepsilon \boldt_1^\da/n\}$, we add it to $S_{T}$;
    \item   otherwise, if its value $v_e$ is at most $\boldt_{\ell}$ and larger than $\boldt_{\text{next}(\ell)}$,
    where $\ell=2^i$, we add it to $S_{i+1}$, for $0\leq i\leq T-1$. Recall $\text{next}(\ell)$ means the next element in $\POS$.
\end{itemize}

Now, consider the optimal solution $\boldo=\boldu^\da[S^*]$ for \minnorm\ corresponding to the valid guess $\boldt$.
For $1\leq i\leq T-1$, consider how many elements in $S^*$ that are added in $S_i$. By definition, an element $e$ in $S_i$ have value $\boldv_e$ larger than $\boldt_{2^i}$.
Also, the value $v_e$
is larger than $\varepsilon o_1^\da/n$. 
For $\ell\in\POS$, $o_\ell\leq \boldt_\ell$ by definition of valid guess. 
Thus, since $\boldt$ is a valid guess, there are at most $2^i$ elements of $S^*$ from $S_i$
(note that this also holds for $i=T=\ceil{\log n}$ as there are $n$ elements).

If we can solve the $c$-\LBO\ problem, we can get a solution $S\in \calF$
such that there are at most $c\cdot 2^i$ elements in each $S_i$.
As each element $e$ with value $v_e>t_\ell$ for $\ell=2^i$ cannot be in $S_j$ for $j>i$, for each $\ell$, the number of elements with value larger than $t_\ell$ is at most $2c\ell$.
We partition $S$ into $A,B$ such that the elements in $A$ have values less than $\varepsilon o_1/n$ and 
the elements in $B$ have values at least $\varepsilon o_1/n$. 
We need this partition because of the condition 
$o_\ell^\da\geq \varepsilon o_1^\da/n$ in Lemma~\ref{lm:polyguess}.
Note that $A,B$ are not needed in the algorithm, but only useful in the analysis.
So we have
$$
f(\boldv[S])\leq f(\boldv[A])+f(\boldv[B])\leq n\cdot f(\varepsilon o_1/n)+4cf(g(\boldt))\leq \varepsilon f(\boldo)+4c(1+\varepsilon)f(\boldo)\leq 4c(1+2\varepsilon)f(\boldo)
$$
where the first inequality following from the triangle inequality of the norm,
the second from the definition of $A$ and Lemma~\ref{lem:d-3} and the third from Lemma~\ref{lem:d-2}.
Therefore, the first part of Theorem \ref{thm:equivalence} follows.

Now we prove the other direction of Theorem \ref{thm:equivalence}. 
Suppose we can $c$-approximate the norm optimization problem \minnorm\
for any monotone symmetric norm.
Consider the norm
$$
f(\mathbf{v})=\max_{\ell=2^i,1\leq i\leq \log (n+2)}\left\{\frac{\topp{\ell-2}(\mathbf{v})}{2^{i/2}}\right\}.
$$
We construct $\calU$ to be the union of $S_1,S_2,\cdots,S_T$ and
the elements in $S_i$ have value $1/2^{i/2}$ for any $1\leq i\leq T-1$. 
The elements in $S_T$ have value 0. 
So for the 1-valid optimal solution $S^*$ for the $\LBO$ problem, consider the value of $f(\boldv\assigned{S^*})$. For any $1\leq i\leq T-1$ and the corresponding $\ell=2^i$,
we can see that
$$
\topp{\ell-2}(\mathbf{v})=\sum_{j=1}^{i-1}\frac{2^j}{2^{j/2}}\leq \frac{2^{i/2}}{\sqrt{2}-1}.
$$
So the norm $f(\boldv\assigned{S^*})$ is at most $1/(\sqrt{2}-1)$.
    
Let $a=\ceil{\log \frac{4c^2}{(\sqrt{2}-1)^2}}$ and $c'=2^a$. 
For a feasible solution $S$ for the $\LBO$, suppose there are at least $2^a\cdot 2^j$
elements in $S_j$ (i.e. $|S\cap S_j|\geq 2^{a+j}$). 
We consider the $\topp{\ell-2}(\mathbf{v})$ for $i=a+j,\ell=2^i$ (As there are at most $n$ elements, $i\leq\log n$). 
It shows that the norm $f(\boldv[S])$ is at least $$\frac{2^{a+j}-2}{2^{(a+j)/2+j/2}}> \frac{1}{2}\sqrt{2^a}$$
as $a+j\geq 3$.
So consider the solution $D\in\calF$ of $\LBO$ we get from solving 
the $\minnorm$ problem by reduction above.
If the norm $f(\boldv[D])$ is larger than $c/(\sqrt{2}-1)$, there is no 1-valid solution for $\LBO$. Otherwise, the norm $f(\boldv[D])$ is at most $c/(\sqrt{2}-1)\leq \sqrt{2^a}/2$. Since the norm is at most $\frac{\sqrt{2^a}}{2}$, we know that the solution for the $\minnorm$ problem is a solution for the $c'$-$\LBO$ problem. As
$$c'\leq 2\times\frac{4c^2}{(\sqrt{2}-1)^2} <47c^2,$$
the second part of Theorem~\ref{thm:equivalence} follows.
\end{proofofthm}

\vspace{0.1cm}
\noindent
{\bf A Logarithmic Approximation:}
Based on \cref{thm:equivalence}, we can easily deduce the following general theorem.
We use $\myproblem$ to denote a general combinatorial optimization problem with the min-sum objective function
$\min_{S\in \calF} \boldv(S)$, 
where we write $\boldv(S)=\sum_{e\in S}v_e$
and 
$\calF$ is the set of feasible solutions.

\begin{theorem}
\label{thm:logapprox}
If there is a poly-time approximation algorithm for the min-sum problem $\myproblem$ (with approximation factor $\alpha\geq 1$), 
there is a poly-time factor $(4\alpha\lceil\log n\rceil+\epsilon)$ approximation algorithm for the corresponding \minnorm\ problem for any fixed constant $\epsilon>0$.
\end{theorem}

\begin{proof}
By \cref{thm:equivalence}, we just need to find a poly-time $O(\log n)$ approximation algorithm for the \LBO\ version.
Consider the input $\calU,\calF,S_1,\cdots,S_T$ where $T=\lceil \log n\rceil$ (recall the definitions in \cref{sec:reduction}).
Then we construct $v_e$ for each $e\in \calU$ by:
\begin{enumerate}
\item $v_e=0$ if $e\not\in S_1,S_2,\cdots,S_T$;
\item $v_e=1/2^i$ if $e\in S_i$.
\end{enumerate}
Then, if there is a $1$-valid solution for the \LBO\ problem, 
we can see that 
there is a feasible set $S^\star\in\calF$ with $\boldv(S^\star)\leq T$.
Then, by the assumption of the theorem, the approximation algorithm for $\myproblem$ can output a feasible solution $S\in \calF$ with $\boldv(S)\leq \alpha T$.
This further implies that
$S$ is a $\alpha T$-valid solution, because
$$
|S\cap S_i|=\boldv(S\cap S_i)\cdot 2^i\leq \alpha T \cdot 2^i,
\text{ for each } i\in [T].
$$
This means that the $\alpha T$-\LBO\ problem can be solved in polynomial time.
By \cref{thm:equivalence}, we complete the proof.
\end{proof}

%% file: knapsackcover.tex
\section{Multi-dimensional Knapsack Cover Problem} 
\label{sec:knapsackcover}
In this section, we consider the multi-dimensional knapsack cover problem
defined as follows.

\begin{definition}[Min-norm $d$-dimensional Knapsack Cover Problem (\minnormknapcov)]
Let $d$ be a positive integer.
We are given a set of items $\calU = \{1, 2, \dots, n\}$, where each item $i \in \calU$ has a weight vector $w_i \in \mathbb{R}^d$. The feasible set $\calF$ is defined as:

\[
\calF = \left\{ D \subseteq \calU : \sum_{v \in D} w_{v,i} \geq 1 \quad \forall i \in \{1, 2, \dots, d\} \right\}.
\]

Now, given a symmetric monotone norm $f$ and a value vector $\boldv \in \mathbb{R}_{\geq 0}^{\calU}$,
we can define the norm minimization problem for $d$-dimensional Knapsack Cover and denote it as \minnormknapcov.
\end{definition}

In light of \cref{thm:equivalence}, 
we introduce $T = \lceil \log n \rceil$ disjoint sets $S_1, S_2, \dots, S_T$ and consider the \LBO\ problem with $(\calU;S_1,\ldots,S_{T};\calF)$,
which We denote as 
\LBOknapcov.
We consider \LBOknapcov\ for two cases:
(1)
$d=O(1)$ and 
(2) \(d = O( \sqrt{\log n/\log\log W})\) ($W$ will be defined 
in Section~\ref{sec:Knap-2}).
For both cases, we use the following natural linear programming formulation for \LBOknapcov\ :

\begin{equation}
\begin{aligned}
	\min && 0 && & \\
    s.t. && \sum_{v\in \calU} x_v w_{v,i}\geq 1 && & \forall 1\leq i\leq d\\
    && \sum_{v\in S_j} x_v\leq 2^j && & \forall 1\leq j\leq T\\
    && 0\leq x_v\leq 1 && & \forall v\in S_j,1\leq j\leq T
\end{aligned}
\tag{LP-KnapCover-1}
\label{LP-KnapCover-1}
\end{equation}

For both cases, we develop a method called \textbf{partial enumeration}. Partial enumeration lists a subset of possible partial solutions for the first several groups. Here is the complete definition:

\begin{definition}[Partial Enumeration] \label{def:part-enum}
For a $(c,c_0)$-\LBO\ problem with set $S_1,S_2,\cdots S_{T}$, the \textbf{partial enumeration} algorithm first determine a quantity $T_0$ (depending on the problem at hand). The partial enumeration algorithm returns a subset $X\subseteq 2^{S_1}\times 2^{S_2}\times \cdots \times 2^{S_{T_0}}$. Each element of $X$ is a \textbf{partial solution} $(D_1,D_2,\cdots, D_{T_0})$ (Recall the definition of partial solution: $D_i\subseteq S_i$), and this algorithm ensures: 

\begin{enumerate}
    \item If there exists a $c_0$-valid solution, then at least one partial solution $(D_1,D_2,\cdots, D_{T_0})\in X$ satisfies
    that 
    there exists an $c$-valid \textbf{extended solution} (A solution $D$ is called an extended solution of a partial solution $(D_1,\cdots,D_{T_0})$ if $D\cap S_i=D_i$ for all $i=1,2,\cdots, T_0$).

    \item The size of $X$ is polynomial, and this partial enumeration algorithm runs in polynomial time.
\end{enumerate}
\end{definition}

\subsection{An Algorithm for $d=O(1)$}
\label{sec:Knap-1}
In this subsection, we design a polynomial-time constant-factor approximation algorithm for \minnormknapcov\
with $d=O(1)$. 
\footnote{
We are grateful to an anonymous reviewer for her/his insightful suggestions, which significantly simplifies the algorithm
in this subsection (in particular the rounding algorithm
Algorithm~\ref{Algo-IR-Knap}). 
Here we only present the simplified algorithm.
}
\begin{theorem} \label{thm:knap-1}
If $d$ is a constant, then for any 
constant $\varepsilon>0$, there exists an polynomial-time algorithm which can solve $(1+\varepsilon)$-\LBOknapcov . Thus we have a polynomial-time $(4+\varepsilon)$-approximation algorithm for \minnormknapcov\ when $d=O(1)$.
\end{theorem}
In this subsection, the partial enumeration algorithm needs an integer $T_0>\log d$ (we discuss  how to determine the exact value of $T_0$ later). Recall that a partial solution can be described as a vector $(D_1,D_2,\cdots, D_{T_0})$, where $D_i\subseteq S_i$ for $i=1,2,\cdots, T_0$.
Our partial enumeration algorithm 
in this section simply outputs 
the set of all possible \textbf{partial solution}s 
(i.e., $X=X^{all}_1\times X^{all}_2\times \ldots \times X^{all}_{T_0}$ where $X^{all}_i=\lbrace X'\subseteq S_i: |X'|\le 2^i\rbrace$). The number of of such partial solutions 
can be bounded as in Lemma~\ref{lem:Knap-2}.  

Then we need to extend a partial solution to a complete solution, i.e., construct the rest of the sets $D_{T_0+1},D_{T_0+2},\cdots,D_T$.
For a partial solution $(D_1, D_2, \dots, D_{T_0})$, a solution $D \subseteq \calU$ is called an \textbf{Extended Solution} of $(D_1, D_2, \dots, D_{T_0})$ if and only if $D \cap S_j = D_j \quad \text{for each } 1 \leq j \leq T_0$. 
Without loss of generality, we sometimes describe an extended solution in the form of $(D_1, D_2, D_3, \dots, D_T)$. 

After the enumeration, we sequentially check these partial solutions. For each partial solution $(D_1, D_2, \dots, D_{T_0})$, 
we consider the following linear program:

\begin{equation}
\begin{aligned}
    \min \quad & 0 \\
    \text{s.t.} \quad & \sum_{T_0 < j \leq T} \sum_{u \in S_j} x_u w_{u,i} \geq \max \left(0,1-\sum_{1\le j\le T_0} \sum_{u\in D_j} w_{u,i}\right) \quad & \forall 1 \leq i \leq d \\
    & \sum_{u \in S_j} x_u \leq 2^j \quad & \forall T_0 < j \leq T \\
    & 0 \leq x_u \leq 1 \quad & \forall u \in S_j, \; T_0 < j \leq T
\end{aligned}
\tag{LP-KnapCover}
\label{LP-KnapCover}
\end{equation}

We solve the LP \cref{LP-KnapCover} and perform the following steps:

\begin{enumerate}
    \item Check the feasibility of 
    Linear Programming \cref{LP-KnapCover}. If there is no feasible solution, then return "No Solution" directly. 
    Otherwise, obtain an extreme point $x^{*}$ for \cref{LP-KnapCover}.

    \item For any $T_0< j\leq T$, round 
    all nonzero $x^{*}_u$ to 1 for all $u \in S_j$. More specifically, we add them to our extended solution: $D_j \gets \{u \in S_j : x^{*}_u > 0\}$ for all $j$ such that $T_0 < j \leq T$.
\end{enumerate}

\begin{algorithm}
\caption{Rounding Algorithm for $d$-dimensional Knapsack Cover Problem} \label{Algo-IR-Knap}
\KwResult{Solution $D$ or "No Solution"}
$D \gets D_1 \cup D_2 \cup \dots \cup D_{T_0}$\;
\If{\cref{LP-KnapCover} has no solution}{
    Return "No Solution"\;
}
Solve Linear Program \cref{LP-KnapCover} and obtain an extreme point $x^{*}$\;
$D \gets D \cup \{u \in S_j \mid T_0 < j\leq T, x^{*}_u > 0\}$\;
Return $D$\;
\end{algorithm}

\begin{lemma} \label{lem:Knap-2}
There exists a $\exp(O(2^{T_0}\log n))$-time algorithm to enumerate all partial solutions $(D_1, \cdots, D_{T_0})$ satisfying $|D_i| \le 2^i$ for all $1 \le i \le T_0$. 
\end{lemma}
\begin{proof}
Assume that $d=o(n)$ and $n$ is large enough. We need to select at most $2^i$ items from $S_i$ for $1\leq i\leq T_0$. The number of ways to choose these items is not greater than: 
\[
\sum_{j=0}^{2^i} \binom{|S_i|}{j} \leq \sum_{j=0}^{2^i} \binom{n}{j}\leq (2^i+1)\binom{n}{2^i}\leq (2^i+1)\cdot\frac{n^{2^i}}{(2^i)!}.
\]
It is at most $n^{2^i}$ when $i\geq 2$. When $i=1$, as $n\geq 4$,
\[
\sum_{j=0}^{2^i} \binom{|S_i|}{j} \leq \sum_{j=0}^{2^i} \binom{n}{j}=1+n+\frac{n(n-1)}{2}\leq n^2.
\]
Therefore, the total number of partial solutions is no more than:
\[
\prod_{i=1}^{T_0} n^{2^i} = n^{\sum_{i=0}^{T_0} 2^i} \le n^{2^{T_0+1}}.
\]
Hence, the time complexity of enumeration 
algorithm can be bounded $\exp(O(2^{T_0}\log n))$.
\end{proof}

\begin{lemma} \label{lem:Knap-3}
If Algorithm \ref{Algo-IR-Knap} returns a solution $D$, then $|D \cap S_j| \le \left( \frac{d}{2^{T_0}}+1\right)\cdot 2^{j}$ for all $T_0 < j \le T$.
\end{lemma}
\begin{proof}

Let $x^{*}$ be an extreme point 
of \cref{LP-KnapCover}. 
By Lemma~\ref{lem:extreme:points}, if $x^*$ has $m$ dimensions, then there are $m$ linearly independent constraints of \ref{LP-KnapCover} that hold with equality at $x^*$.
If the basis contains equality $x_u = 0$ or $x_u = 1$ for some $u \in S_j$, then $x^{*}_u$ must be an integer. 
Excluding those constraints, There are two types of constraints:

\begin{enumerate}
    \item  \textbf{Feasibility Constraints}: $\sum_{T_0 < j \leq T} \sum_{u \in S_j} x_u w_{u,i} \geq \max\left(0,1-\sum_{1\le j\le T_0} \sum_{u\in D_j} w_{u,i}\right), \quad  \forall 1 \leq i \leq d$;
    \item  \textbf{Cardinality Constraints}: $\sum_{u \in S_j} x_u \leq 2^j, \quad \forall T_0 < j \leq T$.
\end{enumerate}

Without loss of generality, we assume that there are $m_1$ tight feasibility constraints, $m_2$ tight cardinality constraints, and $m_1+m_2$ fractional entries in $x^{*}$
(By \cref{lem:extreme:points} based on the linear program after removing integer entries according to Lemma~\ref{lem:LinearProg-2}). For a {\em tight} cardinality constraint corresponding to $S_j \,\,(T_0<j\leq T)$, since it only contains integer coefficients, $x^{*}\assigned{S_j}$ must include at least two fractional entries. Therefore, $m_1+m_2\ge 2m_2$,
(the total number of fractional entries is at least 
the number of fractional entries associated with cardinality constraint),
which implies $m_2\le m_1\le d$. For $T_0< j\leq T$, we think about the fractional entries in the constraint corresponding to $S_j$. The total number of fractional entries is no more than $m_1+m_2$ and each other group ($S_{j'}$, $j'\ne j$) contains at least two fractional entries. Therefore, the fractional variables from $D\cap S_j$ is at most $m_1+m_2-2\max(m_2-1,0)$.
Hence, we can conclude the following:
\begin{align*}
|D \cap S_j| &= \sum_{u \in S_j} \mathbf{1}_{x^{*}_u > 0} \le \left( \sum_{u \in S_j} \mathbf{1}_{0 < x^{*}_u < 1} \right) + \left( \sum_{u \in S_j} x^{*}_u \right) \\
& \le \left( m_1+m_2-2\max(0,m_2-1)\right) + 2^j \le d+1+2^{j}\le \left( \frac{d}{2^{T_0}}+1\right)\cdot 2^j.
\end{align*}
\end{proof}

\begin{proofofthm}{\ref{thm:knap-1}}
According to \cref{lem:Knap-2} and \ref{lem:Knap-3}, our algorithm 
runs in $\exp(O(2^{T_0}\log n))$ time,
and it can return a 
$\left( \frac{d}{2^{T_0}}+1\right)$-valid solution or confirm there is no $1$-valid solution.
Hence, for any constant $\epsilon>0$, we can choose $T_0=\log d+\log(1/\epsilon)+O(1)$. Then, the algorithm runs $\exp(O(d\log n/\epsilon))$-time for $(1+\epsilon)-$\LBOknapcov.

Finally, applying \cref{thm:equivalence}, we obtain a polynomial-time $(4+\varepsilon)$-approximation algorithm for the \minnormknapcov\ problem with $d=O(1)$.
\end{proofofthm}

\subsection{An Algorithm for Larger $d$}
\label{sec:Knap-2}

In this subsection, we provide a polynomial-time constant-factor approximation algorithm for $d = O\left( \sqrt{\frac{\log n}{\log\log W}}\right)$, where $W$ is defined as
$$W=\max_{1\le i\le d} \frac{\max_{u\in \calU} w_{u,i}}{\min_{u\in \calU,w_{u,i}>0} w_{u,i}}.$$
To ensure there exist valid solutions, $\lbrace u\in \calU: w_{u,i}>0\rbrace$ must be a non-empty set. Here is the main result in this subsection:
\begin{theorem} \label{thm:knap-2}
There exists a \(poly(n, \log(W))\) algorithm that can solve $2$-\LBOknapcov\ when $d=O\left(\sqrt{\frac{\log n}{\log\log W}}\right)$. Thus we have a $(4+\varepsilon)$-approximation algorithm for \minnormknapcov\ with $d=O\left(\sqrt{\frac{\log n}{\log \log W}}\right)$.
\end{theorem}
This algorithm employs the same rounding procedure but modifies the partial enumeration method. The new algorithm choose $T_0=\lceil \log d\rceil$. For \( 1 \leq j \leq T_0 \), it partitions \( S_j \) into multiple subsets based on vectors of size \( d \), which represent the logarithms of weights. Instead of enumerating all subsets, we only enumerate the number of elements within each subset and then take double the number of any elements in this subset.

Now, we provide an overview of the algorithm for the \LBOknapcov . We first describe the partial enumeration algorithm. 
We first obtain a set $X_j\subseteq 2^{S_j}$ for each \(1 \le j \le T_0\), and then combine them by direct product: \(X = X_1 \times X_2 \times \cdots \times X_{T_0}\). For each \(j\), we classify \(w_{u,i}\) into \(O(\log W)\) groups for each coordinate \(i = 1, 2, \cdots, d\), and partition \(u \in S_j\) into a total of \(O((\log W)^d)\) groups. We then enumerate the number of elements in each group and select arbitrary elements in the groups.
We can show the number of such partial solutions is bounded by
a polynomial (Lemma~\ref{lem:Knap-Enum-1}).
For each partial solution \((D_1, D_2, \cdots D_{T_0}) \in X\), we use it as part of the input the rounding algorithm (Algorithm \ref{Algo-IR-Knap}). 

Since the rounding algorithm is the same as in the previous section, 
we only describe the partial enumeration procedure in details. 
The details are as follows:

\begin{enumerate}
\item For each dimension \(i\), find the minimum positive weight \(\gamma_{j,i}\) among all items in \(S_j\) (if all of them are 0, \(\gamma_{j,i}\) can be anything and we do not care about it).

\item Construct a modified vector \(w'_u\) for each \(u\) in \(S_j\):
\begin{itemize}
    \item If \(w_{u,i} > 0\), set \(w'_{u,i} = \lfloor \log(w_{u,i}/\gamma_{j,i}) \rfloor + 1\).
    \item If \(w_{u,i} = 0\), set \(w'_{u,i} = 0\).
\end{itemize}

\item We say that two vectors \(w'_u\) and \(w'_v\) are equal if 
\(w'_{u,i} = w'_{v,i}\) for all \(i = 1, 2, \cdots, d\). 
Let \(r_j\) as the number of different vector \(w'_u\)s for $u\in S_j$.

\item Partition \(S_j\) into \(r_j\) groups: \(S_{j,1}, S_{j,2}, \cdots, S_{j,r_j}\) based on different \(w'_u\). Formally, 
$u\in S_j$ and $v\in S_j$ belong to the same group
if $w'_u = w'_v$.

\item For each vector \(c \in \mathbb{Z}^{r_j}\) such that \(c_k \le |S_{j,k}|\) for all \(1 \le k \le r_j\) and $c_1+c_2+\cdots+c_{r_j}\leq 2^j$, pick a subset from \(S_j\) with exactly \(\min \{2c_k, |S_{j,k}|\}\) elements from each partition \(S_{j,k}\), and add this subset to \(X_j\). (Recall that $X_j\subseteq 2^{S_j}$ is a set of subsets of $S_j$. )

\item Finally, the algorithm returns $X=X_1\times X_2\times \ldots\times X_{T_0}$ as the set of partial solutions.
\end{enumerate}

The pseudo-code can be found in Algorithm~\ref{Algo-Enum-Knap}.
\begin{algorithm}
\caption{Partial Enumeration Algorithm for $d$-dimensional Knapsack Cover Problem} \label{Algo-Enum-Knap}
\KwResult{a set $X$ that $X\subseteq 2^{S_1}\times \cdots \times 2^{S_{T_0}}$}
$X_j\gets \emptyset\forall j=1,\cdots T_0$\;
\ForAll{$j=1$ to $T_0$}{
$\gamma_{j,i}\gets \min_{u\in S_j,w_{u,i}>0}\lbrace w_{u,i}\rbrace$ for all $1\le i\le T_0$\;
Define vector $w'_u\in \mathbb{R}^{d}$ for all $u\in S_j$\;
\ForAll{$u\in S_j$, $i\in \lbrace 1,2,\cdots d\rbrace$}{
    \eIf{$w_{u,i}=0$}{
        $w'_{u,i}\gets 0$\;
    }{
        $w'_{u,i}\gets 1+\lfloor \log(w_{u,i}/\gamma_{j,i})\rfloor$\;
    }
}
Partition $S_j$ into several sets $S_{j,1},S_{j,2},\cdots S_{j,r_j}$ based on different $w'$(For $u\in S_{j,x},v\in S_{j,y}$, $w'_u=w'_v$ if and only if $x=y$)\;

\ForAll{Vector $c\in \mathbb{Z}_{\ge 0}^{r_j}$ that $c_k\le |S_{j,k}|$ for all $1\le k\le r_j$ and $c_1+c_2+\cdots+c_{r_j}\leq 2^j$} {
    Pick a subset $D^{(c)}_j\subseteq S_j$ such that contains exactly $\min\lbrace 2c_k,|S_{j,k}|\rbrace$ elements in $S_{j,k}$ for all $1\le k\le r_j$\;

    $X_j\gets X_j\cup \lbrace D^{(c)}_j\rbrace$
}
}
Return $X=X_1\times \cdots \times X_{T_0}$\;
\end{algorithm}

\begin{lemma}\label{lem:Knap-Enum-1}
If \(d = O\left(\sqrt{\frac{\log n}{\log\log W}}\right)\), then \cref{Algo-Enum-Knap} runs in \(poly(n, \log(W))\)-time.
\end{lemma}
\begin{proof}
For an arbitrary $i$, notice that the number of different $w'_{u,i}$ ($u\in S_j$) is not larger than \(\log(W) + 2\). The number of elements in $S_i$ is no more than $d$, so we have:

\[
r_j \le (\log(W) + 2)^d = \exp(O(d \cdot \log\log(W)))
\]

To enumerate the vector \(c\) for $j$, we consider the number of non-negative integer solutions of:
$$y_1+y_2+\cdots+y_{r_j}\leq 2^j.$$
The number of solutions is well known to be
$$\binom{2^j+r_j-1}{r_j-1}=\binom{2^j+r_j-1}{2^j}\leq r_j^{2^j}.$$
Recall we just need to enumerate the first \(T_0 = \lceil \log d \rceil\) sets. Thus the total complexity is no more than:
\[
\prod_{j=1}^{T_0} r_j^{2^j} \le (\max_{j} r_j)^{2^{T_0}} = (\max_{j} r_j)^{O(d)} = \exp(O(d^2 \cdot \log\log(W)))
\]
which is \(poly(n, \log(W))\).
\end{proof}

\begin{lemma}\label{lem:Knap-Enum-Replace}
For any integer \(j, x\) and \(u_0, u_1, u_2 \in S_{j,x}\), \(w_{u_0,i} + w_{u_1,i} \ge w_{u_2,i}\) for all \(1 \le i \le d\). 
\end{lemma}
\begin{proof}
By definition, we know \(w'_{u_0,i} = w'_{u_1,i} = w'_{u_2,i}\). If \(w'_{u_0,i} = 0\), then \(w_{u_0,i} = w_{u_1,i} = w_{u_2,i} = 0\). The inequality holds. Otherwise, \(\lfloor \log(w_{u_0,i}/\gamma_{j,i}) \rfloor = \lfloor \log(w_{u_2,i}/\gamma_{j,i}) \rfloor\). Then, we know
\(w_{u_0,i} \ge \frac{1}{2} w_{u_2,i}\). Similarly, \(w_{u_1,i} \ge \frac{1}{2} w_{u_2,i}\), so the inequality holds.
\end{proof}

\begin{lemma}\label{lem:Knap-Enum-2}
If the original problem has a $1$-valid solution, then there exists a partial solution $(D_1, D_2, \cdots D_{T_0})$ returned by 
Algorithm~\ref{Algo-Enum-Knap} 
such that there exists an extended solution $(D_1, \cdots D_T)$ satisfying:
\begin{enumerate}
    \item $|D_j| \le 2 \times 2^j$ for $1 \le j \le T_0$, and 
    \item $|D_j| \le 2^j$ for $T_0 < j \le T$.
\end{enumerate}
\end{lemma}
\begin{proof}
Assume that $(D^{*}_1, D^{*}_2, \cdots D^{*}_T)$ is a $1$-valid solution. For each $1 \le j \le T_0$, assume that the partial enumeration algorithm allocates $c_k = |D^{*}_j \cap S_{j,k}|$ for $1\le k\le r_j$. Then 
the algorithm chooses $\min \{2c_k, |S_{j,k}|\}$ elements from 
$S_{j,k}$.
We consider two cases:

\begin{enumerate}
    \item If $c_k < \frac{1}{2} |S_{j,k}|$, due to Lemma \ref{lem:Knap-Enum-Replace}, the sum of chosen weight vectors is no less than that of $S_{j,k} \cap D^{*}$ in all dimensions, and $|D_j \cap S_{j,k}| \le 2|D^{*}_j \cap S_{j,k}|$.
    \item Otherwise, if $c_k \ge \frac{1}{2} |S_{j,k}|$, then it selects the whole group $S_{j,k}$. The total weight
    of this group cannot be less, and the number of elements is also no more than $2|D^{*}_j \cap S_{j,k}|$.
\end{enumerate}
To summarize, we have that:
\begin{enumerate}
    \item $\sum_{u \in D_j \cap S_{j,k}} w_{u,i} \ge \sum_{u \in D^{*}_j \cap S_{j,k}} w_{u,i}$ for all such $i$, and
    \item $|D_j \cap S_{j,k}| \le 2|D^{*}_j \cap S_{j,k}|$.
\end{enumerate}
Therefore, $D^{*}_1, D^{*}_2, \cdots D^{*}_{T_0}, D_{T_0+1}, \cdots D_T$ form such an extended solution satisfying those constraints.
\end{proof}

\begin{lemma}
Algorithm~\ref{Algo-Enum-Knap} is a \textbf{Partial Enumeration Algorithm} for $2$-\LBOknapcov\ when $d=O(\sqrt{\frac{\log n}{\log\log W}})$. More specifically, it satisfies the following properties:

\begin{enumerate}

\item If this problem has a $1$-valid solution, at least one of its output partial solutions has a $2$-valid extended solution.

\item It always runs in polynomial time.

\end{enumerate}
\end{lemma}
\begin{proof}
\cref{lem:Knap-Enum-2} implies the first property and \cref{lem:Knap-Enum-1} implies the second property.
\end{proof}

\begin{proofofthm}{\ref{thm:knap-2}}
We combine the new enumeration method and the previous rounding algorithm (Algorithm \ref{Algo-IR-Knap}) with $T_0=\lceil \log d\rceil$.

Recall the output of Algorithm \ref{Algo-Enum-Knap} is a set of partial solutions: \(X \subseteq 2^{S_1} \times 2^{S_2} \times \cdots \times 2^{S_{T_0}}\). For any \((D_1, D_2, \cdots D_{T_0}) \in X\), we run Algorithm \ref{Algo-IR-Knap} with 
input
\((D_1, \ldots,D_{T_0})\).

\begin{itemize}
\item According to Lemma \ref{lem:Knap-Enum-2}, if there exists a $1$-valid solution, then \cref{Algo-Enum-Knap} can find a partial solution such that there exists an extended solution $D'$ satisfying \begin{itemize}
    \item $|D'\cap S_j|\le 2\cdot 2^j$ for all $1\le j\le T_0$.
    \item $|D'\cap S_j|\le 2^j$ for all $T_0 < j\le T$. 
\end{itemize}
\item If there exists an extended solution $D'$ such that $|D'\cap S_j|\le 2^j$ for all $T_0 < j\le T$, then the solution set of \cref{LP-KnapCover} is non-empty, 
so \cref{Algo-IR-Knap} can return a solution.

\item According to \cref{lem:Knap-3}, if \cref{Algo-IR-Knap} returns a solution $D$, then $|D\cap S_j|\le 2\cdot 2^j$ for all $T_0< j\le T$.  
\item Based on \cref{lem:Knap-Enum-1}, we know this algorithm is polynomial-time.
\end{itemize}
Combining the above conclusions, it is clear that there exists a $poly(n,\log W)$-time algorithm that can solve $2$-\LBOknapcov\ for $d=O(\left(\sqrt{\frac{\log n}{\log\log W}}\right))$.

\end{proofofthm}

%% file: intervalcover.tex
\section{Interval Cover Problem}
\label{sec:intervalcover}

In this section, we study the norm minimization for the interval cover problem, which is defined as follows:

\begin{definition}[Min-norm Interval Cover Problem (\minnormintcov)]
Given a set $\intU$ of intervals and a target interval $\Gamma$ on the real axis, a feasible solution of this problem is a subset $D \subseteq \intU$ such that $D$ fully covers the target interval $\Gamma$ (i.e., $\Gamma \subseteq \bigcup_{I \in D} I$).
Suppose $\boldv\in \mathbb{R}_{\ge 0}^{\intU}$ is a value vector and
$f(\cdot)$ is a monotone symmetric norm function. 
Our goal is to find a feasible subset $D$ such that $f(\boldv\assigned{D})$ is minimized.
We denote the problem as \minnormintcov.
\end{definition}

As the algorithms and proofs of this section are complicated, we just provide our main ideas in this section, and defer the full proofs to \cref{appendix:intcov}.
In light of Theorem~\ref{thm:equivalence},
we can focus on obtaining a constant-factor approximation algorithm for $(c,c_0)$-\LBOintcov. 
The input of a \LBOintcov\ problem is a tuple $\inteta=(\intU;\intS{1},\ldots,\intS{T};\Gamma)$, which is the \LBO\ problem with input $(\intU;\intS{1},\ldots,\intS{T};\intF)$
where the set of feasible solutions is $\intF=\lbrace D\subseteq \intU : \Gamma\subseteq \bigcup_{I\in D} I\rbrace$. 

In this section, we begin by transforming the interval cover problem into a new problem called the {\em tree cover} problem. The definitions of both problems are provided in \cref{sec:Inter-Cover-1}. 
These transformation results in only a constant-factor loss in the approximation factor
(i.e., if a polynomial-time algorithm can solve $c$-\LBOtreecov\ for some constant $c$, then there exists a polynomial-time constant-factor approximation algorithm for \LBOintcov).

Next, we focus on \LBOtreecov. 
We first employ the \textbf{partial enumeration} algorithm, as defined in \cref{sec:knapsackcover}, to list partial solutions for the first $T_0 = \floor{\log\log\log n}$ sets. 
The details of this process are provided in \cref{sec:Inter-Cover-2}. Following partial enumeration, we apply a rounding algorithm to evaluate each partial solution. The entire  rounding process is detailed in \cref{sec:Inter-Cover-3}.

\subsection{From Interval Cover to Tree Cover}
\label{sec:Inter-Cover-1}

We first introduce some notations for the tree cover problem. Denote a rooted tree as $G = (V, E, r)$, where $(V, E)$ forms an undirected tree and $r$ is the root. For each node $u \in V$, let $\Ch(u)$ be the set of children of $u$, and $\Des(u)$ be the set of all descendants of $u$ (including $u$). 
It is easy to see that $\Des(u) = \{u\} \cup \Des(\Ch(u))$.
For a subset of vertices
$P \subseteq V$, we define $\Des(P) = \bigcup_{u \in P} \Des(u)$. 
We also define $\Par(u)$ as the parent node of $u$ and define $\Anc(u)$ as the set of ancestors of $u$ ($\Anc(r) = \{r\}$, and for any $u \in V \setminus \{r\}$, $\Anc(u) = \{u\} \cup \Anc(\Par(u))$). In addition, we define the set of leaves $\Leaf(G) = \{u \in V : \Ch(u) = \emptyset\}$.

\begin{definition}[$\LBO$ Tree Cover Problem (\LBOtreecov)]
We are given a tuple $\treta=(\trU;\trS{1},\ldots,\trS{T};G)$, where $G=(V,E,r)$ is a rooted tree and $\trU = V \setminus \{r\}$. $T=\lceil \log n\rceil$, and $\trS{1},\trS{2},\cdots \trS{T}$ is a partition of $V\backslash\{r\}$ (and $\trS{i}$ 
is called the $i$th group), and $r$ is not in any group. 
The partition $\trS{1},\trS{2},\cdots \trS{T}$ 
satisfies the following property:
For any node $u\in \trS{i}$ and an arbitrary child $v$ of $u$,
$v$ belongs to group $\trS{j}$ with $j>i$.
For each $u \in V \setminus \{r\}$, we define $\Lev(u) = j$ if $u \in \trS{j}$. In particular, we denote $\Lev(r) = 0$. So $\Lev(u) > \Lev(\Par(u))$ for all $u \in \trU$.

A feasible solution for the tree cover problem is 
a subset $D \subseteq \trU$ such that the descendants of $D$ covers all leaves. Formally, the feasible set is defined as
$$
\trF = \{D \subseteq \trU : \Leaf(G) \subseteq \Des(D)=\bigcup\nolimits_{u \in D} \Des(u)\}.
$$
\end{definition}

We prove the following theorem to reduce the interval cover problem to the tree cover problem. The proof of the theorem can be found in \cref{sec:proofof6.1}.

\begin{theorem}\label{thm:inter-to-tree}
If there exists a polynomial-time algorithm for the $(c, 8c_0)$-\LBOtreecov\ problem, then there exists a polynomial-time algorithm for the $(3c, c_0)$-\LBOintcov\ problem.
\end{theorem}


Based on this theorem, we mainly need to deal with the tree cover problem in the following subsections.

\subsection{Partial Enumeration Method for Tree Cover Problem}
\label{sec:Inter-Cover-2}

In this subsection, we present a partial enumeration algorithm for 
the \LBOtreecov\ problem.
Recall that we introduced the concept of partial enumeration in \cref{sec:knapsackcover}. For a \LBOtreecov\ problem with input $(\trU;\trS{1},\trS{2},\ldots,\trS{T};G)$, where $G=(V,E,r)$ is a rooted tree, and $n=|\trU|$. 
we set $T_0 = \lfloor \log \log \log n \rfloor$ and perform partial enumeration for the first $T_0$ sets. The goal is to find a set $X \subseteq 2^{\trS{1}} \times 2^{\trS{2}} \times \dots \times 2^{\trS{T_0}}$ such that there exists a partial solution $(D_1, D_2, \dots, D_{T_0}) \in X$ satisfying: at least one of 
the partial solution can be extended to a $c$-valid solution for some constant $c$.

In this subsection, we define $\Lev(u)$ as the group index of $u$ for $u\in \trU$.
Now we focus on the \LBOtreecov\ problem. 
For each $u \in \trU$, define the first type of cost $C_1(u) = \frac{1}{2^{\Lev(u)}}$. We then define the second type of cost:

$$
C_2(u)=\left\lbrace \begin{aligned}
&C_1(u) && & \text{if $u\in \Leaf(G)$}\\
&\min \lbrace C_1(u),\sum_{v\in \Ch(u)} C_2(v)\rbrace&& & \text{if $u\not \in \Leaf(G)$}\\
\end{aligned}\right.
$$
Intuitively, the cost $C_1(u)$ represents the "cost" of selecting $u$, as it indicates the proportion of the group that $u$ occupies.
Meanwhile, $C_2(u)$ denotes the minimum cost required to cover $u$ using its descendants.

We now present the partial enumeration algorithm. The pseudo-code can be found in \cref{Algo-Enum-Tree} in \cref{sec:proofof6.2}. Here, we briefly describe the main idea of the partial enumeration algorithm.

We employ a depth-first search (DFS) strategy to explore most of the states in the search space. During the search process, we maintain two sets:
\begin{itemize}
    \item \( P\subseteq \trU \), representing the set of candidate elements that can still be explored, i.e., \( \Des(P) \) contains all uncovered leaves.
    \item \( D\subseteq \trU \), storing the elements that have already been selected as part of the partial solution.
\end{itemize}

Initially, $P=\Ch(r)$ is the child set of the root, and $D=\emptyset$. 
At each recursive step, we select \( u\in P \) with the smallest group index. The recursion proceeds by exploring two possibilities:
\begin{enumerate}
    \item Adding \( u \) to the partial solution, i.e., including \( u \) in \( D \) and continuing the search.
    \item Excluding \( u \) from the partial solution, i.e., replacing \( u \) with its child nodes while keeping \( D \) unchanged. (If \( u \) is a leaf, this option is not applicable.)
\end{enumerate}

The search terminates when \( (P, D) \) fails to satisfy at least one of the following conditions:
\begin{enumerate}
\item $\exists u\in P$, $\Lev(u)\le T_0$
\item $\forall u\in D$, $C_2(u) > \frac{1}{\log n}$
\item \( \left(\sum_{v \in D} C_2(v)\right) + \left(\sum_{u \in P} C_2(u)\right) \le 2c_0 T \)
\item \( \forall 1\le i \leq T_0, \quad |D \cap \trS{i}| \le 2c_0 \cdot 2^i \)
\end{enumerate}
The first and fourth conditions are derived from the objective of the Partial Enumeration Method. Regarding the second condition, we observe that for all \( u \in \trU \), it holds that \( \Lev(u) \leq T_0 \) and \( C_1(u) > \frac{1}{\log n} \). Furthermore, if \( C_2(u) \leq \frac{1}{\log n} \), the impact of ignoring \( u \) is negligible.

The third condition is based on the property that for any \( 2c_0 \)-valid solution \( D^{*} \subseteq \trU \), it satisfies:
\[
\sum_{u\in D^{*}} C_2(u) \leq 2c_0 T.
\]

Due to these conditions, for each group \( S_j \) where \( 1\le j\le T_0 \), we only need to determine at most \( 2c_0T^2 \) items. Consequently, we can establish that our partial enumeration algorithm runs in polynomial time.

We then present the following theorem:
\color{black}

\begin{theorem}\label{thm:tree-Enum}
If the \LBOtreecov\ problem with input $\treta$ has a $c_0$-valid solution, \cref{Algo-Enum-Tree} runs
in polynomial time and at least one of the output partial solutions has a $2c_0$-valid extended solution.
\end{theorem}

The complete proof can be found in \cref{sec:proofof6.2}.
\subsection{A Rounding Algorithm for Tree Cover Problem}\label{sec:Inter-Cover-3}

We now focus on the $(c, c_0)$-\LBOtreecov\ problem with input $\treta=(\trU;\trS{1},\ldots,\trS{T};G)$, where $G=(V,E,r)$ is a rooted tree. Let $L=\Leaf(G)$ be the set of leaves. Recall that $\Anc(u)$ represents the set of ancestors of node $u$.
For sets $\calV,\calL\subseteq \trU$ and $c\ge 1$, we express the formulation of the linear program as follows:

\begin{equation}
\begin{aligned}
\min && 0 && & \\
s.t. && \sum_{v\in \Anc(u)\cap \calV} x_v=1 && & \forall u\in \calL\\
	 && \sum_{v\in \trS{i}\cap \calV} x_v\le c\cdot 2^i && & \forall T_0+1\le i\le T\\
    && x_v \ge 0&& & \forall v\in \calV\\
\end{aligned}
\tag{LP-Tree-Cover(c,$\calV$,$\calL$)}
\label{LP-Tree-Cover}
\end{equation}

We call $\sum_{v\in \trS{i}\cap \calV} x_v\le c\cdot 2^i$ {\em cardinality constraints}, and call $\sum_{v\in \Anc(u)\cap \calV} x_v=1$ {\em feasibility constraints}.
Recall that $T_0 = \lfloor \log \log \log n \rfloor$. Also, define $T_1=\floor{\log\log n}$.

The algorithm is as follows, and the pseudocode is \cref{Algo-Round-Tree} in \cref{appendix:intcov}:

\begin{enumerate}
\item Check if LP-Tree-Cover($2c_0$, $V_0$, $L_0$) has a feasible solution. If so, obtain an extreme point $x^{*}$. Otherwise, confirm that there is no such integral solution.

\item Remove the leaves $u$ with $x^{*}_u = 0$, and delete all the descendants of their parents. Then $\Par(u)$ becomes a leaf. Repeat this process until $x^{*}_u\neq 0$ for each leaf $u$. Let the modified node set and leaf set be $V_1$ and $L_1$, respectively.

\item For $u \in V_1$, attempt to round $x^{*}_u$. If $x^{*}_u \geq 1/2$, round it to $1$. If $x^{*}_u > 0$, and $u$ is not a leaf in $\trS{T_1+1} \cup \cdots \cup \trS{T}$, also round it to $1$. In all other cases, round $x^{*}_u$ to $0$. Let $D'$ be the set of nodes $u$ for which $x^{*}_u$ was rounded to $1$. Note that $D'$ may not cover $L_1$.

\item Remove all descendants in $D'$, and attempt to choose another set from $\trS{T_0+1} \cup \cdots \cup \trS{T_1}$ to cover all leaves. Formalize this objective as LP-Tree-Cover. Specifically, 
\begin{itemize}
\item $V_2 = (V_1 \setminus \Des(D')) \cap \{u\in \trU: T_0+1\le \Lev(u)\le T_1\}$, and

\item $L_2 = (V_2\cap L_1)\cup \{u \in V_2 :\exists v \in \Ch(u) \cap (V_1 \setminus V_2), (V_2 \setminus \Des(D')) \cap \Des(v) \cap L_1 \neq \emptyset\}$. 
\end{itemize}
To understand this, observe that \( V_2 \) consists of the nodes in the \((T_0+1)\)th to \( T \)th groups that remain uncovered. The set \( L_2 \) includes nodes in \( V_2 \) that are either leaves or have at least one uncovered child with a group index greater than \( T_0 \) (i.e., at least one descendant leaf remains uncovered).

Then solve LP-Tree-Cover($2c_0$, $V_2$, $L_2$). The fact that this problem must have feasible solutions is proved later, so we do not need to consider the case of no solution.

\item Let $x^{**}$ be an extreme point of LP-Tree-Cover($2c_0$, $V_2$, $L_2$). For each $u \in V_2$, round it to $1$ if and only if $x^{**}_u > 0$. Let $D'' = \{u \in V_2 : x^{**}_u > 0\}$, then $D''$ covers $L_2$.

\item Combine the three parts of the solution. That is, return $\left(\bigcup_{i=1}^{T_0} D_i\right) \cup D' \cup D''$.
\end{enumerate}


We prove the following lemma for \cref{Algo-Round-Tree}:

\begin{lemma}\label{lem:tree-round-3}
Let \( T_0 = \lfloor \log\log\log n \rfloor \). If a partial solution \( (D_1, D_2, \dots, D_{T_0}) \) has a \( 2c_0 \)-valid extended solution, then \cref{Algo-Round-Tree} finds a \( (4c_0 + 1) \)-valid solution.
\end{lemma}

By combining \cref{thm:tree-Enum} and \cref{lem:tree-round-3}, we establish the following theorem:

\begin{theorem}\label{thm:tree-solve}
For any \( c_0 \geq 1 \), there exists a polynomial-time algorithm for \( (4c_0 + 1, c_0) \)-\LBOtreecov.
\end{theorem}

Furthermore, applying \cref{thm:inter-to-tree} and \cref{thm:tree-solve}, we obtain the following result:

\begin{theorem} \label{thm:inter-cover}
There exists a polynomial-time algorithm that solves the \( (3(32c_0 + 1), c_0) \)-\LBOintcov. Consequently, we obtain a polynomial-time constant-factor approximation algorithm for \minnormintcov.
\end{theorem}

The complete proofs are provided in \cref{appendix:intcov}.

%% file: intgap.tex
\section{Integrality Gap for Perfect Matching, $s$-$t$ Path, and $s$-$t$ Cut}
\label{sec:integralgap}

In this section, we argue that it may be challenging to achieve constant approximations for the norm optimization problems for perfect matching, $s$-$t$ path, and $s$-$t$ cut just by LP rounding. 
We show that the natural linear programs
have large integrality gaps.

\begin{definition}[Min-Norm $s$-$t$ Path Problem (\minnormpath)]
Given a directed graph $\pathG = (\pathV, \pathU)$ (here $\pathV$ is the set of vertices and $\pathU$ is the set of edges) and nodes $s, t \in \pathV$, define the feasible set:
$$
\pathF = \{D \subseteq \pathU : D \text{ forms a path from } s \text{ to } t \}.
$$
For the \minnorm\ version, we are also given a monotone symmetric norm $f$ and a value vector $\boldv \in \mathbb{R}_{\geq 0}^{\pathU}$. The goal is to select an $s$-$t$ path $D \in \pathF$ that minimizes $f(\boldv\assigned{D})$.
\end{definition}

In light of \cref{thm:equivalence}, we can define the \LBOpath\ problem with input tuple $\patheta = (\pathU; \pathS{1}, \ldots, \pathS{T}; \pathG; s; t)$, which is the \LBO\ problem defined in \cref{def:LBO} 
with input $(\pathU; \pathS{1}, \ldots, \pathS{T}; \pathF)$.

\begin{definition}[Min-Norm Perfect Matching Problem (\minnormmatch)]
Given a bipartite graph $\pmG = (L, R, \pmU)$
with $|L|=|R|$, define the feasible set:
$$
\pmF = \{D \subseteq \pmU : D \text{ forms a perfect matching in }\pmG \}.
$$
For the \minnorm\ version, we are also given a monotone symmetric norm $f$ and a value vector $\boldv \in \mathbb{R}_{\geq 0}^{\pmU}$. The goal is to select $D \in \pmF$ that minimizes $f(\boldv\assigned{D})$.
\end{definition}
We define the \LBOmatch\ problem with $\pmeta = (\pmU; \pmS{1}, \ldots, \pmS{T}; \cutG; s; t)$ as the \LBO\ problem with $(\pmU; \pmS{1}, \ldots, \pmS{T}; \pmF)$.

Due to \cref{thm:equivalence}, we establish the equivalence between approximating the \minnorm\ problem and the \LBO\ problem. Hence, 
if we can show that the \LBO\ version is hard
to approximate, then the same hardness (up to constant factor) also applies to the \minnorm\ version. 

Now, we consider using the linear programming rounding approach to approximate
the \LBO\ problem with $\eta = (\calU; S_1, \ldots, S_T; \calF)$.
Such an algorithm proceeds according to 
the following pipeline:

\begin{itemize}
    \item 
 First, we formulate the natural linear program (for $c \geq 1$) as the following:

\begin{equation}
\begin{aligned}
    \min && 0 \\
    s.t. && \text{$x$ satisfies} && \text{relaxed constraints of } \calF, \\
    && \sum_{u \in S_i} x_u \leq c \cdot 2^i && \forall 1 \leq i \leq T, \\
    && x_u \geq 0 && \forall u \in \calU.
\end{aligned}
\tag{LP-LBO-Original($\eta$, $c$)}
\label{LP-LBO-Original}
\end{equation}

\item 
Then, to solve the $(c, c_0)$-\LBO\ problem, we first check whether LP-LBO-Original($\eta$, $c_0$) has feasible solutions. If feasible, we use a rounding algorithm to find an integral solution for \eqref{LP-LBO-Original}.
\end{itemize}

Since the factor $c$ in LP-LBO-Original($\eta$, $c$)
determines the approximation factor,
we study  
the following linear program.

\begin{equation}
\begin{aligned}
    \min && z \\
    s.t. && \text{$x$ satisfies} && \text{relaxed constraints of } \calF, \\
    && \sum_{u \in S_i} x_u \leq z \cdot 2^i && \forall 1 \leq i \leq T, \\
    && x_u \geq 0 && \forall u \in \calU.
\end{aligned}
\tag{LP-LBO($\eta$)}
\label{LP-LBO}
\end{equation}

Clearly, if LP-LBO-Original($\eta$, $c$)
has a feasible solution, the optimal
value of LP-LBO($\eta$) is at most $c$.
Suppose we can round 
a fractional solution LP-LBO-Original($\eta$, $c$)
to an integral feasible 
$c'$-valid solution 
$\bar{x}$ 
for some constant $c'$
(i.e., $\bar{x}$ satisfies $\calF$ and 
$\sum_{u \in S_i} \bar{x}_u \leq c' \cdot 2^i \,\, \forall 1 \leq i \leq T
$).
Then, we get an 
integral solution for 
LP-LBO($\eta$) with objective value $c'$, contradicting the fact that
the integrality gap of LP-LBO($\eta$) is $\omega(1)$.
Hence, we can conclude that 
if the integrality gap of LP-LBO($\eta$) is $\omega(1)$,
it would be difficult to derive a constant factor approximation algorithm for both $\LBO$ and \minnorm-version of the problem
using the LP LP-LBO-Original($\eta, c$).

\subsection{Reduction from Perfect Matching to $s$-$t$ Path}\label{sec:integralgap-1}
For an \LBO\ perfect matching problem with $\pmeta = (\pmU; \pmS{1}, \ldots, \pmS{T}; \pmG)$, where $\pmG = (L, R, E)$ is a bipartite graph, we consider the following LP (on the left). The LP on the right is for \LBOpath\ with $\patheta = (\pathU; \pathS{1}, \ldots, \pathS{T}; \pathG; s; t)$.

\begin{minipage}{0.4\textwidth}
\begin{equation}
    \begin{aligned}
    \min && z && \\
     \text{s.t.} && \sum_{i \in L, (i, j) \in E} x_{i, j} = 1 && \forall j \in R,\\
     && \sum_{j \in R, (i, j) \in E} x_{i, j} = 1 && \forall i \in L,\\
    && \sum_{e \in \pmS{i}} x_e \leq z \cdot 2^{i} && \forall 1 \leq i \leq T,\\
    && 0 \leq x_{i, j} \leq 1 && \forall (i, j) \in \pmU. \\
    && \vphantom{\sum_{j \in \pathV, (s, j) \in \pathU} x_{s, j} = 1}
    \end{aligned}
\tag{LP-LBO-PM($\pmeta$)}
\label{LP-LBO-PM}
\end{equation}
\end{minipage}
\hspace{0.05\textwidth}
\begin{minipage}{0.4\textwidth}
\begin{equation}
    \begin{aligned}
    \min && z && \\
    \text{s.t.} && \sum_{j \in \pathV, (i, j) \in \pathU} x_{i, j} = 0 && \forall i \in \pathV \setminus \{s, t\},\\
    && \sum_{j \in \pathV, (s, j) \in \pathU} x_{s, j} = 1 && \\
    && \sum_{i \in \pathV, (i, t) \in \pathU} x_{i, t} = 1 && \\
    && \sum_{e \in S_i} x_e \leq z \cdot 2^{i} && \forall 1 \leq i \leq T,\\
    && 0 \leq x_e \leq 1 && \forall e \in \pathU.
    \end{aligned}
\tag{LP-LBO-Path($\patheta$)}
\label{LP-LBO-Path}
\end{equation}
\end{minipage}


We have the following theorem showing that \LBOpath\ problem is not harder than \LBOmatch\ problem.

\begin{theorem}\label{thm:intgap-0}
For any arbitrary function $\alpha(n)\ge 1$, We have the following conclusions:
\begin{enumerate}[label=(\alph*), format=\normalfont]
\item If the integrality gap of \cref{LP-LBO-PM} is no more than $\alpha(n)$ for all instances 
$\pmeta$ of the \LBOmatch\ problem, 
then the integrality gap of \cref{LP-LBO-Path} is also $O(\alpha(n))$ for all instances of 
the \LBOpath\ problem.

\item If we have a polynomial-time $\alpha(n)$-approximation algorithm for the \minnormmatch\ problem, then we have a polynomial-time $O(\alpha(n))$-approximation algorithm for the \minnormpath\ problem.
\end{enumerate}
\end{theorem}

\begin{proof}
For the \LBOpath\ problem with $\patheta = (\pathU; \pathS{1}, \ldots, \pathS{T}; \pathG; s; t)$, where\\
$\pathG = (\pathV, \pathU)$, we construct a \LBOmatch\ problem 
with  $\pmeta = (\pmU; \pmS{1}, \ldots, \pmS{T}; \pmG)$, where $\pmG = (L, R, \pmU)$ 
is a bipartite graph constructed by the following process:
\begin{enumerate}
\item We define $s_0 \in L$, $t_1 \in R$. For any other vertex $v \in \pathV \setminus \{s, t\}$, we define $v_0 \in L, v_1 \in R$, and add an edge $(v_0, v_1) \in \pmS{T}$ between the two vertices.

\item Also, for any $1 \leq i \leq T$ and edge $e = (u, v) \in \pathS{i}$, we add an edge $(u_0, v_1) \in \pmS{i}$.
\end{enumerate}
For any fractional solution $(x', z)$ of LP-LBO-PM($\pmeta$), we can construct a fractional solution $(x, z)$ for LP-LBO-Path($\patheta$) by setting $x_{(u, v)} = x'_{(u_0, v_1)}$ for all $(u, v) \in \pathU$ by the fact that the degree of $v_0$ and $v_1$ is the same in perfect matching for each $v$.

On the other hand, each path from $s$ to $t$ corresponds to a perfect matching in $\pmG$. If we have a $c$-valid solution $D$ for \LBOpath, then we set $x^{int}_{u_0, v_1} = 1$ if and only if $(u_0, v_1) \in D$ or $u = v$ and $u$ is not in the path. $(x^{int}, c)$ is an integral solution for LP-LBO-PM($\pmeta$).

Due to \cref{thm:equivalence}, this implies statement (b). Combining the two parts implies that the integrality gap of LP-LBO-PM($\pmeta$) is no less than the integrality gap of LP-LBO-Path($\patheta$), so statement (a) is true.
\end{proof}

\subsection{Integrality Gaps for Min-Norm $s$-$t$ Path and Min-Norm Perfect Matching}  \label{sec:integralgap-2}

\begin{theorem} \label{thm:intgap-1}
    For infinitely many $n$, there exists an instance $\cuteta$ of size $n$ such that the integrality gap of \cref{LP-LBO-Path} can be $\Omega(\log n)$.
    Thus, the relaxations of the natural linear programming of both \LBOpath\ and \LBOmatch\ have integrality gaps of $\Omega(\log n)$.
\end{theorem}

\begin{proof}
Consider two positive integers $c,k$ such that $c=O(1),k=\Omega(\log n)$ and $1\leq c,k\leq T$. We construct an example
$\patheta = (\pathU; \pathS{1}, \ldots, \pathS{T}; \pathG; s; t)$, and demonstrate that the integrality gap of LP-LBO-Path($\patheta$) is $\Omega(\log n)$.
The construction of the graph is as follows: for each $1 \leq i \leq k$, there is a path $P_i$ with $c \times 2^i$ edges in $\pathS{i}$ from $s$ to $t$. Additionally, there is another path $P$ with $2^i$ edges in $\pathS{i}$ for each $1 \leq i \leq k$.
Then we have $k = \Theta(\log n / c)$. The optimal integral solution chooses $P$. However, there is a feasible fractional solution where $x_e = 1/k$ for $e \in P_i,1\leq i\leq k$. Thus, the integrality gap is at least $k / c = \Omega(\log n)$ when $c = O(1)$.
\end{proof}
\input{Example4}

\subsection{Integrality Gaps for Min-Norm $s$-$t$ Cut} 
Also, we have the result about the integrality gap for $s$-$t$ cut.
\begin{definition}[Min-Norm $s$-$t$ Cut Problem (\minnormcut)]
Given a directed graph $\cutG = (\cutV, \cutU)$ and nodes $s, t \in \cutV$, define the feasible set:
$$
\cutF = \{D \subseteq \cutU : \text{There is no path from } s \text{ to } t \text{ in the graph } G' = (\cutV, \cutU \setminus D)\}.
$$
For the \minnorm\ version, we are also given a monotone symmetric norm $f$ and a value vector $\boldv \in \mathbb{R}_{\geq 0}^{\cutU}$. The goal is to select an $s$-$t$ cut $D \in \cutF$ that minimizes $f(\boldv\assigned{D})$.
\end{definition}

Similarly, we define the \LBOcut\ problem with $\cuteta = (\cutU; \cutS{1}, \ldots, \cutS{T}; \cutG; s; t)$ as the \LBO\ problem with $(\cutU; \cutS{1}, \ldots, \cutS{T}; \cutF)$, where $\cutG = (\cutV, \cutU)$ is a directed graph, we can write the following LP:

\begin{equation}
\begin{aligned}
    \min && z && \\
    \text{s.t.} && 0 \leq p_v \leq 1 && \forall v \in \cutV,\\
    && x_{u, v} \geq p_u - p_v && \forall (u, v) \in \cutU,\\
    && \sum_{(u, v) \in \cutS{i}} x_{u, v} \leq z \cdot 2^{i} && \forall 1 \leq i \leq T,\\
    && p_s = 0 && \\
    && p_t = 1 && \\
    && 0 \leq x_{u, v} \leq 1 && \forall (u, v) \in \cutU.
\end{aligned}
\tag{LP-LBO-Cut($\cuteta$)}
\label{LP-LBO-Cut}
\end{equation}

\begin{theorem}\label{thm:intgap-2}
For infinitely many $n$, there exists an instance $\cuteta$ of size $n$ such that the integrality gap of \cref{LP-LBO-Cut} can be $\Omega(\log n)$. 
\end{theorem}

\begin{proof}
We use a construction $\cuteta = (\cutU; \cutS{1}, \ldots, \cutS{T}; \cutG; s; t)$ similar to the one used for \LBOpath. We select positive integers $k$ and $c$.

Let node $A_0 = s$. For $1 \leq i \leq k$, there is a node $A_i$. There are $c \cdot 2^{i-1}$ edges in $\cutS{i}$ between $A_{i-1}$ and $A_i$. Additionally, there are $2^{i-1}$ edges in $\cutS{i}$ for all $1 \leq i \leq k$ between node $A_k$ and $t$. Let $A_{k+1} = t$. 

To create a valid cut, there must exist some $0 \leq i< k$ such that all edges between $A_i$ and $A_{i+1}$ have been selected. For each $i$, all edges between $A_i$ and $A_{i+1}$ form a valid solution. Thus, similar to the integrality gap of LP-LBO-Path($\cdot$), this example shows that the integrality gap is $\Omega(\log n)$.
\end{proof}

\emph{Remark.} In the example in the proof,
the gap between any feasible subset of $\calU'=\{e\in\calU:x_e>0\}$ ($\{x_e\}_{e\in\calU}$ is the fractional solution) and the fractional solution is larger than any given constant. Thus any rounding algorithm that deletes zero-value variables (including the rounding algorithm  we developed in this paper and the iterative rounding method in \cite{chakrabarty2019approximation}) 
cannot successfully yield a constant-factor approximation.


%% file: Example4.tex
\begin{figure}[t]
\begin{minipage}{\textwidth}
\centering
\begin{tikzpicture}[  
    point/.style={circle, draw, fill=white, inner sep=2.5pt},  
    node distance=1cm and 1cm, 
]  
	
    \node[circle,draw,fill=white,inner sep=3.5pt] (s) at (0,2) {\large \bf $s$};
	\node[circle,draw,fill=white,inner sep=3.5pt] (t) at (15,2) {\large \bf $t$};
	\foreach \i in {1,...,15} {  
        \node[point] (p1\i) at (1+13/16*\i, 3) {};  
    }  
	\draw[->,dashed,line width=1.5pt,>={Stealth}] (s) -- (p11);
	\draw[->,dashed,line width=1.5pt,>={Stealth}] (p115) -- (t);
    \foreach \i in {1,...,14} {  
        \draw[->, dashed, line width=1.5pt,>={Stealth}] (p1\i) -- (p1\the\numexpr\i+1\relax);  
    }  
      
    \foreach \i in {1,...,7} {  
        \node[point] (p2\i) at (1+13/8*\i, 2.5) {};  
    }  
	\draw[->,line width=1.5pt,>={Stealth}] (s) -- (p21);
	\draw[->,line width=1.5pt,>={Stealth}] (p27) -- (t);
    \foreach \i in {1,...,6} {  
        \draw[->,line width=1.5pt,>={Stealth}] (p2\i) -- (p2\the\numexpr\i+1\relax);  
    }  
      
    \foreach \i in {1,...,3} {  
        \node[point] (p3\i) at (1+13/4*\i, 2) {};  
    }  
	\draw[->,color=red, line width=1.5pt,>={Stealth}] (s) -- (p31);
	\draw[->,color=red, line width=1.5pt,>={Stealth}] (p33) -- (t);
    \foreach \i in {1,...,2} {  
        \draw[->, color=red, line width=1.5pt,>={Stealth}] (p3\i) -- (p3\the\numexpr\i+1\relax);  
    } 

    \foreach \i in {1,...,13} {  
        \node[point] (p4\i) at (1+13/14*\i, 1.5) {};  
    } 
	\draw[->,dashed, line width=1.5pt,>={Stealth}] (s) -- (p41);
	\draw[->,color=red, line width=1.5pt,>={Stealth}] (p413) -- (t);

    \foreach \i in {1,...,7} {  
        \draw[->, dashed, line width=1.5pt,>={Stealth}] (p4\i) -- (p4\the\numexpr\i+1\relax);  
    } 
	\foreach \i in {8,...,11} {
		\draw[->, line width=1.5pt,>={Stealth}] (p4\i) -- (p4\the\numexpr\i+1\relax);  
	}
	\foreach \i in {12,...,12} {
		\draw[->,color=red, line width=1.5pt,>={Stealth}] (p4\i) -- (p4\the\numexpr\i+1\relax);
	}
\end{tikzpicture}
\caption{An example for $c = 2, k = 3$. The red solid edges are in $\pathS{1}$. The black solid edges are in $\pathS{2}$. The dashed edges are in $\pathS{3}$.}
\end{minipage}
\end{figure}
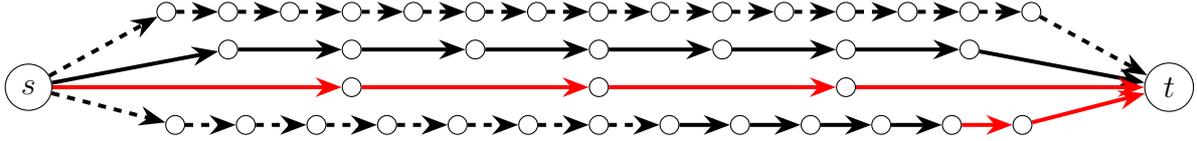

%% file: shortestpath.tex
\section{An Algorithm for Min-Norm $s$-$t$ Path Problem}
\label{sec:algo-path}
Recall that we define the \minnormpath\ problem and prove that the natural linear program has a large integrality gap in \cref{sec:integralgap}. 
In this section, 
we provide a factor $\alpha$ approximation algorithm 
that runs in $n^{O(\log\log n/\alpha)}$ time,
for any $9\le \alpha\le \log\log n$. 
In particular, this implies an $O(\log\log n)$-factor polynomial-time approximation algorithm and a constant-factor quasi-polynomial $n^{O(\log\log n)}$-time algorithm for \minnormpath.
Note that this does not contradict 
Theorem~\ref{thm:intgap-1}, since we do not use the LP 
rounding approach in this section.

In light of \cref{thm:equivalence} with $\varepsilon=1$, we consider $\frac{\alpha-1}{4}$-\LBOpath\ problem,
with input tuple $\eta=(\calU;S_{1},S_{2},\ldots,S_{T};G;s;t)$, where $n=|\calU|$ and $T=\lceil \log n\rceil$, where 
$\alpha$ is the approximation factor we aim to achieve.

We first provide an overview of our main ideas.
A natural approach to solve \LBOpath\ is to employ dynamic programming, 
in which the states keep track of the number of edges used from each group.
However, since we have $T=\lceil \log n\rceil$ groups, the number of states
may be as large as $n^{O(T)}$.
To resolve this issue, 
we perform an {\em approximation dynamic programming}, in which we only approximate the number of edges in each group. In particular, the numbers are rounded to the nearest power of \( p \) after each step, for carefully chosen value \( p > 1 \) 
(to ensure that we do not lose too much in the rounding).
This rounding technique is inspired by a classic approximation algorithm for the subset sum problem~\cite{10.1145/321906.321909}
Now, the dynamic programming state is a vector that approximates the number of edges used from each group in a path from \( x \) to \( y \) with at most \( 2^i \) edges, where \( x, y \in V \) and \( 0 \leq i \leq \lceil \log n \rceil \).
The dynamic programming process involves storing the number of selected items in each group and rounding them to the nearest power of \( p \) at each iteration. However, this method results in a state space of size \( (\log_p n)^{T} = n^{\Omega(\log \log n)} \), 
which is better than $n^{O(T)}$, but still super-polynomial. 
To resolve the above issue, we need the second idea, which is to trade off the approximation factor and the running time. 
In particular, we introduce an integer parameter \( \beta \), defined as $\beta = \frac{\alpha - 1}{4(1 + \delta)}$
for some \( \delta \in \left[ \frac{1}{2}, 2 \right] \). We then partition the original \( T \) groups into \( T/\beta \) {\em supergroups},
each containing \(\beta\) groups.
This reduces the number of states to $(\log_p n)^{O(T/\beta)}$, but incurs a loss of 
$O(\beta)$ in the approximation factor. Our main result in this section is the following theorem:

\begin{theorem}\label{thm:algo-path}
For any $9\le \alpha\le \log\log n$, there exists a $n^{O(\log\log n/\alpha)}$-time algorithm (\cref{Algo-Path}) for $\frac{\alpha-1}{4}$-\LBOpath. Thus we have an approximation approximation which runs in $n^{O(\log\log n/\alpha)}$-time 
and achieves an approximation factor of $\alpha$ for \minnormpath.
\end{theorem}

Now, we present the details of our algorithm.
Let $K=\lceil T/\beta \rceil$, and $B_i=\min (T,i\cdot \beta)$ for all $0\le i\le K$. 

For $1\le i\le K$ and $D\subseteq \calU$, define $C_i(D)=\sum_{j=B_{i-1}+1}^{B_i} \frac{1}{2^j} |D\cap S_{j}|$ (specifically,  $C_i(\emptyset)=0$). Furthermore, we define the vector $C(D)=(C_1(D),\ldots C_{K}(D))\in \mathbb{R}^{K}$. 
It is important to notice that:
\begin{itemize}
    \item If $D\subseteq\calU$ is a $c$-valid solution ($c>0$), then $C_i(D)\le c\beta$ for all $1\le i\le K$.
    \item If $C_i(D)\le c\beta$ for all $1\le i\le K$, then $D\subseteq\calU$ is a $c\beta$-valid solution.
\end{itemize}

In iteration \( i \) (\( 1 \leq i \leq \lceil \log n \rceil \)), for each pair of vertices \( x, y \in V \), we define \( Q_{i, x, y} \) as a set of vectors in \( \mathbb{R}^{K} \) that encodes information about paths from \( x \) to \( y \) containing at most \( 2^i \) edges. Specifically, for any path \( D \subseteq \mathcal{U} \) from \( x \) to \( y \) with at most \( 2^i \) edges, the set \( Q_{i, x, y} \) includes a corresponding vector that approximates \( C(D) \).

\begin{itemize}
\item Initially, $Q_{0, x, x} = \{ C(\emptyset) \}$. For $Q_{0, x, y}$, it is set to $\{ C(\{(x, y)\}) \}$ if $(x, y)$ is an edge, and $\emptyset$ if $(x, y)$ is not an edge.

\item In the $i$-th iteration ($1 \leq i \leq \lceil \log n \rceil$), we begin by initializing $Q_{i, x, y} = \emptyset$ for all $x, y$. Then, for each pair $x, y$, we enumerate all vertices $z$ and add the sum of $Q_{i-1, x, z}$ and $Q_{i-1, z, y}$ to $Q_{i, x, y}$. Here, the sum of two sets is defined as the set of all pairwise sums of elements from the two sets.
 
\item To reduce the size of $Q_{i, x, y}$ , we round the components of these vectors in $Q_{i,x,y}$ to $0$ or powers of $p = 1 + \frac{\delta/2}{\log n}$ (recall that $\delta\in \lbrack 1/2,2\rbrack$).

\end{itemize}
\begin{algorithm}[H]
\caption{An Algorithm for \LBOpath}
\label{Algo-Path}
\KwOut{A $(\alpha-1)/4$-valid Solution $D$ or "No Solution"}
Initially, set all $Q_{i,x,y}$ to $\emptyset$.\;
$\forall x\in V$, $Q_{0,x,x}\gets Q_{0,x,x}\cup \lbrace (C(\emptyset) \rbrace$. \;
$\forall$ $(x,y)\in \calU$, $Q_{0,x,y}=Q_{0,x,y}\cup \lbrace C(\lbrace (x,y)\rbrace)\rbrace$.\; 
\For{$i \gets 1$ \KwTo $\lceil \log n\rceil$}{
    \ForAll{$x,y,z\in V,u\in Q_{i-1,x,z},u'\in Q_{i-1,z,y}$}{
            Add $u$ and $u'$ up, and round it to the power of $p$. (Recall that $p=(1+\frac{\delta/2}{\log n})$). Formally, define $w\in \mathbb{R}^{K}$ by $w_j=p^{\lceil \log_p (u_j+u'_j)\rceil}$ for $1\le j\le K$ (Specially, if $u_j+u'_j=0$, then $w_j=0$.)\;
            \If{$\forall 1\le j\le K, w_j\le p^{i}\beta$}{
                $Q_{i,x,y}\gets Q_{i,x,y}\cup \lbrace w\rbrace$\;
            }
    }
}
\If{$Q_{\lceil \log n\rceil,s,t}=\emptyset$}{
\Return "No Solution"\;
}
Select $u\in Q_{\lceil \log n\rceil,s,t}$\;
\Return $D$
\end{algorithm}

\begin{lemma}\label{lem:algo-path-1}
If there exists a $1$-valid solution, then \cref{Algo-Path} finds an $(\alpha-1)/4$-valid solution. 
\end{lemma}
\begin{proof}
Assume $D$ is a $1$-valid solution, then for each $1\le i\le K$, $C_i(D)=\sum_{j=B_{i-1}+1}^{B_i} \frac{1}{2^j}|D\cap S_{j}|\le \sum_{j=B_{i-1}+1}^{B_i} \frac{1}{2^j} 2^j=B_i-B_{i-1}\le \beta$. 

We use mathematical induction. For a subpath $D'\subseteq D$ from $x$ to $y$ with no more than $2^i$ edges, we need to prove that there exists a vector $w\in Q_{i,x,y}$, $w_j\le p^i C_j(D')$ for $1\le j\le K$.

It is obvious for $i=0$. When $i>0$, by induction hypothesis, there exists $z\in V, u\in Q_{i-1,x,z},u'\in Q_{i-1,z,y}$ that $D'$ is a combination of two subpaths corresponding to $u,u'$. 
Let the two paths be $D'_1,D'_2$ such that $u_j\le p^{i-1}C_j(D'_1),u'_j\le p^{i-1}C_j(D'_2)$ for $1\le j\le K$.
By the algorithm, there is $w$ such that
$$w_j=p^{\lceil \log_p(u_j+u'_j)\rceil}\le p(u_j+u'_j)\le p(p^{i-1}C_j(D'_1)+p^{i-1}C_j(D'_2))=p^i C_j(D')\le p^i\beta,$$
so it can be added in $Q_{i,x,y}$ and the assumption of the induction holds.

Since the length of $D$ is no more than $n$, we can find $w\in Q_{\lceil \log n\rceil,s,t}$ and for $1\le j\le K$, 
$$w_j\le p^{\lceil \log n\rceil} C_j(D)\le \left( 1+\frac{\delta/2}{\log n}\right)^{\log n+1} \beta\le e^{\delta/2}\cdot \left( 1+\frac{\delta/2}{\log n}\right)\beta\le (1+\delta)\beta=\frac{\alpha-1}{4}.$$
When we repeat the algorithm process in reverse, we can find a $\frac{\alpha-1}{4}$-valid solution. Therefore, this statement is correct.
\end{proof}

\begin{proofofthm}{\ref{thm:algo-path}}
The values of vectors in any $Q_{i,x,y}$ are rounded to 0 or a power of $p=(1+\frac{1}{\log n})$.
Also, at step $i$, any value $w_j$ is at most
$$2^i C_j(D)\leq 2^i \beta\leq 2^i\log\log n,$$
and if it is larger than 0, it is at least $1/2^T$.
Therefore, for any $1\le i\le \lceil \log n\rceil,x,y\in V, 1\le j\le K$, the number of different values of $w_j$ in $Q_{i,x,y}$ is no more than $2+\log_p\left(\frac{p^i\log\log n}{1/2^{T}}\right)=2+\frac{T}{\log (1+(\delta/2)/\log n)}+i+\frac{\log \log\log n}{\log(1+(\delta/2)/\log n)}\le \frac{1}{\delta}\log^2 n$.

Therefore,
\begin{align*}
     |Q_{i,x,y}|& \le \left(\frac{1}{\delta}\log^2 n\right)^{K}=\left(\frac{1}{\delta}\log^2 n\right)^{O\left(\frac{\log n}{\beta}\right)}\\
    & =\exp\left(O\left(\frac{\log(1/\delta)}{\beta}\cdot \log\log n\cdot \log n\right)\right)
    =n^{O\left(\frac{\log\log n}{\alpha}\right)}
\end{align*}
Therefore, this algorithm is $n^{O(\log\log n/\alpha)}$-time. According to \cref{lem:algo-path-1} and \cref{thm:equivalence}, we find a $\alpha$-approximation $n^{O(\log\log n/\alpha)}$-time algorithm for \minnormpath.
\end{proofofthm}



%% file: matching.tex
\section{A Bi-criterion Approximation for Matching}
\label{sec:matching}
In this section, we consider the matching problem. 
While we do not know how to design a constant factor approximation algorithm for the \minnorm\ version of 
perfect matching problem yet, we demonstrate that it is possible to find a nearly perfect matching within a constant approximation factor --  a
bi-criterion approximation algorithm.
In particular, for any given norm, we can find a matching that matches $1-\epsilon$
fraction of nodes and its corresponding norm is at most $c$ (a constant) times the norm of the optimal integral perfect matching.
Note that a constant factor 
approximation algorithm (without relaxing the perfect matching requirement) is impossible using the natural linear program, due to Theorem~\ref{thm:intgap-1}.

First, we introduce some notations.
We are given a bipartite graph $G=(L,R,E)$, where $L$ is the  set of nodes of the first color class, $R$ is the set of node of the other color class, and $E$ is the set of edges. Let $m=|L|=|R|$. We define $\calU=E, n=|\calU|$. 
We also study the \LBO\ version and we let
$S_1,S_2\cdots,S_T$ be the disjoint subsets of $\calU$.


We consider two relaxations of the perfect matching requirement.

\begin{definition}[$\epsilon$-relaxed Matching]
Let $0<\epsilon<1$. We define the following problem as $\epsilon$-relaxed matching. Given a bipartite graph $G=(L,R,E)$, a set $S\subseteq E$ is a relaxed matching if in the subgraph $G'=(L,R,S)$, the degree of each vertex is 1 or 2 and the number of vertices with degree 2 is at most $2\epsilon m$ (recall $m=|L|=|R|$).
\end{definition}

\begin{definition}[$\epsilon$-nearly Matching]
Let $0<\epsilon<1$. We define the following problem as $\epsilon$-nearly matching. Given a bipartite graph $G=(L,R,E)$, a set $S\subseteq E$ is a nearly matching if it is a matching with at least $(1-\epsilon)m$ edges.
\end{definition}

Let $T=\lceil \log n\rceil$. Let $\boldd\in\R_{\geq 0}^T$ be a vector. We define (LP-Perfect-Matching($\boldd$)) as follows:
\begin{equation}
    \begin{aligned}
    \min && 0 && \\
     s.t. && \sum_{i\in L, (i,j)\in E} x_{i,j}=1 && \forall j\in R\\
     && \sum_{j\in R, (i,j)\in E} x_{i,j}=1 && \forall i\in L\\
    &&\sum_{e\in S_i} x_e\leq d_i && \forall 1\leq i\leq T\\
    && 0\leq x_e\leq 1 && \forall e\in\calU
\end{aligned}
\tag{LP-Perfect-Matching($\boldd$)}
\label{LP-Perfect-M}
\end{equation}

We call the first and second line of constraints {\em degree constraints} and the third line {\em budget constraints}. For any of the problems above, a solution $D\subseteq E$ is called $c$-valid for a constant $c$ if $|D\cap S_i|\leq c\times 2^i$ for $1\leq i\leq T$.

Let $\boldd^*=(2^1,2^2,\cdots,2^T)$. A 1-valid solution for perfect matching problem is an integral solution for LP-Perfect-Matching($\boldd^{*}$).

Here is a rounding process for $\epsilon$-relaxed matching.

\begin{lemma} \label{lem:matching-2}
For any $\boldd\in \mathbb{R}^{T}_{\ge 0}$, if \ref{LP-Perfect-M} has feasible fractional solutions, then there exists an polynomial-time rounding algorithm to find an integral solution of $\epsilon$-relaxed matching, and the integral solution satisfies the following conditions for each $i$:

\begin{itemize}

\item If $d_i>0$, the number of elements chosen from $S_i$ is at most $2d_i+9/\epsilon$. 

\item If $d_i=0$, the number of elements chosen from $S_i$ is 0.
\end{itemize}
\end{lemma}
\begin{proof}
At first, we have a fractional solution by solving the linear programming.
Then, we use the iterative rounding. At each step, if a new variable takes a value of 1 or 0, we fix it, remove it from the system, and proceed to the next rounding step. 
(If the value is 1, we subtract 1 from the corresponding right-hand side constraints.)

At one step, if there are no budget constraints, the coefficient matrix of tight constraints form a union of two laminar families. Thus, by \cref{theorem:laminar:intersection} and \cref{theorem:TUM}, all $x_e$'s are integers. So the rounding process has terminated. Otherwise, we can assume that there is at least one budget constraint.

Next, we consider the variables larger than $1/2$. 
If a degree constraint has only two nonzero variables, we label it as a {\em bad constraint}, while the other degree constraints are labeled as {\em good constraints}.
We then focus on a subgraph $G$ of the original graph, where an edge appears in this subgraph if and only if it connects two nodes and each node corresponds to a {\em bad constraint}. Since each node in this subgraph has a degree of at most 2, the graph is composed of cycles and paths.

If there is a cycle, we examine the variables corresponding to the edges. There are only at most $2$ different values for these variables, with one being at least $1/2$. We round the variables with the larger value to 1 and set those with the smaller value to 0. We then delete the cycle. For any budget constraint of the form "$\cdots \leq d_i$," we adjust $d_i$ by subtracting the contribution of the variables on the cycle, ensuring that the budget constraint remains valid.

Assuming there is no cycle, if there is a path of length at least $2/\epsilon + 5$, we focus on this path. The variables on the path can also take only two values, with one being at least $1/2$. Let the path include a subpath $u_1,v_1,u_2,v_2,\ldots,u_k,v_k$, where $k \geq 1/\epsilon + 1$. Here, $(u_i,v_i)$ corresponds to a variable that is at least $1/2$. We round each variable corresponding to $(u_i,v_i)$ to 1 and each variable corresponding to $(v_i,u_{i+1})$ to 0. We delete the degree constraints corresponding to the $2k$ nodes. Additionally, for each budget constraint of the form "$\cdots \leq d_i$," we adjust $d_i$ by subtracting the contribution of the variables for these $2k-1$ edges, thereby ensuring the budget constraint holds. Among the $2k$ deleted nodes, only $u_1$ and $v_k$ may have a degree of 2; the others must have a degree of 1.

Next, we assume that there is no cycle and no path of length at least $2/\epsilon + 5$. For any path in $G$, each of the two ending nodes has a degree of 2. We extend the paths by adding one edge at each ending node. By definition, these new paths do not share common edges. Now, for the new subgraph $H$, consider the variables not included in it. Each of these variables allocates $1/3$ of a token to each of the degree constraints they appear in, and they allocate the remaining token to the budget constraint they appear in. For the paths in $H$, let them be $u_1,u_2,\ldots ,u_k$, where $k \leq 3/\epsilon$. Let $(u_i,u_{i+1})$ allocate $1/3$ of a token to $u_i$ and $2/3$ to $u_{i+1}$ for $i < k-1$, while $(u_{k-1},u_k)$ allocates $1/3$ to both $u_{k-1}$ and $u_k$. Additionally, $(u_{k-1},u_k)$ allocates at least $1/(3k)$ to the budget constraints for each variable in this path appears in. Consequently, each degree constraint has a total token of at least 1, and there must be one budget constraint with a total token of at most 1. For this budget constraint, each variable with a value less than $1/2$ contributes a token of at least $\epsilon/9$. Therefore, for this budget constraint,
$$\sum_{e\in S_i:x_e>0}(1-2x_e)\leq 9/\epsilon.$$

Thus, we can eliminate this constraint. By applying the same analysis, at the conclusion of the iterative rounding process, for any budget constraint that has been eliminated, let the right hand side when it is deleted be $d_i'\leq d_i$. So
$$\sum_{e\in S_i:x_e>0}1=\sum_{e\in S_i:x_e>0}(1-2x_e)+2\sum_{e\in S_i:x_e>0}x_e\leq 2d_i'+9/\epsilon.$$
Therefore, the total number of elements chosen in $S_i$ at last is at most $2d_i'+9/\epsilon+(d_i-d_i')\leq 2d_i+9/\epsilon$. Also, if $d_i=0$, the elements in $S_i$ cannot be chosen in the process because $x_e=0$ for $e\in S_i$. So the lemma follows.
\end{proof}

Similarly, given $\epsilon$, if we partially enumerate solutions from $S_i$ for constant number of $i$ (from 1 to some constant $k$), we can let the right hand side be not larger than $(2+\delta)d_i$. This gives the following theorem:

\begin{theorem}
\label{thm:relaxed}
If there exists a 1-valid solution for a \LBOmatch\ problem, then there exists an $O(n^{\frac{18}{\delta\epsilon}})$-time algorithm to obtain a $(2+\delta)$-valid solution for the corresponding $\epsilon$-relaxed matching problem for any $\delta,\epsilon<1$.
\end{theorem}

\begin{proof}
Let $k=\floor{\log (\frac{9}{\delta\epsilon})}$. We partially enumerate the partial solutions from $S_1$ to $S_k$ to get a set $X$ consisting of all partial solutions $(D_1,D_2,\cdots,D_k)$ such that $|D_i|\leq 2^i$ and $D_i\subseteq S_i$ for $1\leq i\leq k$. As the number of choices of $D_i$ is at most
$$\sum_{j=0}^{2^i}\binom{|S_i|}{j}\leq \sum_{j=0}^{2^i}n^j\leq (n+1)^{2^j},$$
we have
$$|X|\leq \prod_{i=1}^{k}(n+1)^{2^i}\leq (n+1)^{2^{k+1}}=O(n^{\frac{18}{\delta\epsilon}}).$$
For each $(D_1,D_2,\cdots,D_k)\in X$, we run the following algorithm. If one node corresponds to at least two edges in $D_1\cup D_2\cdots\cup D_k$, we return "No Solution". 
Then if the edges in $(D_1,D_2,\cdots,D_k)$ form a matching, we delete the edges and corresponding nodes in the graph. Then we just need to consider $S_{k+1},S_{k+2},\cdots,S_T$. 
We check whether \ref{LP-Perfect-M} has feasible fractional solutions for $\boldd=(0,0,\cdots,0,2^{k+1},2^{k+2},\cdots,2^T)$. 
If there is no solution, we return "No solution". Otherwise, by Lemma~\ref{lem:matching-2}, we can round the solution such that the number of elements chosen in $S_i$ is at most $2d_i+\frac{9}{\delta\epsilon}$ for $k+1\leq i\leq T$, where $d_i=2^i$. 
By $d_i\geq \frac{9}{\delta\epsilon}$, we have $\frac{9}{\epsilon}+2d_i\leq (2+\delta)d_i$. If the perfect matching has a 1-valid solution, the corresponding partial solution (for $1\leq i\leq k$, $D_i$ is the set of elements chosen in $S_i$ in the 1-valid solution for perfect matching) won't give return "No Solution".
Therefore, the theorem follows.
\end{proof}

Notice that for the solution we got in Theorem~\ref{thm:relaxed}, as only $2\epsilon m$ number of nodes are of degree 2, we just delete one edge from each of them. 
We can get a matching with at least $(1-\epsilon)m$ edges (for an vertex in $L$ with no edge left, it corresponds to at least one vertex in $R$ with 2 original edges, and the number of such vertices in $R$ is at most $\epsilon m$). 
By deleting edges, the number of edges in $S_i$ for each $i$ does not increase. So this gives the following theorem from Theorem~\ref{thm:relaxed}:

\begin{theorem} \label{thm:matching-i3}
If there exists a 1-valid solution for a \LBOmatch\ problem, then there exists an $O(n^{\frac{18}{\delta\epsilon}})$-time algorithm to obtain a $(2+\delta)$-valid solution for the corresponding \LBO\ $\epsilon$-nearly matching problem for any $\delta,\epsilon<1$.
\end{theorem}

Now we consider the \minnorm\ version of these problems. By Theorem~\ref{thm:diffequiv}, we can get:

\begin{theorem}
\label{thm:normrelaxmatching}
Given a \minnormmatch\ problem with a monotone symmetric 
norm $f(\cdot)$ and a value vector $\boldv$, let $D^{*}$ be an optimal solution for this problem.
For any constants $\epsilon,\delta<1$, there exists a polynomial-time algorithm to obtain a $\epsilon$-relaxed matching $D$ such that $\frac{f(\boldv\assigned{D})}{f(\boldv\assigned{D^{*}})}\leq (8+\delta)$.
\end{theorem}

\begin{theorem}
\label{matching-i5}
Given a \minnormmatch\ problem with a monotone symmetric 
norm $f(\cdot)$ and a value vector $\boldv$, let $D^{*}$ be an optimal solution for this problem.
For any constants $\epsilon,\delta<1$, there exists a polynomial-time algorithm to obtain a $\epsilon$-nearly matching $D$ such that $\frac{f(\boldv\assigned{D})}{f(\boldv\assigned{D^{*}})}\leq (8+\delta)$.
\end{theorem}

%% file: conclusion.tex
\section{Concluding Remarks}

In this paper, we propose a general formulation for general norm minimization in combinatorial optimization. Our formulation captures a broad class of combinatorial structures, encompassing various fundamental problems in discrete optimization. Via a reduction of the norm minimization problem to a multi-criteria optimization problem with logarithmic budget constraints, we develop constant-factor approximation algorithms for multiple important covering problems, such as interval cover, multi-dimensional knapsack cover, and set cover (with logarithmic approximation factors). We also provide a bi-criteria approximation algorithm for min-norm perfect matching, and an $O(\log\log n)$-approximation algorithm
for the min-norm $s$-$t$ path problem, via a nontrivial approximate dynamic programming approach.

Our results open several intriguing directions for future research. First, one can explore other combinatorial optimization problems, such as Steiner trees and other network design problems, within our general framework. Additionally, our formulation could be extended to encompass the min-norm load balancing problem studied in \cite{chakrabarty2019approximation} (where job processing times are first summed into machine loads before applying a norm), and even the generalized load balancing \cite{deng2022generalized} and cascaded norm clustering problems \cite{chlamtavc2022approximating,abbasi2023parameterized} (which allow for two levels of cost aggregation via norms).
Second, obtaining a nontrivial true approximation algorithm for perfect matching -- rather than a bi-criterion approximation -- remains an important open problem. Third, it is an important open problem whether a polynomial-time constant-factor approximation exists for the min-norm $s$-$t$ path problem. Lastly, it would be interesting to study other general objective functions beyond symmetric monotone norms and submodular functions (such as general subadditive functions \cite{GNS17} and those studied in \cite{li2011generalized}).

%% file: appendix.tex
\newpage
\appendix
\section{Additional Preliminaries}
We need the following version of Chernoff bounds.

\begin{lemma}\label{lemma:chernoff}
\emph{(Chernoff bounds (see, e.g., \cite{mitzenmacher2005probability})).}
Let $X_1,\dots,X_n$ be independent Bernoulli variables with $\E[X_i]=p_i$. 
Let $X=\sum_{i=1}^nX_i$ and $\mu=\E[X]=\sum_{i=1}^np_i$. 
For $\nu\geq6\mu$, one has
\(
\Pr[X\geq\nu]\leq 2^{-\nu}
\).
\end{lemma}

\section{\topdash{\ell} Norm Optimization} \label{sec:topl-norm-opt}

\paragraph{Minimization Problems.}
We first consider combinatorial optimization problems for which 
minimizing a linear objective is poly-time solvable.
We use $\myproblem$ to denote the ordinary combinatorial optimization problem
under consideration.
In particular, we assume that there is a poly-time algorithm $\Algo$ 
that solves the original min-sum optimization problem $\myproblem$:
$\min_{S\in \calF} \boldv(S)$, 
where we write $\boldv(S)=\sum_{e\in S}v_e$.
The idea of the following theorem is not new and has been used in previous papers
for specific combinatorial problems 
(e.g., \cite{byrka2018constant,maalouly2022exact}).
However, to the best of our knowledge, the exact formulation of the following theorem 
in this general form
has not appeared before and we think it is worth recording.

\begin{theorem} \label{thm:topl-1}
If the minimization problem $\myproblem$ can be solved in poly-time, 
there is a poly-time algorithm for solving the top-$\ell$ {\em minimization} problem optimally.
\end{theorem}

\begin{proof}
The idea has already appeared implicitly or explicitly in several prior work \cite{byrka2018constant,maalouly2022exact}. 
We first guess the $\ell$-th largest weight $t$.
Define the truncated value vector $\boldv^t$ as follows:
$v^t_e = \max \{v_e-t, 0\}$.
We apply algorithm $\Algo$ to the same instance except that we use the truncated weight $\boldv^t$.
Now we show that if our guess is correct, the solution $S$ found by $\Algo$
is the optimal solution.
We denote the optimal solution by $S^\star$.
The optimality of $S$ can be seen from the following chain of inequalities:
\begin{align*}
\topp{\ell}(\boldv[S]) & \leq \topp{\ell}(\boldv^t[S]) +\ell \cdot t
\leq \boldv^t(S) +\ell \cdot t
\leq \boldv^t(S^\star) +\ell \cdot t \\
& =\topp{\ell}(\boldv^t[S^\star])+\ell \cdot t
  =\topp{\ell}(\boldv[S^\star]).
\end{align*}
This completes the proof.
\end{proof}

We can generalize the above theorem to approximation algorithms.

\begin{theorem}
If there is a poly-time approximation algorithm for the minimization problem $\myproblem$ (with approximation factor $\alpha\geq 1$), 
there is a poly-time factor $\alpha$ approximation algorithm for the top-$\ell$ {\em minimization} problem optimally.
\end{theorem}
\begin{proof}
We apply the approximation algorithm $\Algo$ to the instance with 
truncated weight $\boldv^t$.
Suppose our guess is correct.
We denote the optimal solution by $S^\star$.
The performance guarantee of $S$ can be seen from the following chain of inequalities:
\begin{align*}
\topp{\ell}(\boldv[S]) & \leq \topp{\ell}(\boldv^t[S]) +\ell \cdot t
\leq \boldv^t(S) +\ell \cdot t
\leq \alpha\boldv^t(S^\star) +\ell \cdot t \\
& =\alpha\cdot \topp{\ell}(\boldv^t[S^\star])+\ell \cdot t.
  \leq \alpha\cdot\topp{\ell}(\boldv[S^\star]).
\end{align*}
This completes the proof.
\end{proof}

\noindent
{\bf Remark:}
The above theorem does not imply a constant approximation for the top-$\ell$ minimization of $k$-median. 
The reason is that the modified edge weights do not satisfy triangle inequality
any more, which is required for known constant approximation algorithms of $k$-median.
In fact, achieving a constant approximation factor for the top-$\ell$ minimization of $k$-median requires significant effort, as done in \cite{chakrabarty2018interpolating,byrka2018constant,chakrabarty2019approximation}.

\paragraph{Maximization Problems.}

The above approach does not directly work for the maximization problem.
We adopt the same idea in \cite{maalouly2022exact}.
The exact version of problem $\myproblem$ asks the question whether there is a feasible solution of with weight exactly equal to a given integer $K$. We say an algorithm runs in pseudopolynomial time for
the exact version of $\myproblem$ if the running time is 
polynomial in $n$ and $K$ (the algorithm can detect if there is no feasible solution with weight exactly $K$ and return ``no solution"). For many combinatorial problems, 
a pseudopolynomial algorithm for the exact version is known. Examples include shortest path, spanning tree, matching and knapsack.

\begin{theorem}
If the exact version of problem $\myproblem$ admits a pseudopolynomial time
algorithm, we can solve the top-$\ell$ {\em maximization} problem for 
$\myproblem$ if all weights are polynomially bounded nonnegative integers.
\end{theorem}
\begin{proof}
Suppose all weights are nonnegative and upper bounded by $B$.
Again, we define the modified value vector $\boldv^t$ as follows:
$$
v^t_e = (v_e-t+nB)\cdot \mathbf{1}\{v_e\geq t\}.
$$
Again, we guess the $\ell$-th largest weight $t$.
For each $i\in [0,B\ell]$,
we apply the pseudopolynomial time algorithm $\Algo$ to the same instance with the modified weight $\boldv^t$ and target $K=nB\ell+i$.
We return the feasible answer with the smallest $i$.
Obviously, if a feasible solution is found, it contains exactly $k$
elements with weight at least $t$.
It is also easy to see that the running time is $O(\poly(n, B))= \poly(n)$.
We show that if our guess $t$ is correct and 
$i=\boldv^t(S^\star)$, the solution $S$ found by $\Algo$
is the optimal solution.
We denote the optimal solution by $S^\star$.
The optimality of $S$ can be seen from the following chain of equations:
\begin{align*}
\topp{\ell}(\boldv[S]) & = \topp{\ell}(\boldv^t[S]) + t\ell -B\ell
= i +t\ell - B\ell
= \boldv^t(S) +t\ell - B\ell
= \boldv^t(S^\star) +t\ell - B\ell \\
& =\topp{\ell}(\boldv^t[S^\star])+t\ell - B\ell
  =\topp{\ell}(\boldv[S^\star]).
\end{align*}
This completes the proof.
\end{proof}

\section{Ordered Norm Optimization} \label{sec:ordered-norm}

In this section, we consider the ordered norm minimization problem.
Recall that $\ordered{\boldw}{\boldv}=\boldw^\top\boldv^\da$.
We use $\myproblem$ to denote the original combinatorial optimization problem
under consideration.
In particular, we assume that there is a poly-time algorithm $\Algo$ 
that solves the minimization problem $\myproblem$:
$
\min_{S\in \calF} \boldv(S), 
$
where we write $\boldv(S)=\sum_{e\in S}v_e$.

In this section, we fix $\delta,\varepsilon>0$.
Define the set $\POS=\{\min\{\ceil{(1+\delta)^s},n\}:s\geq 0\}$.
For each $\ell\in\POS$, we define $\text{next}(\ell)$ to be the smallest element in $\POS$ that is larger than $\ell$ for $\ell<n$.
And we define $\text{next}(n)=n+1$.
We sparsify the weight vector $\boldv$ so that $\widetilde{w}_i=w_\ell$ where $\ell$ is the smallest element in $\POS$ such that $\ell\ge i$.
Also, let $\widetilde{w}_{n+1}=0$.
For any $\boldt\in \R^{\POS}$, define
$$h_{\boldt}(\ww;a):=\sum_{\ell\in\POS}(\widetilde{w}_\ell-\widetilde{w}_{\text{next}(\ell)})(a-t_\ell)^+.$$
And
$$\text{prox}_{\boldt}(\ww;\boldv)=\sum_{\ell\in\POS}(\widetilde{w}_\ell-\widetilde{w}_{\text{next}(\ell)})\ell\cdot t_\ell+\sum_{i=1}^n h_{\boldt}(\ww;v_i).$$
Let $\boldsymbol{o}=\boldv[S^\star]$ be the vector for the optimal solution.

We need some lemmas from prior work \cite{chakrabarty2019approximation}.
In particular, they are Lemma 4.2, 6.9, 6.5 and 6.8 in \cite{chakrabarty2019approximation}.

\begin{lemma}
\label{lm:C.1}
For any $\boldv\in\Rpos^n$, 
$$\ordered{\widetilde{\boldw}}{\boldv}\leq \ordered{\boldw}{\boldv}\leq 
(1+\delta)\ordered{\widetilde{\boldw}}{\boldv}.$$
\end{lemma}

\begin{lemma}
Suppose we can obtain in polynomial time a set $Q$ of polynomial size that contains a real number in $[o_1^\da,(1+\varepsilon)o_1^\da]$.
Then, in time $O\left(|Q|\max\{(n/\varepsilon)^{O(1/\varepsilon)},n^{O(1/\delta)}\}\right)$, we can obtain a set $T$ of polynomial number of vectors in $\R^{\POS}$ which contains a threshold vector (denoted by $\boldsymbol{t}^*$) satisfying:
$o_\ell^\da\leq t^*_\ell\leq (1+\varepsilon)o_\ell^\da$ if $o_\ell^\da\geq \varepsilon o_1^\da/n$ and $t^*_\ell=0$ otherwise.
Moreover, for all $\ell\in \POS$,
$t^*_\ell$ is either 0 or a power of $1+\varepsilon$.
\end{lemma}

\begin{lemma}
\label{lm:C.3}
For any $\boldv\in\Rpos^n,\boldsymbol{t}\in \R^{\POS}$, 
the following inequality holds:
$$\ordered{\widetilde{\boldw}}{\boldv}\leq \text{\em prox}_{\boldt}(\widetilde{\boldw};\boldv).$$
\end{lemma}

\begin{lemma}
\label{lm:C.4}
Let $\boldsymbol{t}\in \R^{\POS}$ be a valid threshold vector
(i.e., $t_\ell \geq t_{\text{next}(\ell)}$ for all $\ell\in \POS$)
satisfying
$v_\ell^\da\leq t_\ell\leq (1+\varepsilon)v_\ell^\da$ for $v_\ell^\da\geq \varepsilon v_1^\da/n$, and $t_\ell=0$ otherwise.
Then, for any value vector $\boldv\in\Rpos^n$, we have that
$$
\text{\em prox}_{\boldsymbol{t}}(\widetilde{\boldw};\boldv)
\leq (1+2\varepsilon)\ordered{\widetilde{\boldw}}{\boldv}.$$
\end{lemma}

\begin{theorem}
If the original minimization problem $\myproblem$ can be solved in poly-time, 
there is a poly-time factor $(1+\epsilon)$ approximation algorithm for the ordered {\em minimization} problem, for any positive constant $\epsilon$.
\end{theorem}

\begin{proof}
Let $\boldsymbol{o}$ the value vector of the optimal solution.
We first guess the largest $v_e$ in $\boldsymbol{o}$ 
(polynomial number of choices because $e\in \calU$).
Using Lemma \ref{lm:polyguess}, we can obtain $T$ in poly-time, 
and we know $\boldt^*$ (satisfying the statement in Lemma \ref{lm:polyguess}) is in $T$.
For each vector $\boldsymbol{t}$ in $T$, we replace the value 
$v_e$ ($e\in \calU$) by 
$\hv_e = h_{\boldsymbol{t}}(\widetilde{\boldw};v_e)$.
Then we apply $\Algo$ to solve the original minimization problem
(i.e., minimize $\sum_{e\in S} \hv_e$). 
For each $\boldsymbol{t}$, we can obtain a solution. Among these solutions, 
we choose the solution $\widetilde{S}$ with $\widetilde{\boldv}=\boldv[\widetilde{S}]$ that minimizes $\ordered{\boldw}{\widetilde{\boldv}}$. Let the solution corresponding to $\boldsymbol{t}^*$ be $\boldv^*=\boldv[S']$. Then we have
$$\sum_{e\in S'} h_{\boldt^*}(\ww;v_e)\leq \sum_{e\in S^\star} h_{\boldt^*}(\ww;v_e)$$
and thus
$$\text{prox}_{\boldsymbol{t}^*}(\ww;\boldv^*)\leq \text{prox}_{\boldsymbol{t}^*}(\ww;\boldsymbol{o}).$$

By \cref{lm:C.1,lm:C.3,lm:C.4}, we have that
\begin{align*}
    \ordered{\boldw}{\widetilde{\boldv}}&\leq\ordered{\boldw}{\boldv^*}\leq (1+\delta)\ordered{\ww}{\boldv^*}\leq (1+\delta)\text{prox}_{\boldsymbol{t}^*}(\ww;\boldv^*)\\
    &\leq (1+\delta)\text{prox}_{\boldsymbol{t}^*}(\ww;\boldsymbol{o})\leq (1+\delta)(1+2\varepsilon)\ordered{\ww}{\boldsymbol{o}}\\
    &\leq (1+\delta)^2(1+2\varepsilon)\ordered{\boldw}{\boldsymbol{o}}.
\end{align*}
We can make $\delta,\varepsilon$ sufficiently small positive constants (e.g., $\epsilon/10$), so that the final approximation factor is at most $1+\epsilon$.
\end{proof}

We can generalize the above theorem to approximation algorithms.

\begin{theorem}
\label{thm:C3}
If there is a poly-time approximation algorithm for the minimization problem $\myproblem$ (with approximation factor $\alpha\geq 1$), 
there is a poly-time factor $\alpha(1+\epsilon)$ approximation algorithm for the ordered {\em minimization} problem optimally for any positive constant $\epsilon$.
\end{theorem}
\begin{proof}
We use the algorithm similar as in \cref{thm:C3}. We construct the same $T$ and run the algorithm $\Algo$ for each $\boldt\in T$ by changing each $v_e$ into $h_{\boldt}(\widetilde{\boldv};v_e)$. 
Similarly, we choose the solution $\widetilde{S}$ with $\widetilde{\boldv}=\boldv[\widetilde{S}]$ that minimizes $\ordered{\boldv}{\widetilde{\boldv}}$. Let the solution corresponding to $\boldsymbol{t}^*$ be $\boldv^*=\boldv[S']$. Thus we have
$$\text{prox}_{\boldsymbol{t}^*}(\ww;\boldv^*)\leq \alpha\cdot\text{prox}_{\boldsymbol{t}^*}(\ww;\boldsymbol{o}).$$
Therefore,
\begin{align*}
    \ordered{\boldv}{\widetilde{\boldv}}&\leq\ordered{\boldv}{\boldv^*}\leq (1+\delta)\ordered{\widetilde{\boldv}}{\boldv^*}\leq (1+\delta)\text{prox}_{\boldsymbol{t}^*}(\widetilde{\boldv};\boldv^*)\\
    &\leq \alpha(1+\delta)\text{prox}_{\boldsymbol{t}^*}(\widetilde{\boldv};\boldsymbol{o})\leq \alpha(1+\delta)(1+2\varepsilon)\ordered{\widetilde{\boldv}}{\boldsymbol{o}}\\
    &\leq \alpha(1+\delta)(1+2\varepsilon)\ordered{\boldv}{\boldsymbol{o}}.
\end{align*}
Setting $\delta,\varepsilon$ to be $\epsilon/10$, 
we can get the approximation ratio of
$\alpha(1+\epsilon)$.
\end{proof}

\section{Another form of \cref{thm:equivalence}}
\label{sec4:proof}
In Section~\ref{sec:matching}, we need to use the following theorem, whose proof is very similar to Theorem~\ref{thm:equivalence}.

\begin{theorem}
\label{thm:diffequiv}
Assume that we have two $\minnorm$ instances with the same $\calU$, the same value vector $\boldv$ and same norm $f$, but differ in their feasible sets, denoted by $\calF_1$ and $\calF_2$, where $\calF_1,\calF_2\neq \emptyset$.
Let $c$ be a positive integer and $\epsilon>0$. 
Consider the \LBO\ versions of these two problems.
Suppose that for any disjoint subsets $S_1,S_2,\cdots,S_T$ of $\calU$, if there exists a 1-valid solution in $\calF_1$, then we can compute a $c$-valid solution in $\calF_2$
(with the same $S_1,S_2,\cdots,S_T$) in polynomial time.
Under this assumption, we can find a set $S\in\calF_2$ such that $f(\boldv[S])\leq (4c+\epsilon)\cdot f(\boldv[S^*])$, where $S^*$ is the optimal solution in $\calF_1$ that minimizes the norm $f$ in polynomial time.
\end{theorem}

\begin{proofofthm}{\ref{thm:diffequiv}}
We enumerate all possible threshold vectors $\boldsymbol{t}\in \mathbb{R}^{\text{POS}}$ as defined in Lemma~\ref{lm:polyguess} (for $\calF_1$).
Suppose $\boldsymbol{t}$ is a valid guess for the optimal solution of $\calF_1$. We can also assume that we have an exactly correct guess of $o_1^\da$ ($o_1$ is also for $\calF_1$).
We construct sets $S_1,\cdots,S_T$ in the following way. 
For each element $e\in \calU$, if its value $v_e$ is larger than $\boldt_1$, we do not add it to any set. If it is at most $\max\{\boldt_n,\varepsilon o_1^\da/n\}$, we add it to $S_{T}$. For any other element, if it is at most $\boldt_{\ell}$ and larger than $\boldt_{\text{next}(\ell)}$,
where $\ell=2^i$, we add it to $S_{i+1}$, for $0\leq i\leq T-1$. 
Now, consider the optimal solution $\boldo$ for \minnorm for $\calF_1$ (if there are multiple optimal solutions, consider the one corresponding to the valid guess $\boldt$).
For $1\leq i\leq T-1$, consider how many elements in $\boldo$ that are added in $S_i$. By definition, elements in $S_i$ are with value larger than $\boldt_{2^i}$.
Also, the values are larger than $\varepsilon o_1^\da/n$. So for $\ell=2^i$, $o_\ell\leq \boldt_\ell$ by definition of valid guess. 
Thus, there are at most $2^i$ elements in $\boldo$ that are added in $S_i$. This also holds for $i=T$ as there are only $n$ total number of elements.

Assume that we can get a solution $S\in \calF_2$ such that there are at most $c\cdot 2^i$ elements in each $S_i$. We partition $S$ into $A,B$ such that the elements in $A$ have value less than $\varepsilon o_1/n$ and elements in $B$ have value at least $\varepsilon o_1/n$. 
We need this partition because of the condition 
$o_\ell^\da\geq \varepsilon o_1^\da/n$ in Lemma~\ref{lm:polyguess}.
Note that $A,B$ are not needed in the algorithm. They are only useful in the analysis.
So we have
$$
f(\boldv[S])\leq f(A)+f(B)\leq n\cdot f(\varepsilon o_1/n)+4cf(g(\boldt))\leq \varepsilon f(\boldo)+4c(1+\varepsilon)f(\boldo)\leq 4c(1+2\varepsilon)f(\boldo)
$$
where the first inequality following from the triangle inequality of the norm,
the second from the definition of $A$ and Lemma~\ref{lem:d-3} and the third from Lemma~\ref{lem:d-2}.
Therefore, the theorem follows.
\end{proofofthm}

%% file: appendix-6.tex
\section{Algorithms and Proofs Omitted in \cref{sec:intervalcover}}
\label{appendix:intcov}

\subsection{Proofs Omitted in \cref{sec:Inter-Cover-1}}
\label{sec:proofof6.1}
We call an input tuple $\inteta = (\intU; \intS{1}, \dots, \intS{T}; \Gamma)$ \emph{locally disjoint} if $\forall 1 \leq j \leq T$, $I, J \in \intS{j},I\neq J$, $I \cap J = \emptyset$. We say that $\inteta$ satisfies the \emph{laminar family} constraint if: 
\begin{itemize}
    \item $\intU$ forms a laminar family (i.e. each two intervals either have no common interior points or one contain the other), and \item $\forall I, J \in \intU$, if $I \subseteq J$, then the group index of $I$ is not less than that of $J$.
\end{itemize}
Now we focus on the \LBOintcov\ problem with input $\inteta = (\intU; \intS{1}, \intS{2}, \dots, \intS{T}; \Gamma)$. 
We want to prove that it is equivalent to \LBOtreecov\ 
up to a constant approximation factor 
(see Theorem~\ref{thm:inter-to-tree}).
We first change the input into $\ldeta$, which is a locally disjoint input.
Then we change it into $\lameta$, which is a laminar family input.
At last, we change it to the input of \LBOtreecov.

\input{FrameworkFig1}

First, We use the following process to construct a new \LBOintcov\ problem with input $\ldeta = (\ldU; \ldS{1}, \ldS{2}, \dots, \ldS{T}; \Gamma)$, where $\ldU = \bigcup_{1 \leq i \leq T} \ldS{i}$, and for all $1 \leq j \leq T$, $\ldS{j}$ is constructed by the following process (the process for different $j$ are independent):

We use two interval sets, $\intS{tmp}$ and $\ldS{tmp}$, as temporary variables for this conversion process. Initially, we set $\intS{tmp} \gets \intS{j}$ and $\ldS{tmp} \gets \emptyset$. Then we start a loop. At each iteration, we do the following steps:

\begin{enumerate}
\item If $\intS{tmp} = \emptyset$, we set $\ldS{j} = \ldS{tmp}$ and finish the loop.

\item Choose the interval $I$ in $\intS{tmp}$ whose left endpoint is the leftmost, and among those, select the one whose right endpoint is the rightmost.

\item Add the chosen interval $I$ to $\ldS{tmp}$. i.e., $\ldS{tmp} \gets \ldS{tmp} \cup \{I\}$.

\item Remove the intersection part from $\intS{tmp}$. i.e., $\intS{tmp} \gets \{J \setminus I : J \in \intS{tmp}, J \not\subseteq I\}$.
\end{enumerate}

\begin{lemma}\label{lem:inter-convert-1}
The \LBOintcov\ problem with input $\ldeta = (\ldU; \ldS{1}, \dots, \ldS{T}; \Gamma)$, constructed from $\inteta$ by the above process, satisfies the following properties:

\begin{enumerate}[label=(\alph*), format=\normalfont]

\item $\ldeta$ is locally disjoint. i.e., $\forall 1 \leq j \leq T$, $I, J \in \ldS{j}$, $I \cap J = \emptyset$.

\item For any $c \geq 1$, if there exists a $c$-valid solution for the problem with input $\ldeta$, then there exists a $c$-valid solution for the problem with input $\inteta$.

\item For any $c_0 \geq 1$, if there exists a $c_0$-valid solution for the problem with input $\inteta$, then there exists a $2c_0$-valid solution for the problem with input $\ldeta$.
\end{enumerate}
\end{lemma}

\begin{proof}

Statement (a) is straightforward. 

To demonstrate statement (b), consider any $1 \leq j \leq T$ and $I' \in \ldS{j}$. There exists a corresponding element $I \in \intS{j}$ such that $I' \subseteq I$. Thus, if we have a $c$-valid solution $\ldD_b$ for the input $\ldeta$, we can construct a solution $\intD_b$ by replacing each $I' \in \ldS{j}$ with the associated $I \in \intS{j}$. Then, $\intD_b$ is also $c$-valid.

To prove statement (c), for any $1 \leq j \leq T$ and $I \in \intS{j}$, whenever we add an interval $J$ to $\ldS{tmp}$ that intersects with $I$, $J$ either covers $I$, or in the next round, an interval that covers $J \setminus I$ is added. Therefore, $I$ intersects with at most two intervals in $\ldS{j}$. Assume we have a $c_0$-valid solution $\intD_c$ for the \LBOintcov\ with input $\inteta$. We can construct a solution $\ldD_c$ by replacing each $I \in \intD_c$ with the associated two intervals. Therefore, 
$$
|\ldD_c \cap \ldS{j}| \leq 2 |\intD_c \cap \intS{j}| \leq 2c_0 \cdot 2^j \text{ for all } 1 \leq j \leq T.
$$
\end{proof}

Next, we use the new input tuple $\ldeta$ to construct another \LBOintcov\ problem with input tuple $\lameta = (\lamU; \lamS{1}, \lamS{2}, \ldots, \lamS{T}; \Gamma)$, where $\lamU = \bigcup_{1 \leq j \leq T} \lamS{j}$, and $\lamS{1}, \lamS{2}, \ldots, \lamS{T}$ are constructed by the following process:

We use a series of interval sets $S'_1, S'_2, \dots, S'_T$ as temporary variables for this conversion process. Initially, we set $S'_j \gets \emptyset$ for all $1 \leq j \leq T$. 
Then, we begin a loop to sequentially enumerate elements from $\ldS{1}, \ldS{2}, \dots, \ldS{T}$. In the $i$th loop
for $1\leq i\leq T$, we sequentially consider all intervals in $\ldS{i}$. This means that when we enumerate $I \in \ldS{k+1}$, all intervals $J \in \ldS{j}$ for $j \leq k$ have already been enumerated. 
During this loop, we ensure that $\calU' = \bigcup_{1 \leq j \leq T} S'_j$ always forms a laminar family. When we reach $I \in \ldS{k+1}$, we perform the following steps:

\begin{enumerate}
\item For all $t \leq k$ and $J \in S'_t$, if $J \subseteq I$, we delete $J$ from $S'_t$ (i.e., $S'_t \gets S'_t \setminus \{J\}$). We say that interval $I$ \emph{removes} interval $J$ at this iteration.

\item We select intervals from $S'_1, S'_2, \dots, S'_k$ such that the selected interval $J$ has exactly one endpoint in $I$ (i.e., $J \cap I \neq \emptyset$, $I \setminus J \neq \emptyset$, and $J \setminus I \neq \emptyset$, also each of them doesn't only consist of one point). Denote this set as $M$, and split it into two parts: $M_L$ is the part where the left endpoint is in $I$, and $M_R$ is the part where the right endpoint is in $I$.

\item Choose the leftmost left endpoint in $M_L$. Denote $A$ as the interval from this endpoint to the right endpoint of $I$ (if $M_L = \emptyset$, then $A = \emptyset$). 
For each $j = 1, 2, \ldots, k$ and $J \in M \cap S'_j$, 
if $J \in M_L$, we add the interval $A$ (i.e., $S'_j \gets (S'_j \setminus \{J\}) \cup \{J \cup A\}$).
If $J \in M_R$, we add the interval $I\setminus A$. (i.e., $S'_j \gets (S'_j \setminus \{J\}) \cup \{J \cup (I\setminus A)\}$).
We say that interval $I$ \emph{extends} interval $J$ at this iteration.

\item By the fact that before this iteration $\bigcup_{1 \leq j \leq T} S'_j$ forms a laminar family, $A,B$ don't have common interior points. So we split $I$ into three parts: $A$, $B$, and $(I \setminus A) \setminus B$. Add the non-empty parts to $S'_{k+1}$.
\end{enumerate}

It is easy to prove that during the process, $\bigcup_{1 \leq j \leq T} S'_j$ is always a laminar family. So the process can finish.
After the loop finishes, we set $\lamS{j} \gets S'_j$ for all $1 \leq j \leq T$.

\begin{lemma}\label{lem:inter-convert-2}
The \LBOintcov\ problem with input $\lameta = (\lamU; \lamS{1}, \ldots, \lamS{T}; \Gamma)$, constructed from $\ldeta$ by the above process, satisfies the following properties:

\begin{enumerate}[label=(\alph*), format=\normalfont]
\item $\lameta$ satisfies the laminar family constraint. i.e., (i) $\lamU$ forms a laminar family, and (ii) $\forall I, J \in \lamU$, if $I \subseteq J,I\neq J$, then $\Lev(\lameta; I) < \Lev(\lameta; J)$.

\item For any $c \geq 1$, if there exists a $c$-valid solution for the problem with input $\inteta$, then there exists a $3c$-valid solution for $\ldeta$. 

\item For any $c_0 \geq 1$, if there exists a $c_0$-valid solution for the problem with input $\ldeta$, then there exists a $4c_0$-valid solution for $\inteta$.
\end{enumerate}
\end{lemma}
\begin{proof}
During the loop, we denote $\calU' = \bigcup_{1 \leq j \leq T} S'_j$, and $\eta' = (\calU'; S'_1, \ldots, S'_T; \Gamma)$ as temporary variables for the conversion process.

\emph{To prove statement (a)}, note that during the loop, we maintain $S_k'$ is always \emph{locally disjoint} for any $1\leq k\leq T$. Thus, based on our algorithm, it is easy to conclude statement (a).

\emph{To prove statement (b)}, assume we have a $c$-valid solution $\lamD_c$ for the input $\lameta$. We construct $\ldR_{I'}$ for each $I' \in \lamD_c$ and then construct the solution $\ldD_c$ as the union of all $\ldR_{I'}$. For each $I' \in \lamD_c$, we initialize $\ldR_{I'} \gets \emptyset$ and use a temporary variable $I_{tmp} \gets I'$ initially. We then simulate the previous loop in reverse order. For an iteration with $1 \leq k \leq T$ and interval $J \in \ldS{k}$, we consider the following cases:

\begin{enumerate}[label=(\alph*), format=\normalfont]
    \item If $J$ extends some interval $I'_{tmp}$ to $I_{tmp}$, we set $I_{tmp} \gets I'_{tmp}$ and add $J$ to $\ldR_{I'}$ (i.e., $\ldR_{I'} \gets \ldR_{I'} \cup \{J\}$). 
    Since $\ldeta$ is locally disjoint, this case can occur at most twice for each $k > j$.  
    \item If $I_{tmp}$ is one of the no more than two parts of $J$, then add $J$ to $\ldR_{I'}$.
\end{enumerate}

After this simulation finishes, $\ldR_{I'}$ must be a subset of $\ldU$ and cover $\Gamma$. We set 

$$\ldD_c = \bigcup_{I' \in \lamD_c} \ldR_{I'}.$$

Recall that $I' \in \lamS{j}$. Notice that the group indices of intervals in $\ldR_{I'}$ must be at least the group index of $I'$. Also, we can observe that $|\ldR_{I'} \cap \ldS{j}| = 1$, and for each $k > j$, $|\ldR_{I'} \cap \ldS{k}| \leq 2$. Thus, for each $k = 1, 2, \dots, T$, 

$$|\ldD_c \cap \ldS{k}| \leq \sum_{j \leq k} \sum_{I' \in \lamD_c \cap \lamS{j}} |\ldR_{I'} \cap \ldS{k}| \leq c \cdot 2^k + \sum_{j < k} c \cdot 2^j \cdot 2 \leq 3c \cdot 2^k.$$

This means that $\ldD_c$ is a $3c$-valid solution.

\emph{To prove statement (c)}, assume we have a $c_0$-valid solution $\ldD_b$ for $\ldeta$. For each $I \in \ldD_b$, we construct an interval set $\lamR_I \subseteq \lamU$ such that $I \subseteq \bigcup_{J \in \lamR_I} J$. 
We then construct $\lamD_b = \bigcup_{I \in \ldD_b} \lamR_I$. Clearly, $\lamD_b$ must cover $\Gamma$. Initially, we set $\lamR_I \gets \emptyset$ and then simulate the previous conversion process:    
\begin{enumerate}
\item When we enumerate $k$ and $I \in \ldS{k}$, we split $I$ into at most two parts and add these parts to $\lamR_I$.

\item For a sequential iteration with interval $J \in \ldU$, we check the changes for each element in $\lamR_I$:
\begin{itemize}
    \item If $J$ removes $I' \in \lamR_I$, $J$ is split into no more than two parts, and only one part contains $I'$ (otherwise $\calU'$ can't be a laminar family), so we replace $I'$ with this part of $J$.
    \item If $J$ extends $I' \in \lamR_I$, then we replace $I'$ with the interval after extension.
\end{itemize}
\end{enumerate}
\input{Example2}
When the simulation finishes, we set $\lamD_b \gets \bigcup_{I \in \lamU} \lamR_I$, and we observe that $|\lamR_I| \leq 2$. Also, the group indices of intervals in $\lamR_I$ is not smaller than $I$. Therefore, for all $1 \leq j \leq T$, we have:

$$|\lamD_b \cap \lamS{j}| \leq \sum_{k \leq j} \sum_{I \in \ldD_b \cap \ldS{k}} |\lamR_I \cap \lamS{j}| \leq \sum_{k \leq j} c_0 \cdot 2^k \cdot 2 \leq 4c_0 \cdot 2^j.$$

Thus, $\lamD_b$ is a $4c_0$-valid solution.

\end{proof}

Now we change the \LBOintcov\ with laminar family input into a \LBOtreecov\ problem.
For an interval $I\in \lamU$, $I\not =\bigcup_{J\in \lamU, J\subset I,J\not =I} J$. Since $\lamU$ is a laminar family, we must select at least one interval that covers $I$. So we can ignore other intervals contained in $I$.
We define
$\trM = \lamU\setminus \{I \in \lamU: \exists I', I\subseteq I', I' \neq \bigcup_{J \in \lamU, J \subset I', J\not =I'} J\}$.   
Then we construct a \LBOtreecov\ with input $\treta=(\trU; \trS{1}, \trS{2}, \ldots, \trS{T}; G)$, where $G=(V,E,r)$ is a rooted tree $G=(V,E,r)$ constructed from the laminar family $\trM$:

\begin{enumerate}
    \item Define the Vertex Set The vertex set V consists of $|\trM|+1$ nodes. This includes 
    \begin{enumerate}
        \item A root node $r$ 
        \item A node for each set in $\trM$.
    \end{enumerate}
    \item Establish the Parent-Child Relationships. The tree structure is determined by the hierarchical containment in $\trM$, since $\trM$ is laminar (i.e., sets are either disjoint or nested).
    \begin{enumerate}
    \item Choose the Root $r$: Define $r$ as a virtual node representing the universal set that contains all elements considered in $\trM$.
    \item Determine the Parent of Each Node: For each set interval $I\in \trM$, find the smallest interval $I'\in \trM\cup \lbrace \Gamma\rbrace$ (Recall that $\Gamma$ is the universal set) that strictly contains $I$. The node corresponding to $I$ is a child of the node corresponding to $I'$.
    \item Construct the Edge Set $E$: Add an edge from the parent node $I'$ to the child node $I$.
    \end{enumerate}
\end{enumerate}

In addition, we set $\trS{j}$ be the set of corresponding items in $\lamS{j}\cap \trM$ for $j=1,2,\ldots,T$, and $\trU=\bigcup_{1\le j\le T} \trS{j}$.
\input{Example3}

\begin{lemma}\label{lem:inter-convert-3}
The tuple $\treta = (\trU; \trS{1}, \trS{2}, \ldots, \trS{T}; G)$, constructed by the above process, satisfies the following properties:

\begin{enumerate}[label=(\alph*), format=\normalfont]
\item $\Lev(u) > \Lev(\Par(u))$ for all $u \in \trU$.

\item For any $c \geq 1$, there exists a $c$-valid solution for the \LBOintcov\ problem with input $\lameta$ if and only if there exists a $c$-valid solution for the \LBOtreecov\ with input $\treta$.
\end{enumerate}
\end{lemma}

\begin{proof}
- \cref{lem:inter-convert-2} (a) and the above process clearly imply statement (a).

- To establish statement (b), note that $\Leaf(G) \subseteq \trM$, and covering $\Leaf(G)$ is equivalent to covering $\Gamma$ for the \LBOintcov\ problem. Moreover, solutions can easily be converted between the two problems without any constant-factor loss.    
\end{proof}

\begin{proofofthm}{\ref{thm:inter-to-tree}}
This theorem follows directly from \cref{lem:inter-convert-1}, \cref{lem:inter-convert-2}, and \cref{lem:inter-convert-3}.
\end{proofofthm}

\subsection{Proofs Omitted in \cref{sec:Inter-Cover-2}}
\label{sec:proofof6.2}
For each $u \in \trU$, recall that the first type of cost $C_1(u) = \frac{1}{2^{\Lev(u)}}$ and the second type of cost

$$C_2(u)=\left\lbrace \begin{aligned}
&C_1(u) && & \text{if $u\in \Leaf(G)$}\\
&\min \lbrace C_1(u),\sum_{v\in \Ch(u)} C_2(v)\rbrace&& & \text{if $u\not \in \Leaf(G)$}\\
\end{aligned}\right.$$

If $C_2(u) \neq C_1(u)$, we can find a series of descendants of $u$ such that the sum of their $C_1$ values equals $C_2(u)$. We define the \textbf{replacement set} $R(u)$ of a node $u$ as follows:

$$R(u)=\left\lbrace \begin{aligned}
&\lbrace u\rbrace&& & C_2(u)=C_1(u)\\
&\cup_{v\in \Ch(u)} R(v)&& & C_2(u)\not =C_1(u)\\
\end{aligned}\right.$$

\begin{lemma}\label{lem:tree-1}
We consider a \LBOtreecov\ instance with input $\treta=(\trU;\trS{1},\ldots,\trS{T};G)$, $n=|\trU|$, $T=\ceil{\log n}$, and $T_0=\floor{\log\log\log n}$. Without loss of generality, we assume $n>10$. Then, for any constant $c\ge 1$, we have the following properties:
\begin{enumerate}[label=(\alph*), format=\normalfont]
\item For any $D\subseteq \trU$, if $\sum_{u\in D} C_1(u)\leq c$, then $|D \cap \trS{i}| \leq c \times 2^i$ for all $1 \leq i \leq T$.

\item For any $c$-valid solution $D^{*}$, $$\sum_{u \in D^*} C_2(u) \leq \sum_{u \in D^*} C_1(u) \leq c \cdot T.$$

\item For any $c$-valid solution $D^{*}$, $$\sum_{u \in D^{*}, \Lev(u)\le T_0, C_2(u) \leq \frac{1}{\log n}} C_2(u) \leq c.$$

\item For any $u\in \trU$ with $\Lev(u)\le T_0$, if $C_2(u) \leq \frac{1}{\log n}$, then $C_2(u) \neq C_1(u)$. 

\item For any $u\in \trU$, if a set $R'\subseteq \trU$ covers all leaves in $\Des(u)$, then 
$$\sum_{v \in R'} C_2(v) \geq C_2(u) = \sum_{v \in R(u)} C_2(v).$$
\end{enumerate}
\end{lemma}
\begin{proof}

\begin{itemize}
\item To see statement (a), assume there exists some $i$ such that $|D \cap \trS{i}| > c \times 2^i$. Then, $\sum_{u \in D} C_1(u) > \frac{c \times 2^i}{2^i} = c$, which contradicts the assumption.

\item To see statement (b), the first inequality follows from $C_2(u) \leq C_1(u)$ for any $u\in \trU$. Then,
$$\sum_{u \in D^*} C_1(u) = \sum_{j = 1}^T \sum_{u \in \trS{j}} C_1(u) \leq \sum_{j = 1}^T \frac{1}{2^j} \times (c \times 2^j) = c \times T.$$

\item To prove statement (c), 
$$\sum_{u \in D^{*}, \Lev(u)\le T_0, C_2(u) \leq \frac{1}{\log n}} C_2(u)\le \sum_{1 \leq j \leq T_0} \frac{c \times 2^j}{\log n} \leq \frac{2c \log\log n}{\log n} \leq c.$$

\item To prove statement (d), $C_1(u) = \frac{1}{2^j} \geq \frac{1}{2^{T_0}} = \frac{1}{\log \log n} > \frac{1}{\log n}$. It follows that $C_2(u) \neq C_1(u)$

\item (e) follows directly from the definitions.
\end{itemize}
\end{proof}

\begin{algorithm}[H]
\caption{Partial Enumeration Algorithm for Tree cover Problem} \label{Algo-Enum-Tree}
\KwData{A \LBOtreecov\ problem with input $\treta=(\trU;\trS{1},\trS{2},\ldots,\trS{T};G)$, where $G=(V,E,r)$ is a rooted tree}
\KwResult{A partial solution set $X\subseteq 2^{\trS{1}}\times \cdots 2^{\trS{T_0}}$}
\SetKwFunction{FRecursion}{PartEnumTree}
\SetKwProg{Fn}{Function}{:}{}
\Fn{\FRecursion{$P$,$D$}}{
    \If{$\left(\sum_{v\in D} C_2(v)\right)+\left(\sum_{u\in P} C_2(u)\right)> 2c_0T$ \textbf{or} $\exists i\le T_0, |D\cap \trS{i}|>2c_0\cdot 2^{i}$}{
        \Return $\emptyset$
    }
    \If{$P=\emptyset$ \textbf{or} $\forall u\in P, \Lev(u)>T_0$}{
        \Return $\lbrace D\rbrace$\;
    }
    Select $u$ from $P$ with minimum $\Lev(u)$\;
    $X'\gets \emptyset$\;
    \If{$u\not \in Leaf(G)$}{
        $X'\gets X'\cup$ \FRecursion{$(P\setminus \lbrace u\rbrace)\cup \Ch(u),D$}\;
    }
    \If{$C_2(u)>\frac{1}{\log n}$}{
        $X'\gets X'\cup $\FRecursion{$P\setminus \lbrace u\rbrace,D\cup \lbrace u\rbrace$}\;
    }
    \Return $X'$\;
}
$X'\gets$ \FRecursion{$Ch(r)$,$\emptyset$}\;
\Return $X=\lbrace (D\cap \trS{1},D\cap \trS{2},\cdots, D\cap \trS{T_0}) : D\in X'\rbrace$
\end{algorithm}
We employ a depth-first search (DFS) strategy to explore most of the states in the search space. During the search process, we maintain two sets:
\begin{itemize}
    \item \( P\subseteq \trU \), representing the set of candidate elements that can still be explored, i.e., \( \Des(P) \) contains all uncovered leaves.
    \item \( D\subseteq \trU \), storing the elements that have already been selected as part of the partial solution.
\end{itemize}
Initially, $P=\Ch(r)$ is the child set of the root, and $D=\emptyset$. 
At each recursive step, we select \( u\in P \) with the smallest group index. The recursion proceeds by exploring two possibilities:
\begin{enumerate}
    \item Adding \( u \) to the partial solution, i.e., including \( u \) in \( D \) and continuing the search.
    \item Excluding \( u \) from the partial solution, i.e., replacing \( u \) with its child nodes while keeping \( D \) unchanged. (If \( u \) is a leaf, this option is not applicable.)
\end{enumerate}

The search terminates when \( (P, D) \) fails to satisfy at least one of the following conditions:
\begin{enumerate}
\item $\exists u\in P$, $\Lev(u)\le T_0$
\item \( \forall u\in D, \quad C_2(u) > \frac{1}{\log n} \)
\item \( \left(\sum_{v \in D} C_2(v)\right) + \left(\sum_{u \in P} C_2(u)\right) \le 2c_0 T \)
\item \( \forall 1\le i \leq T_0, \quad |D \cap \trS{i}| \le 2c_0 \cdot 2^i \)
\end{enumerate}

\begin{lemma}\label{lem:tree-6}
For any sets $D, P \subseteq \trU$ satisfying the following conditions:

\begin{itemize}
\item $D \cap \Des(P) = \emptyset$,

\item $\forall u, v \in P$, $\Des(u) \cap \Des(v) = \emptyset$,

\item $\forall u, v \in D$, $\Des(u) \cap \Des(v) = \emptyset$,

\end{itemize}
If $\left(\sum_{v \in D} C_2(v)\right) + \left(\sum_{u \in P} C_2(u)\right) > c_0 T$ for some constant $c_0 \geq 1$, then there is no $c_0$-valid solution $D'$ satisfying $D' \subseteq D \cup \Des(P)$.
\end{lemma}

\begin{proof}
Suppose $D' \subseteq D \cup \Des(P)$ is a $c_0$-valid solution. Then, to cover descendants of $D$, we must have $D\subseteq D'$. Thus

$$c_0\times T\ge \sum_{v\in D'} C_2(v)=\left(\sum_{v\in D\cap D'} C_2(v)\right)+\left(\sum_{u\in D'\cap \Des(P)}C_2(u)\right)\ge \left(\sum_{v\in D} C_2(v)\right)+\left(\sum_{u\in P}C_2(u)\right)$$

The first inequality follows from \cref{lem:tree-1} (a), and the last inequality follows from \cref{lem:tree-1} (e). These inequalities contradict the assumptions,  so $D'$ cannot be a $c_0$-valid solution.

\end{proof}

\begin{proofofthm}{\ref{thm:tree-Enum}}

We ignore all nodes $u \in \trS{j}$ $(j \leq T_0)$ with $C_2(u) \leq \frac{1}{\log n}$. By \ref{lem:tree-1} (d) and the definition of $R(u)$, $\forall v\in R(u)$, $v$ cannot be ignored, and $v$ isn't in the first $T_0$ groups. Assume $D^{*}$ is a $c_0$-valid solution that $\forall u,v\in D^{*},u\neq v$, $\Des(u)\cap \Des(v)=\emptyset$. If $D^{*}$ contains an ignored node $u$, we replace $D^{*}$ by $(D^{*}-\lbrace u\rbrace)\cup R(u)$. Repeating the process, we can get another set $D'$ from $D^{*}$, and $D'$ covers all leaves. 

Formally, denote $I=\lbrace u\in \trU: \Lev(u)\le T_0,C_2(u)\le \frac{1}{\log n}\rbrace$, then $D'=(D^{*}-I)\cup R(I\cap D^{*})$ for $R(I\cap D^{*})=\cup_{u\in (I\cap D^{*})}R(u)$. Also,

$$\sum_{v\in R(I\cap D^{*})} C_2(v)\le \sum_{u\in (I\cap D^{*})} \sum_{v\in R(u)} C_2(v)=\sum_{u\in (I\cap D^{*})} C_2(u)\le c$$

The last inequality is due to \cref{lem:tree-1} (c). Then for all $1\le j\le T$, 

$$|D'\cap \trS{j}|\le |(D^{*}-I)\cap \trS{j}|+|R(I\cap D^{*})\cap \trS{j}|\le (c_0\times 2^j)+(c_0\times 2^j)=2c_0\times 2^j$$

$|R(I\cap D^{*})\cap \trS{j}|\le c_0\cdot 2^j$ follows from \cref{lem:tree-1} (a). Therefore, $D'$ is a $2c_0$-valid solution. Without loss of generality, we assume there is a $2c_0$-valid solution $D''$ without ignored nodes, satisfying $\forall u, v \in D''$, $\Des(u) \cap \Des(v) = \emptyset$. Because of \cref{lem:tree-6}, the boundary cases won't remove the branches with $2c_0$-valid solutions. If we simulate the process of \texttt{PartEnumTree}, whenever we try to decide whether to choose $u$, there must exist a branch in which we choose $u$ if and only if $u \in D''$. This means $D''$ is an extended solution for a partial solution in the output set $X$. Therefore, \cref{Algo-Enum-Tree} must output a partial solution with a $2c_0$-valid solution when there exists a $c_0$-valid solution. Now we only need to discuss its complexity.

For each $1 \leq i \leq T_0$, we focus on the number of branches when we decide on the vertices in $\trS{i}$. By ignoring the nodes $u$ with $C_2(u) \leq \frac{1}{\log n}$, due to \cref{lem:tree-1} (b), the number of nodes we need to consider is not greater than:

$$\left(2c_0\times T\right)/\left(\frac{1}{\log n}\right)\le 2c_0\times T^2$$

We only need to select at most $2c_0 \times 2^i$ elements, so the total number of branches is $\leq (1 + 2c_0 \cdot T^2)^{2c_0 \cdot 2^i}$. Therefore, the total number of branches is no more than:

$$\prod_{i=1}^{T_0} \left( 2c_0T^2+1 \right)^{2c_0\times 2^i}=\exp\left( \sum_{i=1}^{T_0} (2c_0\cdot 2^{i})\ln(1+2c_0T^2)\right)$$

$$=\exp\left(O(2^{T_0}\log\log n)\right)=\exp\left(O(\log\log^2 n)\right)=O(n)$$

Thus the algorithm can be done in polynomial time.

\end{proofofthm}

\subsection{Proofs Omitted in \cref{sec:Inter-Cover-3}}

For a given partial solution $(D_1, D_2, \dots, D_{T_0})$, we remove the vertices that have already been covered. Let $V_0=\left(\bigcup_{i=T_0+1}^{T} \trS{i}\right) \setminus \left(\bigcup_{i=1}^{T_0} \Des(D_i)\right)$ be the set of remaining nodes under this partial solution, and let $L_0=V_0\cap \Leaf(G)$ be the leaf set of $V_0$. In this subsection, we denote $T_1=\floor{\log\log n}$. We expect the rounding algorithm (\cref{Algo-Round-Tree}) to either find an integral solution in LP-Tree-Cover($4c_0+1$, $V_0$, $L_0$) or confirm that there is no integral solution in LP-Tree-Cover($2c_0$, $V_0$, $L_0$). Recall that we design the complete algorithm as follows:

\begin{enumerate}
\item Check if LP-Tree-Cover($2c_0$, $V_0$, $L_0$) has a feasible solution. If so, obtain an extreme point $x^{*}$. Otherwise, confirm that there is no such integral solution.

\item Remove the leaves $u$ with $x^{*}_u = 0$, and delete all the descendants of their parents. Then $\Par(u)$ becomes a leaf. Repeat this process until $x^{*}_u\neq 0$ for each leaf $u$. Let the modified node set and leaf set be $V_1$ and $L_1$, respectively.

\item For $u \in V_1$, attempt to round $x^{*}_u$. If $x^{*}_u \geq 1/2$, round it to $1$. If $x^{*}_u > 0$, and $u$ is not a leaf in $\trS{T_1+1} \cup \cdots \cup \trS{T}$, also round it to $1$. In all other cases, round $x^{*}_u$ to $0$. Let $D'$ be the set of nodes $u$ for which $x^{*}_u$ was rounded to $1$. Note that $D'$ may not cover $L_1$.

\item Remove all descendants in $D'$, and attempt to choose another set from $\trS{T_0+1} \cup \cdots \cup \trS{T_1}$ to cover all leaves. Formalize this objective as LP-Tree-Cover. Specifically, 
\begin{itemize}
\item $V_2 = (V_1 \setminus \Des(D')) \cap \{u\in \trU: T_0+1\le \Lev(u)\le T_1\}$, and

\item $L_2 = (V_2\cap L_1)\cup \{u \in V_2 :\exists v \in \Ch(u) \cap (V_1 \setminus V_2), (V_2 \setminus \Des(D')) \cap \Des(v) \cap L_1 \neq \emptyset\}$. 
\end{itemize}
To understand this, observe that \( V_2 \) consists of the nodes in the \((T_0+1)\)th to \( T \)th groups that remain uncovered. The set \( L_2 \) includes nodes in \( V_2 \) that are either leaves or have at least one uncovered child with a group index greater than \( T_0 \) (i.e., at least one descendant leaf remains uncovered).

Then solve LP-Tree-Cover($2c_0$, $V_2$, $L_2$). The fact that this problem must have feasible solutions is proved later, so we do not need to consider the case of no solution.

\item Let $x^{**}$ be an extreme point of LP-Tree-Cover($2c_0$, $V_2$, $L_2$). For each $u \in V_2$, round it to $1$ if and only if $x^{**}_u > 0$. Let $D'' = \{u \in V_2 : x^{**}_u > 0\}$, then $D''$ covers $L_2$.

\item Combine the three parts of the solution. That is, return $\left(\bigcup_{i=1}^{T_0} D_i\right) \cup D' \cup D''$.
\end{enumerate}

\begin{algorithm}[H]
    \caption{Rounding Algorithm for Interval Cover Problem} \label{Algo-Round-Tree}
    \KwData{A $(c,c_0)$-\LBOtreecov\ problem with input $\treta=(\trU;\trS{1},\ldots,\trS{T};G)$ and a partial solution $(D_1,D_2,\cdots D_{T_0})$. }
    \KwResult{A Solution $D$ or "No Solution"}
    Set $V_0\gets \left(\bigcup_{i=T_0+1}^{T} \trS{i}\right)\setminus \left( \bigcup_{i=1}^{T_0} \Des(D_i)\right)), L_0\gets V_0\cap \Leaf(G)$\;
    \If{LP-Tree-Cover($2c_0$,$V_0,L_0$) has no solution}{
        \Return "No Solution"\;
    }
    Solve LP-Tree-Cover($2c_0$,$V_0,L_0$) and obtain an extreme point $x^{*}$\;
    $L_1\gets L_0,V_1\gets V_0$\;
    \While{$\exists u\in L_1,x^{*}_u=0$}{
        $L_1\gets (L_1\setminus \Des(\Ch(\Par(u))))\cup \lbrace \Par(u)\rbrace$\;
        $V_1\gets V_1\setminus \Des(\Ch(\Par(u)))$
    }
    Set $D'\gets \lbrace v\in V_1\cap (\bigcup_{i=T_0+1}^{T} \trS{i}): x^{*}_v\ge 1/2 \rbrace \cup \lbrace v\in (V_1\setminus L_1)\cap (\bigcup_{i=T_1+1}^{T} \trS{i}) : x^{*}_v>0 \rbrace$\;
    $V_2\gets (V_1\setminus \Des(D')) \cap \left( \bigcup_{i=T_0+1}^{T_1} \trS{i}\right)$\;
    $L_2\gets \lbrace u\in V_2: u\in L_1 \text{ or } \exists v\in \Ch(u)\cap (V_1\setminus V_2), (V_2\setminus \Des(D'))\cap \Des(v)\cap L_1\not =\emptyset\rbrace$\;
    Solve LP-Tree-Cover($4c_0$,$V_2$,$L_2$), and obtain basic feasible solution $x^{**}$\;
    $D''=\lbrace u\in V_2 : x^{**}_{u}>0\rbrace$\;
    \Return $\left(\bigcup_{i=1}^{T_0} D_i\right)\cup D'\cup D''$
\end{algorithm}

\begin{lemma}\label{lem:tree-round-1}
In \cref{Algo-Round-Tree}, $x^{*}\assigned{V_1}$ is an extreme point of LP-Tree-Cover($2c_0$, $V_1$, $L_1$).
\end{lemma}
\begin{proof}
We know $x^{*}$ is an extreme point of LP-Tree-Cover($2c_0$,$V_0$, $L_0$). Consider the loop in lines 6–8. If $x^{*}_u = 0$ for some $u \in L_1$, then $\sum_{v \in \Anc(\Par(u))} x^{*}_v = 1 - x^{*}_u = 1$, and $\forall v \in \Ch(\Par(u)), x^{*}_v = 0$. LP-Tree-Cover($2c_0$,$V_1$,$L_1$) just assigns some variables to $0$ from LP-Tree-Cover($2c_0$,$V_0$,$V_1$). 
Therefore, by \cref{lem:LinearProg-2}, $x^{*}\assigned{V_1}$ is also an extreme point of LP-Tree-Cover($2c_0$, $V_1$, $L_1$).
\end{proof}

\begin{lemma}\label{lem:tree-round-2}
In \cref{Algo-Round-Tree}, if there exists a solution $x^{*}$ for LP-Tree-Cover($2c_0$, $V_0$, $L_0$), then the fractional solution set of LP-Tree-Cover($4c_0$, $V_2$, $L_2$) is non-empty.
\end{lemma}
\begin{proof}
First, consider $\overline{x} = 2 x^{*}\assigned{V_2}$. Clearly, $\overline{x}$ satisfies the cardinality constraints: $\sum_{v \in \trS{i}} \overline{x}_v \leq 4c_0 \cdot 2^i$ for all $T_0 + 1 \leq i \leq T$. We now show that $\sum_{v \in \Anc(u)} \overline{x}_v \geq 1$ for all $u \in L_2$. We consider the cases according to the construction of $L_2$:
\begin{itemize}
\item If $u \in L_1$, LP-Tree-Cover($2c_0$, $V_1$, $L_1$) has the constraint $\sum_{v \in \Anc(u)} x^{*}_v = 1$. So $\overline{x}$ satisfies this inequality.

\item If $u \notin L_1$ and $\exists v \in \Ch(u) \cap (V_1 \setminus V_2)$ such that $(V_2 \setminus \Des(D')) \cap \Des(v) \cap L_1 \neq \emptyset$, then choose $w \in (V_2 \setminus \Des(D')) \cap \Des(v) \cap L_1$. Since $w \notin \Des(D')$, we have $\sum_{z \in \Anc(\Par(w)) \setminus \Anc(u)} x^{*}_z = 0$. Therefore,

$$\sum_{z \in \Anc(u)} \overline{x}_z = \left(\sum_{z \in \Anc(w)} \overline{x}_z\right) - \left(\sum_{z \in \Anc(w) \setminus \Anc(u)} \overline{x}_z\right) \geq 2(1 - x^{*}_w) \geq 1.$$
\end{itemize}
Combining these results, we have:

\begin{equation}
\begin{aligned}
\sum_{v \in \Anc(u) \cap V_2} \overline{x}_v \geq 1 && & \forall u \in L_2 \\
\sum_{v \in \trS{i} \cap V} \overline{x}_v \leq 4c_0 \cdot 2^i && & \forall T_0 + 1 \leq i \leq T \\
\overline{x}_v \geq 0 && & \forall v \in V_2
\end{aligned}
\end{equation}

Next, we adjust $\overline{x}$ to satisfy the constraints $\sum_{v \in \Anc(u) \cap V_2} \overline{x}_v = 1$ for all $u \in L_2$. We perform the following steps iteratively:
\begin{enumerate}
\item Choose $u \in L_2$ with the maximum $\overline{x}_u$. If $\overline{x}_u = 1$, then terminate the loop.

\item Find the ancestor $v \in \Anc(u) \cap V_2$ with the highest group index such that $\overline{x}_v > 0$.

\item Decrease $\overline{x}_v$ until either the constraint for $u$ is satisfied or $\overline{x}_v = 0$, i.e., set $\overline{x}_v \gets \overline{x}_v - \min \{\overline{x}_v, \overline{x}_u - 1\}$. Then return to step 1.
\end{enumerate}

In each iteration, either a non-zero variable is set to zero or a slack constraint becomes tight. Therefore, this loop terminates. For an iteration with $u \in L_2$ and its ancestor $v$, the adjustments affects the constraints for $w \in \Des(v) \cap L_2$. Since $v$ is the ancestor with the highest group index, we have $\sum_{z \in \Anc(w) \cap V_2} \overline{x}_z \geq \sum_{z \in \Anc(v) \cap V_2} \overline{x}_z = \sum_{z \in \Anc(u) \cap V_2} \overline{x}_z$ throughout the process. Therefore, after the loop terminates, $\sum_{v \in \Anc(u) \cap V_2} \overline{x}_v = 1$ for all $u \in L_2$. Thus, after these operations, $\overline{x}$ is a feasible solution for LP-Tree-Cover($4c_0$, $V_2$, $L_2$), and this linear program has a non-empty solution set.
\end{proof}

\begin{proofoflem}{\ref{lem:tree-round-3}}
According to \cref{Algo-Round-Tree}, it is clear that $D' \cup D''$ covers all leaves in $V_0$, i.e., $L_0 \subseteq \Des(D') \cup \Des(D'')$. Therefore, $D = \left(\bigcup_{i=1}^{T_0} D_i\right) \cup D' \cup D''$ covers all leaves in the original tree.

If the input partial solution has a $2c_0$-valid extended solution, by simply deleting the nodes with chosen ancestors in this solution, it can be a $2c_0$-valid solution for the \ref{LP-Tree-Cover}. So LP-Tree-Cover($2c_0$, $V_0$, $L_0$) has feasible solutions. In the linear program LP-Tree-Cover($2c_0$,$L_1$, $V_1$), $x^{*}_u > 0$ for all $u \in L_1$. The number of feasibility constraints, $\sum_{v \in \Anc(u)} x_v = 1$, equals the number of leaves. According to  \cref{lem:extreme:points}, the number of non-zero, non-leaf variables is no more than the number of cardinality constraints. The number of cardinality constraints is $T - T_0 \leq \log n$, so we obtain the following inequality:

$$|\{u \in V_1 \setminus L_1 : x^{*}_u > 0\}| \leq \log n.$$  

For $T_1 + 1 \leq i \leq T$, we choose node $u \in V_1$ if and only if $x^{*}_u \geq 1/2 \lor (x^{*}_u > 0 \land u \in L)$. Then,

$$|D \cap \trS{i}| \leq \left( 2 \sum_{u \in V_1 \cap \trS{i}} x^{*}_u \right) + \log n \leq 2 \cdot 2c_0 \cdot 2^i + 2^{T_1 + 1} \leq (4c_0 + 1) \cdot 2^i.$$

For $T_0 + 1 \leq i \leq T_1$, by \cref{lem:tree-round-2}, we can solve LP-Tree-Cover($4c_0$,$V_2$, $L_2$) and obtain an extreme point $x^{**}$. We perform the same steps as when deriving $(V_1, L_1)$ from $(V_0, L_0)$. We now derive $(V_3, L_3)$ from $(V_2, L_2)$ so that $x^{**}$ corresponds to an extreme point of LP-Tree-Cover($4c_0$, $V_3$, $L_3$), and $x^{**}_u > 0$ for all $u \in L_3$. Then the number of non-zero, non-leaf nodes is no more than the number of cardinality constraints, and here the number of effective cardinality constraints is no more than $T_1 - T_0\le \log\log n$. Therefore, for $T_0 + 1 \leq i \leq T$,

$$|D \cap \trS{i}| \leq 2 \times 2c_0 \times 2^i + \log\log n \leq 4c_0 \times 2^i + 2^{T_0+1} \leq (4c_0 + 1) \times 2^i.$$

Combining these results, we have $|D \cap \trS{i}| \leq (4c_0 + 1) \times 2^i$ for all $i = 1, 2, \dots, T$, and $D$ is a $(4c_0 + 1)$-valid solution.
\end{proofoflem}
\begin{proofofthm}{\ref{thm:tree-solve}}
We first run \cref{Algo-Enum-Tree} to obtain a partial solution set $X$, and then run \cref{Algo-Round-Tree} for each partial solution in $X$. If \cref{Algo-Round-Tree} returns a $(4c_0 + 1)$-valid solution, then we output it directly. Otherwise, if \cref{Algo-Round-Tree} always returns "No Solution," then we confirm that there is no $c_0$-valid solution.

According to \cref{thm:tree-Enum}, \cref{Algo-Enum-Tree} runs in polynomial time, and clearly, \cref{Algo-Round-Tree} can be done in polynomial time. Therefore, the whole algorithm can be done in polynomial time.

If there exists a $c_0$-valid solution, then \cref{Algo-Enum-Tree} outputs a partial solution with a $2c_0$-valid extended solution, and by \cref{lem:tree-round-3}, \cref{Algo-Round-Tree} returns a $(4c_0 + 1)$-valid solution based on this partial solution. Therefore, this algorithm can solve the $(4c_0 + 1, c_0)$-\LBOtreecov.
\end{proofofthm}
\begin{proofofthm}{\ref{thm:inter-cover}}
Here we combine the conclusions of \cref{thm:inter-to-tree} and \cref{thm:tree-solve}. We first convert the interval cover problem to a tree cover problem. If the original problem has a $c_0$-valid solution, then the tree cover problem has a $8c_0$-valid solution. Then, we solve the $(4 \cdot 8c_0 + 1, 12c_0)$-\LBOtreecov\ to get a $(32c_0 + 1)$-valid solution. Finally, we convert this to an interval cover solution, which is $3(32c_0 + 1)$-valid.
\end{proofofthm}

%% file: FrameworkFig1.tex
\begin{figure}[t]
\begin{minipage}{\textwidth}
    \centering
\begin{tikzpicture}[
box/.style={rectangle, draw, minimum height=1cm, minimum width=3cm, align=center}
]
\node[box] (box1) {A \LBOintcov\ \\ Problem};

\node[box, right=of box1] (box2) {A \LBOintcov\ \\ Problem with \\ Locally Disjoint \\ Input};  
  
\node[box, right=of box2] (box3) { A \LBOintcov\ \\ Problem with \\ Laminar Family \\ Input};  
  
\node[box, right=of box3] (box4) { A \LBOtreecov\ \\Problem};  
  
\draw[->,>={Stealth[width=4pt, length=6pt]}] (box1) -- (box2);  
\draw[->,>={Stealth[width=4pt, length=6pt]}] (box2) -- (box3);  
\draw[->,>={Stealth[width=4pt, length=6pt]}] (box3) -- (box4);

\end{tikzpicture}
\caption{Outline of \cref{sec:proofof6.1}.}
\end{minipage}
\end{figure}

%% file: Example2.tex
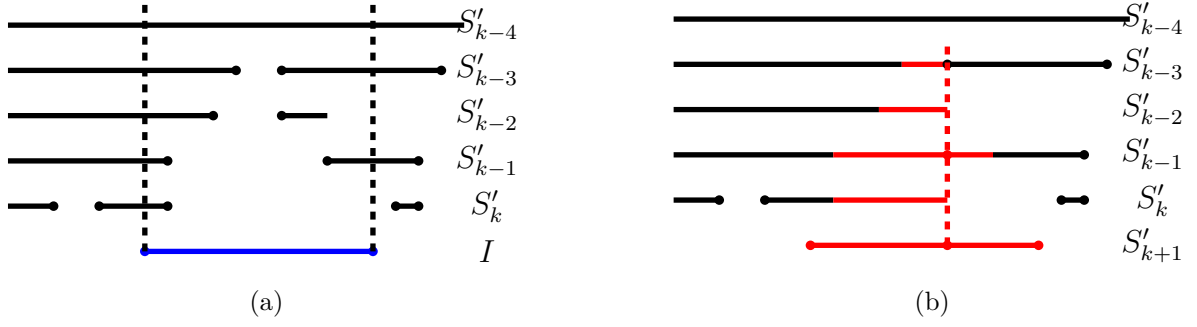
\begin{figure}[t] 
\begin{minipage}{\textwidth}
    \centering  
    \begin{subfigure}{0.45\textwidth}  
        \centering  
\begin{tikzpicture}[scale=0.6]  

\fill[blue] (3,0) circle (3pt);
\fill[blue] (8,0) circle (3pt); 
\draw[thickline,color=blue] (3,0) -- (8,0); 
\draw[thickline,color=black] (0,1) -- (1,1); 
\fill[black] (1,1) circle (3pt);
\draw[thickline,color=black] (8.5,1) -- (9,1);
\fill[black] (8.5,1) circle (3pt);
\fill[black] (9,1) circle (3pt);

\draw[thickline,color=black] (2,1) -- (3.5,1);
\fill[black] (2,1) circle (3pt);
\fill[black] (3.5,1) circle (3pt);

\draw[thickline,color=black] (0,2) -- (3.5,2);
\fill[black] (3.5,2) circle (3pt);
\draw[thickline,color=black] (6,3) -- (7,3);
\fill[black] (6,3) circle (3pt);
\draw[thickline,color=black] (7,2) -- (9,2);
\fill[black] (7,2) circle (3pt);
\fill[black] (9,2) circle (3pt);

\draw[thickline,color=black] (0,3) -- (4.5,3);
\fill[black] (4.5,3) circle (3pt);

\draw[thickline,color=black] (0,4) -- (5,4);
\fill[black] (5,4) circle (3pt);

\draw[thickline,color=black] (6,4) -- (9.5,4);
\fill[black] (6,4) circle (3pt);
\fill[black] (9.5,4) circle (3pt);
\draw[thickline,color=black] (0,5) -- (10,5);

\draw[thickline,dashed] (3,0) -- (3,5.5);
\draw[thickline,dashed] (8,0) -- (8,5.5);
\node at (10.5,0) {\large $I$};
\node at (10.5,1) {\large $S'_{k}$};
\node at (10.5,2) {\large $S'_{k-1}$};
\node at (10.5,3) {\large $S'_{k-2}$};
\node at (10.5,4) {\large $S'_{k-3}$};
\node at (10.5,5) {\large $S'_{k-4}$};
\end{tikzpicture}  

\caption{}  
    \end{subfigure}  
    \hfill 
    \begin{subfigure}{0.45\textwidth}  
        \centering  
\begin{tikzpicture}[scale=0.6]  

\fill[red] (3,0) circle (3pt);
\fill[red] (8,0) circle (3pt); 
\draw[thickline,color=red] (3,0) -- (6,0); 
\draw[thickline,color=red] (6,0) -- (8,0); 
\draw[thickline,color=black] (0,1) -- (1,1); 
\fill[black] (1,1) circle (3pt);
\draw[thickline,color=black] (8.5,1) -- (9,1);
\fill[black] (8.5,1) circle (3pt);
\fill[black] (9,1) circle (3pt);
\draw[thickline,color=black] (2,1) -- (3.5,1);
\fill[black] (2,1) circle (3pt);
\draw[thickline,color=red] (3.5,1)--(6,1);

\draw[thickline,color=black] (0,2) -- (3.5,2);
\draw[thickline,color=red] (3.5,2)--(6,2);
\fill[red] (6,2) circle (3pt);
\draw[thickline,color=black] (7,2) -- (9,2);
\draw[thickline,color=red] (6,2) -- (7,2);
\fill[black] (9,2) circle (3pt);

\draw[thickline,color=black] (0,3) -- (4.5,3);
\draw[thickline,color=red] (4.5,3)-- (6,3);

\draw[thickline,color=black] (0,4) -- (5,4);
\draw[thickline,color=red] (5,4) -- (6,4);

\draw[thickline,color=black] (6,4) -- (9.5,4);
\fill[black] (6,4) circle (3pt);
\fill[black] (9.5,4) circle(3pt);

\draw[thickline,color=black] (0,5) -- (10,5);

\node at (10.5,0) {\large $S'_{k+1}$};
\node at (10.5,1) {\large $S'_{k}$};
\node at (10.5,2) {\large $S'_{k-1}$};
\node at (10.5,3) {\large $S'_{k-2}$};
\node at (10.5,4) {\large $S'_{k-3}$};
\node at (10.5,5) {\large $S'_{k-4}$};

\fill[red] (6,0) circle (3pt);
\draw[thickline,dashed,color=red] (6,0) -- (6,4.5);
\end{tikzpicture}  
        \caption{}  
    \end{subfigure}  
    \caption{An example of adding a new interval. The blue interval represents $I \in \ldS{k+1}$. We remove some intervals, extend some intervals, and split $I$ into at most two parts. }  
    \label{fig:seg-2}

\end{minipage}
\end{figure}  

%% file: Example3.tex
 \begin{figure}[t]
\begin{minipage}{\textwidth}
  
    \centering  
    \begin{subfigure}[t]{0.4\textwidth}  
        \centering  
\begin{tikzpicture}[scale=0.5]

\draw[thickline,color=red] (1,4) -- (6,4);
\draw[thickline,color=black] (6,4) -- (10,4);

\draw[thickline,color=black] (1,3) -- (5,3); 
\draw[thickline,color=red] (5,3) -- (9,3); 
\draw[thickline,color=red] (9,2) -- (10,2);
\draw[thickline,color=black] (1,2) -- (5,2);
\draw[thickline,color=black] (6,2) -- (7,2);

\draw[thickline,color=black] (3,1) -- (4,1);
\draw[thickline,color=red] (8,1) -- (9,1);
\fill[red] (1,4) circle (3pt);
\fill[black] (10,4) circle (3pt);
\fill[black] (6,4) circle (3pt); 
\fill[black] (1,3) circle (3pt);
\fill[black] (5,3) circle (3pt);
\fill[red] (6,3) circle (3pt);
\fill[red] (9,3) circle (3pt);
\fill[black] (1,2) circle (3pt);
\fill[black] (2,2) circle (3pt);
\fill[black] (5,2) circle (3pt);
\fill[black] (6,2) circle (3pt);
\fill[black] (7,2) circle (3pt);
\fill[red] (10,2) circle (3pt);
\fill[black] (3,1) circle (3pt);
\fill[black] (4,1) circle (3pt);
\fill[red] (8,1) circle (3pt);
\fill[red] (9,1) circle (3pt);
\node at (11,1) {\large $\lamS{4}$};
\node at (11,2) {\large $\lamS{3}$};
\node at (11,3) {\large $\lamS{2}$};
\node at (11,4) {\large $\lamS{1}$};
\end{tikzpicture}  

\caption{The Interval Cover Problem with A Laminar Family Input}  
    \end{subfigure}  
    \hfill 
    \begin{subfigure}[t]{0.45\textwidth}  
        \centering  
\begin{tikzpicture}[scale=0.5]  
\draw[thickline,color=red] (1,4) -- (6,4);
\draw[thickline,color=black] (6,4) -- (10,4);

\draw[thickline,color=black] (1,3) -- (5,3); 
\draw[thickline,color=red] (5,3) -- (9,3); 
\draw[thickline,color=black] (1,2) -- (5,2);

\fill[red] (1,4) circle (3pt);
\fill[black] (10,4) circle (3pt);
\fill[black] (6,4) circle (3pt); 
\fill[black] (1,3) circle (3pt);
\fill[black] (5,3) circle (3pt);
\fill[red] (6,3) circle (3pt);
\fill[red] (9,3) circle (3pt);
\fill[black] (1,2) circle (3pt);
\fill[black] (2,2) circle (3pt);
\fill[black] (5,2) circle (3pt);


\node at (12,1) {\large $\lamS{4}\cap \trM$};
\node at (12,2) {\large $\lamS{3}\cap \trM$};
\node at (12,3) {\large $\lamS{2}\cap \trM$};
\node at (12,4) {\large $\lamS{1}\cap \trM$};
\end{tikzpicture}  
        \caption{Laminar Family $\trM$}  
    \end{subfigure}  
    \vspace{0.5cm} 
    \begin{subfigure}{0.45\textwidth}  
        \centering  
\begin{tikzpicture}[scale=0.5]  

\draw[thickline,color=black] (3,4) -- (2.5,3);
\draw[thickline,color=black] (3,4) -- (4,3);
\draw[thickline,color=black] (6.5,4) -- (6,3);
\draw[thickline,color=black] (6.5,4) -- (8,2);
\draw[thickline,color=black] (2.5,3) -- (2,2);
\draw[thickline,color=black] (2.5,3) -- (3.5,2);
\draw[thickline,color=black] (5,5) -- (3,4);
\draw[thickline,color=black] (5,5) -- (6.5,4);

\fill[red] (3,4) circle (5pt);
\fill[black] (6.5,4) circle (5pt);
\fill[black] (2.5,3) circle (5pt);
\fill[red] (4,3) circle (5pt);
\fill[red] (6,3) circle (5pt);
\fill[black] (2,2) circle (5pt);
\fill[black] (3.5,2) circle (5pt);
\fill[red] (8,2) circle (5pt);
\fill[black] (5,5) circle (5pt);
\node at (10,1) {\large $\trS{4}$};
\node at (10,2) {\large $\trS{3}$};
\node at (10,3) {\large $\trS{2}$};
\node at (10,4) {\large $\trS{1}$};
\node at (10,5) {\large $r$};
\end{tikzpicture}  
        \caption{The Corresponding Tree Cover Problem}  
    \end{subfigure}  
    \caption{An example of converting $\lameta$ to a \LBOtreecov\ problem. The red intervals or nodes represent a solution.}

\end{minipage}
\end{figure}
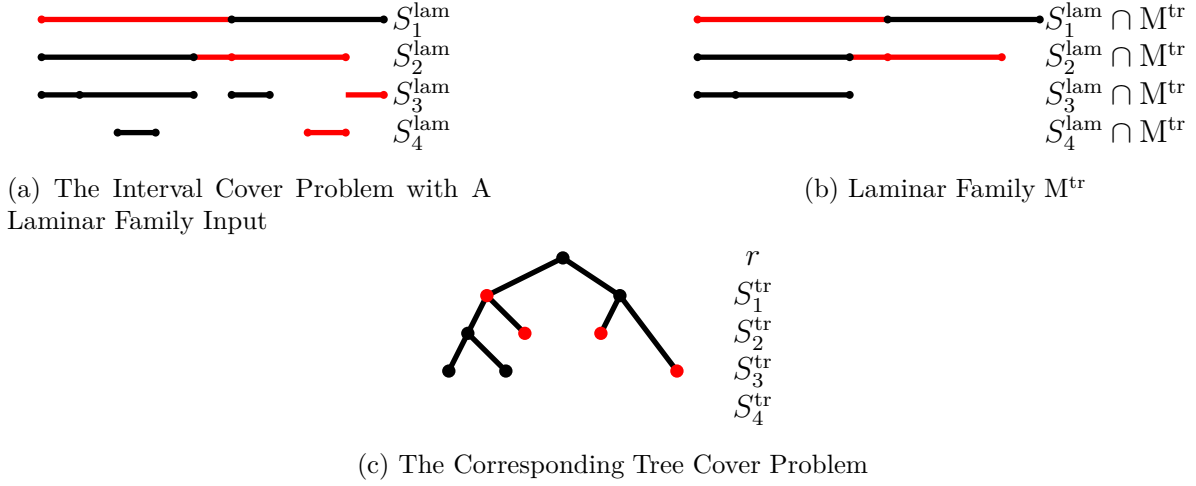 

%% file: appendix-11-setcover.tex
\section{An $O(\log n)$-Approximation for Min-norm Set Cover} \label{sec:set-cover}

We denote the norm minimization problem for set cover as \minnormsetcov.
We present a simple algorithm based on randomized rounding and show 
it is an $O(\log n)$-Approximation. Note this is optimal assuming $P\ne NP$.

By Theorem~\ref{thm:equivalence}, we only need to consider the \LBO\ version of set cover (\LBOsetcov\ ).
We use the following natural LP relaxation:

\begin{equation}
\begin{aligned}
\min && 0 && & \\
s.t. && \sum_{v\in B(u)} x_v\ge 1 && & \forall u\in A\\
     && x_v\ge 0 && & \forall v\in Q\\
	 && \sum_{v\in S_i} x_v\le 2^i && & \forall 1\le i\le T\\
\end{aligned}
\end{equation}

Here $A$ is the set of all elements, $Q$ is the set of all sets and $B(u)$ is the sets that contain $u$ for $u\in A$.

\begin{theorem}
\label{thm:setcover}
Assume that there is a 1-valid solution for \LBOsetcov. Then there is a randomized algorithm that can find an $O(\log n)$-valid solution for \LBOsetcov\ with probability larger than 1/2. 
\end{theorem}

\begin{proof}
As there is a 1-valid solution, there is a fractional solution for this LP. We consider the solution. For each $x_v$, it has probability $\min\{2x_v\log n, 1\}$ to be rounded to 1, and the remaining probability to be rounded to 0. So for each $u\in A$, we consider the probability that it is not covered. If there is a $v\in Q$ such that $2x_v\log n\geq 1$, it is surely covered. Otherwise, the probability that it is not covered is
$$\prod_{v\in B(u)}(1-2x_v\log n)\leq \prod_{v\in B(u)}e^{-2x_v\log n}\leq \frac{1}{n^2}.$$
So the probability that all elements are covered is at least $1-1/n$. For each $1\leq i\leq T$, we consider the probability $1_i$ that we select $2^{i+2}\log n$ elements in $S_i$. For each $e\in S_i$, consider $X_e$ with $p_i=\min\{2x_v\log n, 1\}$ indicating whether $x_e$ is rounded to 1. By Chernoff bounds (Lemma~\ref{lemma:chernoff}), when $\log n\geq 3$, $q_i$ is at most
$$q_i\leq 2^{-2^{i+1}\log n}\leq \frac{1}{n}.$$
So the total failure probability is at most
$$n\cdot\frac{1}{n^2}+\frac{T}{n}<1/2$$
for $n$ that is large enough.
\end{proof}
\newpage